\documentclass{lmcs} 
\pdfoutput=1

\usepackage{lastpage}
\lmcsdoi{16}{1}{6}
\lmcsheading{}{\pageref{LastPage}}{}{}%
{Mar.~25,~2019}{Jan.~27,~2020}{}

\usepackage[iso-8859-1]{inputenx}

\usepackage{stmaryrd}
\usepackage{amsthm}
\usepackage{amssymb}
\usepackage{amsmath}
\usepackage{bussproofs}
\usepackage{dsfont}
\usepackage{xypic}
\usepackage[all]{xy}
\usepackage{tikz-cd}
\usepackage{todonotes}
\usepackage{url}
\usepackage{proof}
\usepackage{comment}
\usepackage{mdframed}

\newtheorem{teorema}{Teorema}[section]

\newtheorem{theorem}[teorema]{Theorem}
\newtheorem{lemma}[teorema]{Lemma}
\newtheorem{proposition}[teorema]{Proposition}
\newtheorem{corollary}[teorema]{Corollary}

\theoremstyle{definition} 
\newtheorem{definition}[teorema]{Definition}

\theoremstyle{remark} 
\newtheorem{remark}[teorema]{Remark}
\newtheorem{example}[teorema]{Example}

\newcommand{\diam}[1]{ \langle #1\rangle }

\newcommand*\diff{\mathop{}\!\mathrm{d}}

\newcommand{\sem}[1] {  \llbracket #1 \rrbracket  }

\newcommand{\CHaus}{\cat{CHaus}}
\newcommand{\blank}{\ensuremath{{\mbox{-}}}}
\newcommand{\op}{\ensuremath{^{\mathrm{op}}}}
\newcommand{\R}{\ensuremath{\mathbb{R}}}

\newcommand{\Rdn}{\ensuremath{\mathcal{R}}}
\newcommand{\id}{\ensuremath{\mathrm{id}}}
\newcommand{\Rdnl}{{\Rdn^{\leq 1}}}
\newcommand{\Id}{\ensuremath{\mathrm{Id}}}
\newcommand{\ie}{\emph{i.e.}}

\newcommand{\Spec}{\ensuremath{\mathrm{Spec}}}

\newcommand{\cat}[1]{\ensuremath{\mathbf{#1}}}
\newcommand{\CoAlg}{\cat{CoAlg}}
\newcommand{\Markov}{\cat{Markov}}
\newcommand{\CRiesz}{\cat{Riesz}}

\newcommand{\CAURiesz}{\cat{CAURiesz}}
\newcommand{\CAURieszD}{\CAURiesz_\Diamond}

\newcommand{\URiesz}{\cat{URiesz}}
\newcommand{\AURiesz}{\cat{AURiesz}}
\newcommand{\AURieszD}{\AURiesz_\Diamond}

\newcommand{\E}{\ensuremath{\mathbb{E}}}
\newcommand{\DAURiesz}{\cat{DAURiesz}}

\DeclareMathAlphabet{\mathbbold}{U}{bbold}{m}{n}
\newcommand{\one}{\ensuremath{\mathbbold{1}}}
\newcommand{\zero}{\ensuremath{\mathbbold{0}}}
\newcommand{\Formulas}{\texttt{Form}}

\newcommand{\dia}[1]{\ensuremath{\Diamond\!_{#1}}}
\newcommand{\hull}{\ensuremath{\mathrm{hull}}}

\newcommand{\RUB}[0]{ \R_{\{\B\}} }

\newcommand\B[0]{
	\textnormal{P}
}

\newcommand\nodeC[1]{*+[o][F]{#1}}

\newcommand\addLabelUL[1]{\ar@{}[]+UR|(1){~\makebox[0pt][l]{$\mathbf{#1}$}}}

\newcommand\addLabelUR[1]{\ar@{}[]+UR|(1){~\makebox[0pt][l]{$\mathbf{#1}$}}}

\newcommand\addLabelDL[1]{\ar@{}[]+UR|(1){~\makebox[0pt][l]{$\mathbf{#1}$}}}

\newcommand\addLabelDR[1]{\ar@{}[]+UR|(1){~\makebox[0pt][l]{$\mathbf{#1}$}}}

\newcommand\addDMD[2]{
	\ar@{-}[]+<#1pt,0pt>;[]+<0pt,#2pt>
	\ar@{-}[]+<0pt,#2pt>;[]+<-#1pt,0pt>
	\ar@{-}[]+<-#1pt,0pt>;[]+<0pt,-#2pt>
	\ar@{-}[]+<0pt,-#2pt>;[]+<#1pt,0pt>
}

\begin{document}
%

\title{Probabilistic Logics \\ based on Riesz Spaces}

  

\author[R.~Furber]{Robert Furber}	
\address{Aalborg University, Denmark}	
\email{furber@cs.aau.dk}  

\author[R.~Mardare]{Radu Mardare}	
\address{University of Strathclyde, Scotland, UK}	
\email{r.mardare@strath.ac.uk}  

\author[M.~Mio]{Matteo Mio}	
\address{CNRS and ENS--Lyon, France}	
\email{matteo.mio@ens-lyon.fr}  


\begin{abstract}
We introduce a novel real--valued endogenous logic for expressing properties of probabilistic transition systems called \emph{Riesz modal logic}.
The design of the syntax and semantics of this logic is directly inspired by the theory of Riesz spaces, a mature field of mathematics
at the intersection of universal algebra and functional analysis. By using powerful results from this theory, we develop the duality theory of the Riesz modal logic
in the form of an algebra--to--coalgebra correspondence.
This has a number of consequences including: a sound and complete axiomatization, the proof that the logic characterizes probabilistic bisimulation and other convenient
results such as completion theorems. This work is intended to be the basis for subsequent research on extensions of Riesz modal logic with fixed--point operators.

\end{abstract}

\maketitle



\section{Introduction}

Directed graphs and similar structures, such as labelled transition systems and Kripke frames, are  mathematical objects often used to represent, by means of operational semantics, the behaviour of (nondeterministic) computer programs \cite{SOS}. For this reason a large body of research has focused on the study of logics for expressing useful properties of directed graphs.  Among these, \emph{modal logic} (see, e.g, \cite{BdRVModal,Stirling96,modallogic2008book}) and its extensions (e.g., \emph{CTL} \cite{CES83}, \emph{modal $\mu$-calculus} \cite{Kozen83}, among others) play a fundamental role. After decades of research, the current state of knowledge regarding modal logics is substantial:
\begin{itemize}[wide,itemindent=0pt]
\item[] \textbf{Model Theory}: the class of transition (often referred to as ``relational'') structures interpreting the language of modal logics and their interplay with modal formulas is well understood. This theory includes key concepts such as that of \emph{bisimulation}, algorithmically relevant properties such as the \emph{finite model property}, expressiveness and definability results and advanced constructions such as \emph{ultraproducts}. We refer to \cite{GorankoOtto2007} for an overview.\\

\item[] \textbf{Algebraic Semantics}: another natural approach to give semantics to modal logic is \emph{algebraic}: formulas are interpreted over algebras equipped with operations corresponding to the connectives of the logic and subject to certain axioms. In the case of basic modal logic (i.e., system \emph{K} \cite{modallogic2008book}), the signature is $\{ \top,\bot,\neg, \vee, \wedge, \Diamond\}$ and the algebras considered are called \emph{modal Boolean algebras} and satisfy the usual axioms of Boolean algebras together with the additional axioms $\Diamond \bot=\bot$ and $\Diamond(F\vee G) = \Diamond(F) \vee \Diamond(G)$ for the modal connective. Starting with the seminal works of J\'{o}nsson, McKinsey and Tarski \cite{bennett_1946,jonssontarski1951,jonssontarski1952}, a precise correspondence between algebraic and transition semantics has been established. They key tool being used is that of \emph{Stone duality}: to each Boolean algebra $B$ there corresponds a certain topological ``dual'' space $X$, and the modal operation $\Diamond: B \rightarrow B$ of each modal Boolean algebra corresponds to a transition (topologically closed) relation $R\subseteq X\times X$:

\begin{center}
modal Boolean algebra $(B,\Diamond)$ $\ \ \ \ \ \Longleftrightarrow  \ \ \ \ \ \ $ (topological) Kripke frame $(X, R)$. 
\end{center}
The correspondence is in fact a duality of categories when morphisms between Kripke frames are defined using the framework of coalgebra theory \cite{KKV2004,Jacobs2016}. This duality provides deep mathematical insights and is considered by Johan van Benthem as ``one of the three pillars of wisdom in the edifice of modal logic'' \cite{JVB1984} (the other two being \emph{completeness} and \emph{correspondence theory}). We refer to \cite{SambinVaccaro1988} and \cite{KKV2004} for an overview.\\

\item[] \textbf{Axiomatizations and Proof Systems}: axiomatizations (sound and complete with respect to the semantics) have been found for modal logic and many of its extensions, including  CTL \cite{reynolds2001} and modal $\mu$-calculus \cite{Walukiewicz-CompletenessofKozen}. Furthermore, structural (analytic) proof systems based on Gentzen's \emph{sequent calculus} have been designed. These constitute the purely syntactical side of the theory of modal logic. We refer to \cite{prooftheorymodallogic} for a general introduction and to  \cite{Studer07,DHL2006,doumanephd} for a selection of some recent results. \\
\end{itemize}


\subsection{Probabilistic logics}
\noindent Despite their wide applicability, directed graph structures are not adequate for modelling all kinds of programs. Most notably probabilistic programs, such as those involving commands for generating random numbers (e.g., \verb!x=rand()! in \verb!C++!), are naturally modelled by Markov chains or similar structures (e.g.,   Markov decision processes). See, e.g., \cite{BaierKatoenBook} for an overview. Consequently, a number of logics specifically designed to express properties of Markov chains have been investigated: e.g., Kozen's \emph{probabilistic PDL} \cite{Kozen1981,Kozen1983}, Larsen--Skou modal logic \cite{LS91}, \emph{probabilistic CTL} (\cite{HJ94,LS1982}), among others. We generally refer to such logics as \emph{probabilistic logics}.

The current status of the theory of probabilistic logics is, compared with that of ordinary modal logics, rather unsatisfactory. For example, for most probabilistic logics capable of expressing properties useful in model checking (probabilistic CTL is a main example), the following problems have been open for more than 35 years (since, at least, \cite{LS1982}):
\begin{enumerate}
\item find a sound and complete axiomatization of the set of valid formulas,
\item find structural proof systems (e.g., sequent calculus)  for deriving valid formulas,
\item establish if the set of valid formulas is decidable or not.
\end{enumerate}

These problems are evidently intrinsically difficult but one reason that makes them harder to tackle is, possibly, the fact that most probabilistic logics (including probabilistic CTL) have been designed with special focus on model--checking (e.g., the ability to express properties useful in practice, availability of efficient algorithms for verifying  finite--state systems, \emph{etc.}) rather than mathematical convenience.  This has led to successful results, with real--world probabilistic programs formally verified using model checking techniques. But, on the other hand, very little progress has been made on the open problems listed above.

\subsection{Real--valued probabilistic logics}

The seminal work of  Kozen on \emph{probabilistic PDL}  \cite{Kozen1983} is noteworthy as being among the first to the design probabilistic logics with main focus on convenient mathematical foundations. A key novelty of probabilistic PDL is the fact that its semantics is \emph{real--valued}: formulas are not interpreted as  \emph{true} or \emph{false}, as in most other probabilistic logics with a Boolean semantics (including probabilistic CTL), but are rather interpreted as real numbers  ($\mathbb{R}$). So in real--valued logics the semantics of a given formula $F$ can be a number like $0$, $1$, $\frac{7}{19}$ and $\pi$. The adequacy and mathematical convenience of a real--valued semantics in the context of probabilistic logics is discussed in detail in \cite{Kozen1981}.

However, the logic probabilistic PDL is, using the terminology introduced by Pnueli \cite{PNUELI77}, \emph{exogenous}:  the language of formulas is both an assertion language able to express properties of probabilistic programs and a programming language. The logic probabilistic CTL and most other logics for model checking are, on the other hand,  \emph{endogenous}: the language of formulas is independent from the concrete syntax of any given programming language. This distinction is important because the ``programming languages'' embedded in exogenous logics are usually quite abstract and restricted (e.g., consist only of the usual Kleene algebra operations) while endogenous logics express properties of models generated by arbitrary programs. One of the consequences is, for example, that the logic probabilistic PDL enjoys the \emph{finite model property} \cite{Kozen1983} while probabilistic CTL does not \cite{LS1982,Brazdil2008}. The failure of the final model property is a fundamental characteristic of probabilistic CTL and a main  source of complexity.  This means that it does not seem possible to directly apply the results available for probabilistic PDL to solve the open problems regarding probabilistic CTL (and other endogenous logics) listed above.

For this reason, following Kozen, some research has subsequently explored the design of real--valued endogenous probabilistic logics based on the idea of interpreting formulas as real numbers. Early works include \cite{MM07,HM96,prakash2000,deAlfaro2003}.  A shortcoming of these attempts is, however, that these probabilistic logics are not sufficiently expressive to interpret the logic probabilistic CTL and other endogenous probabilistic logics having the usual Boolean semantics. Recent works \cite{MioThesis,MIO2012b,MIO2014a,MioSimpsonFI2017} have shown, however, that the desired expressivity can be achieved by extending a simple real--valued probabilistic modal logic (called \emph{\L ukasiewicz modal logic}) with fixed--point operators, in the style of Kozen's modal $\mu$--calculus \cite{Kozen83}. Indeed, the resulting real--valued logic, called \emph{\L ukasiewicz $\mu$--calculus} can interpret the logic probabilistic CTL.  Hence the real--valued approach to  endogenous modal logics for probabilistic systems suffices to express most properties of interest:

\begin{center}
\begin{tabular}{c c c}
simple real--valued modal logic  & & \\
$+$ & $\supseteq$ & $\textnormal{probabilistic CTL}$\\
(co)inductively defined operators\\
\end{tabular}
\end{center}
This observation suggests the following research program:
\begin{enumerate}
\item Identify a simple real--valued endogenous modal logic $\mathcal{L}$ having convenient mathematical foundations which, once extended with fixed--point operators, is sufficiently expressive to interpret probabilistic CTL and other probabilistic logics, just like the \emph{\L ukasiewicz modal logic} mentioned above.

\item Develop the theory of the probabilistic real--valued logic $\mathcal{L}$: model theory, algebraic theory, axiomatizations and proof systems.
\item Extend $\mathcal{L}$ with the (co)inductive operators required to increase its expressive power.
\item Develop the theory of the extended logic using the large body of knowledge on methods for reasoning about fixed points and (co)inductive definitions.
\end{enumerate}

The main contribution of this work is to set down the mathematical foundation of a logic $\mathcal{L}$  enjoying the properties outlined above. 

\subsection{Contributions of this work}

We introduce the \emph{Riesz modal logic}, a real--valued probabilistic endogenous modal logic named in honour of  the Hungarian mathematician Frigyes Riesz. The design of the syntax and the semantics of this logic is inspired by the theory of \emph{Riesz spaces}, also known as \emph{vector lattices} \cite{Luxemburg,JVR1977}, a branch of mathematics at the intersection of algebra and functional analysis, pioneered in the 1930's by F. Riesz, L. Kantorovich and H. Freudenthal among others, with applications in the study of function spaces.

A Riesz space (see Section \ref{sec:background:riesz} for the details) is a real--vector space $V$ equipped with a lattice order ($\leq$) such that the basic vector space operations of addition and scalar multiplication satisfy:\\
\begin{center}
 if $x\leq y$ then $x+z \leq y + z$, $ \ \ \ \ \ \ \ \ $
 if $x\leq y $ then $r x \leq r y$, for any scalar $r\geq0$.\\
\end{center}
$\ $ \\
For example, the linearly ordered set of real numbers $\mathbb{R}$ is a Riesz space. Hence the concept of Riesz space is obtained by combining the notion of lattice,  which is pervasive in logic, with those of addition and scalar multiplication, which are pervasive in probability theory (e.g., convex combinations, linearity of the expected value operator, \emph{etc}.)
 
In the context of our work, it is convenient to think at Riesz spaces as a quantitative generalization of Boolean algebras, obtained by replacing the two--element Boolean algebra $\mathbf{2}=(\{0,1\},\vee,\wedge,\neg)$ with the Riesz space $\mathbb{R}$. This is not just a vague analogy as the theory of Riesz spaces is very rich and includes key results such as:
\begin{itemize}
\item $\mathbb{R}$ generates the variety of Riesz spaces, just like $\mathbf{2}$ generates the variety of Boolean algebras,
\item \emph{Yosida duality}, which is the equivalent of Stone duality for Boolean algebras, provides a  bridge between algebra and topology, 
\item completion theorems, just as in Boolean algebras, allow one to embed Riesz spaces into other Riesz spaces whose order has certain closure properties (e.g., it is a complete lattice), \emph{etc}. This is convenient, for example, when it is required to guarantee the existence of fixed--points of monotone operators as in the Knaster--Tarski fixed--point theorem. 
\end{itemize}
What makes Riesz spaces particularly convenient for our applications to probabilistic logics is that, being vector spaces, the notion of \emph{linear transformation} plays a fundamental role in the theory. For example the theory of Riesz spaces include results such as:  
\begin{itemize}
\item theory of linear functionals: representation theorems such as, e.g., the Riesz representation theorem for (probability) measures,
\item extension theorems: e.g.,  generalizations of the Hahn-Banach theorem.
\end{itemize}
For these reasons we claim that the theory of Riesz spaces is a very convenient mathematical setting to develop the theory of probabilistic logics.

We define the \emph{transition semantics} of Riesz modal logic with respect to transition systems modelled as  (topological) Markov chains, which we refer to as \emph{Markov processes}.
Formally, these are defined as coalgebras of the Radon monad on the category of compact Hausdorff spaces (see Section \ref{sec:background:coalgebra}). The theory of coalgebra then provides automatically appropriate definitions of morphisms between models, products, quotients, bisimulation, \emph{etc.} Beside the choice of the category to work with, which is motivated by mathematical convenience and is at the same time very general (see discussion in Section \ref{other:models:section}), the semantics is essentially standard, it agrees with several other works in the literature and is based on the interpretation of the $\Diamond$ modality with the expected--value operator. And indeed we show that Riesz modal logic can interpret other basic real--valued logics that have appeared in the literature including the \L ukasiewicz modal logic (see Section \ref{sec:other:logics}).

\subsection{Technical results}

Our main technical contribution is to set the mathematical foundation of Riesz modal logic by developing its duality theory, following closely the duality theory framework of ordinary modal logic. To do this, we define an \emph{algebraic semantics} for Riesz modal logic. 

The algebras are called \emph{modal Riesz spaces} and are Riesz spaces $R$ equipped with an additional unary operation $\Diamond: R\rightarrow R$ subject to the following  axioms (see Figure \ref{axioms:of:modal:riesz:spaces} in Section \ref{section_modal}):
\begin{center}
\begin{tabular}{l l l}
Linearity: &&  $\Diamond(r_1x + r_2y ) = r_1\Diamond(x) + r_2\Diamond(y)$,\\
Positivity:  && if $x\geq 0$ then $\Diamond x\geq 0$, and \\
$1$-decreasing:  && $\Diamond 1 \leq 1$.
\end{tabular}
\end{center}
This variety of algebras forms a category by taking homomorphisms  (i.e., mappings preserving all operations) as morphisms.  By applying the machinery of Yosida duality, and other results from the theory of Riesz spaces, we prove that the category of transition models (coalgebras) is dually equivalent with the category of Archimedean modal Riesz spaces (Theorem \ref{main_theorem_paper}).

This result has a number of consequences. Firstly, Riesz modal logic characterizes bisimulation (Corollary \ref{corollary:bisimulation}). Secondly, we obtain a sound and complete axiomatizations of Riesz modal logic (Theorem \ref{completeness_theorem_app}). The axioms and inferences rules are depicted in Figure \ref{figure:full:axiomatisation} in Section \ref{section_applications}. While other simple probabilistic modal logics characterizing bisimulation have been completely axiomatized in the literature (see, e.g., the \emph{Markovian logic} of \cite{StonePrakash,KMP2013}), to the best of our knowledge, Riesz modal logic is the first probabilistic logic which, once extended with fixed--point operators, is sufficiently expressive to interpret other expressive logics such as probabilistic CTL.

Using duality theory, we can investigate properties of the final coalgebra (which, in the context of operational semantics, is understood as the collection of ``behaviours'' \cite{Jacobs2016,KurzPHD}) by establishing results of the initial modal Riesz space, and \emph{vice versa}. We prove some fundamental properties of the initial modal Riesz space in Section \ref{final_coalgebra_section}. These allow us, for instance, to prove that the final coalgebra is a compact Polish space (Theorem \ref{corollary:polish}).

Riesz modal logic is, by design, a very simple formalism and lacks temporal operators needed to express many of the useful properties expressible in logics such as probabilistic CTL. As already mentioned, the required expressiveness can be achieved by extending Riesz modal logic with recursively defined operators, in the style of Kozen's modal $\mu$--calculus. This has been shown in, e.g., \cite{MIO2012b,MioSimpsonFI2017,MIO2014a,Mio18}. Fixed--point definitions usually rely on the Knaster--Tarski theorem on complete lattices. In this context, by applying a theorem of Kantorovich in the theory of Riesz spaces,  we prove a fundamental completion result (Theorem \ref{completion:thm2}): every Archimedean modal Riesz space can be embedded in a Dedekind complete modal Riesz space. This, by duality, means that any topological Markov chain (coalgebra) can be embedded into a topological Markov chain having a state--space which is \emph{Stonean} (i.e., the Stone--dual of a complete Boolean algebra). 

\subsection{Organization of this work}
This article is organized as follows:\\

\noindent
\textbf{Section \ref{section_background}: Technical background}. We provide the necessary background definitions and results. This section is quite lengthy but, hopefully, serves the purpose of keeping this article reasonably self--contained. Subsections \ref{TopMeasRieszRepSubSect} and \ref{sec:background:coalgebra} deal with basic notions from probability theory and coalgebra and can be safely skipped by readers familiar with these topics. Subsections \ref{sec:background:riesz}, \ref{sec:riesz:back2}, \ref{duality_section_background} and \ref{sec:dedekind} deal with the basic definitions and results of the theory of Riesz spaces. Once again, these can be safely ignored by readers familiar with this theory and consulted only when necessary.\\

\noindent
\textbf{Section \ref{logic_section}: Riesz Modal Logic, Syntax and Transition Semantics.} In this section we define the syntax and the transition semantics of Riesz modal logic. The latter is given in terms of topological Markov chains, which we refer to as \emph{Markov processes}, defined in Section  \ref{sec:background:coalgebra}. We give several examples of formulas and, in Subsection \ref{sec:other:logics}, explain how Riesz modal logic can interpret other similar real--valued probabilistic modal logics that have appeared in the literature, including the  \emph{\L ukasiewicz modal logic} of \cite{MioThesis,MIO2014a,MioSimpsonFI2017}.\\

\noindent
\textbf{Section \ref{section_modal}: Modal Riesz spaces.} In this section we introduce the notion of \emph{modal Riesz space}, the algebraic counterpart of Riesz modal logic. In Subsection \ref{dedekind:sec} we establish a completion theorem for modal Riesz spaces (Theorem \ref{completion:thm2}). This result is likely of fundamental importance in the future development of fixed--point extensions of Riesz modal logic based on the Knaster--Tarski theorem. In Subsection \ref{relation:MValgebras} we comment on some similarities with the notion of \emph{state MV--algebra} from \cite{FM2009}.\\

\noindent
\textbf{Section \ref{section_duality}: Duality between Markov processes and modal Riesz spaces.} In this section we establish our main technical result (Theorem \ref{main_theorem_paper}): the categories of uniformly complete Archimedean modal Riesz spaces and that of Markov processes with coalgebra morphisms are dually equivalent.\\

\noindent
\textbf{Section \ref{final_coalgebra_section}: Initial algebra.} In this section we give explicit constructions of the initial objects of several categories of modal Riesz spaces and establish some basic properties. We also leave open an important question which we could not answer so far (see Subsection \ref{sec:AURiesz}). \\

 \noindent
 \textbf{Section \ref{final:section}: Final coalgebra.} In this section we illustrate one application of the duality theory: it is possible to establish properties of the final coalgebra in the category of Markov processes by proving properties of the initial modal Riesz space. We prove that the state--space of the final coalgebra is a compact Polish space.\\

 \noindent
 \textbf{Section \ref{section_applications}: Applications of duality to Riesz modal logic.}  Another application of the duality theory is, of course, to establish properties of Riesz modal logic. In this section we prove that Riesz modal logic characterizes probabilistic bisimilarity and that the proof system of Figure \ref{figure:full:axiomatisation}, for proving semantic equality between pairs of Riesz modal logic formulas, is sound and complete.\\

 \noindent
 \textbf{Section \ref{other:models:section}: Other classes of models}: In this section we show how our notion of Markov process (as given in Definition \ref{markov_process_coalgebra}) is in fact very general in the sense that most of the similar known notions can be embedded into Markov processes in our sense.\\
 
 \noindent
 \textbf{Section \ref{conclusion_section}: Conclusions:} In this section we present some final comments and direction for future research.\\

 \noindent
 \textbf{Appendix \ref{QuotientCompleteSubsection}:} In this appendix we prove a result regarding Archimedean Riesz spaces needed in the proof of Lemma \ref{SurjectiveUniformlyCompleteLemma} of Section \ref{final_coalgebra_section}. This might well be a known result but we could not find any explicit reference for it in the literature.


\section{Technical background}
\label{section_background}

\subsection{Topology, measures and Riesz--Markov--Kakutani representation theorem}
\label{TopMeasRieszRepSubSect}

$ \ $ \\

We denote by $\CHaus$ the category of compact Hausdorff spaces with continuous maps as morphisms. If $X$ is a compact Hausdorff space, we denote with $\mathcal{B}(X)$ the collection of Borel sets of $X$, \ie, the smallest $\sigma$-algebra of subsets of $X$ containing all open sets. A \emph{(Borel) subprobability measure} on $X$ is a function $\mathcal{B}(X)\rightarrow[0,1]$ such that $\mu(\emptyset)\!=\!0$, $\mu(X)\!\leq\!1$ and $\mu(\bigcup_n A_n)\!=\!\sum_n \mu(A_n)$ for all countable sequences $(A_n)$ of pairwise disjoint Borel sets. The measure $\mu$ is a \emph{probability measure} if $\mu(X)\!=\!1$.

A (sub--)probability measure $\mu$ on the compact Hausdorff space $X$ is \emph{Radon} if for every Borel set $A$, $\displaystyle\mu(A)\!=\!\sup\{\ \mu(K) \mid K\!\subseteq\! A \textnormal{ and $K$ is compact}\}$. In other words, a measure is Radon if the measure $\mu(A)$ of every Borel set $A$ can be approximated to any degree of precision by compact subsets of $A$. Most naturally occurring probability (sub--)measures are Radon. In particular, if $X$ is a Polish space, all (sub--)probability measures are Radon. 

Given a set $X$, we denote the collection of all functions $X\!\rightarrow\!\mathbb{R}$ by $\R^X$. If $X$ is a topological space, then $C(X)$ denotes the subset of $\R^X$ consisting of all continuous functions. We use $\zero_X$ and $\one_{X}$ to denote the constant (continuous) functions defined as $\zero_X(x)=0$ and $\one_X(x) = 1$, for all $x \in X$, respectively. Using the vector space operations of $\R$ pointwise, both  $\R^X$ and $C(X)$ can be given the structure of a $\mathbb{R}$-vector space. Furthermore, the ordering ($\leq$) defined  pointwise as $f\leq g \Leftrightarrow \forall x.f(x)\leq g(x)$ is a lattice on both $\R^X$ and $C(X)$.

Given a compact Hausdorff space $X$ and a (sub--)probability measure $\mu$ on $X$, one can define the expectation functional $\mathbb{E}_\mu: C(X)\rightarrow\mathbb{R}$ as 
\vspace{-2mm}
\begin{equation}
\mathbb{E}_\mu (f) = \int_{X} f \diff \mu
\end{equation}

where the integral is well defined because any $f\!\in\! C(X)$, being continuous and defined on a compact space, is measurable and bounded. One can then  observe that:
\begin{enumerate}[label=(\roman*)]
\item  $\mathbb{E}_\mu$ is a \emph{linear} map:  $\mathbb{E}_\mu (f_1+f_2)\!=\! \mathbb{E}_\mu (f_1) + \mathbb{E}_\mu (f_2)$, and $ \mathbb{E}_\mu(r f) = r \mathbb{E}_\mu(f)$, for all $r\!\in\!\mathbb{R}$,
\item  $\mathbb{E}_\mu$ is \emph{positive}:  if $f\geq \zero_X$ then $\mathbb{E}_\mu(f)\geq 0$, and
\item $\mathbb{E}_\mu$ is \emph{$\one_X$-decreasing}: $\mathbb{E}_\mu(\one_X)\leq 1$. 
\end{enumerate}
The latter inequality becomes an equality if $\mu$ is a probability measure.

The celebrated Riesz--Markov--Kakutani representation theorem states 
 that, in fact, any such functional corresponds to a unique Radon subprobability (see \cite{LAX}).

\begin{theorem}[(Riesz--Markov--Kakutani)]
Let $X$ be a compact Hausdorff space. For every functional $F:C(X)\rightarrow\mathbb{R}$ such that (i) $F$ is linear, (ii) $F$ is positive and (iii) $F(\one_X)\leq 1$, there exists one and only one Radon subprobability measure $\mu$ on $X$ such that $F=\mathbb{E}_\mu$.
\end{theorem}

Given a compact Hausdorff space $X$ we denote with $\Rdnl(X)$ the collection of all Radon subprobability measures on $X$. Equivalently, by the Riesz--Markov--Kakutani theorem, we can identify $\Rdnl(X)$ with the collection of functionals
$$
\big\{ F:C(X)\rightarrow\mathbb{R} \mid \textnormal{$F$ is linear, positive and $\one_X$-decreasing}\big\}.
$$

The set $\Rdnl(X)$ can be endowed with the weak-* topology, the coarsest (\ie, having fewest open sets) topology such that, for all $f\!\in\! C(X)$, the map $T_{f}\!:\!\Rdnl(X)\!\rightarrow\!\mathbb{R}$, defined as $T_{f}(F)\!=\!F(f)$, is continuous. The weak-* topology on $\Rdnl(X)$ is compact and Hausdorff by the Banach-Alaoglu theorem. Hence $\Rdnl$ maps a compact Hausdorff space $X$ to the compact Hausdorff space $\Rdnl(X)$. In fact $\Rdnl$ becomes a functor on $\CHaus$  by defining, for any continuous map $f \!:\! X\! \rightarrow\! Y$ in $\CHaus$, the continuous map $\Rdnl(f)\!:\! \Rdnl(X)\!\rightarrow\! \Rdnl(Y)$ as
\begin{equation}
\label{RdnlMapDefn}
\Rdnl(f)(F)(g) = F(g \circ f),
\end{equation}
for all $g\in C(Y)$. 

The functor $\Rdnl$ is shown to be the underlying functor of a monad in \cite[\S 6]{Keimel2009}, based on \'Swirszcz's proof of the probabilistic case \cite{Swirszcz74,swirszcz75} (see also Giry's work \cite{giry1981}). However, we will not require the monad structure for the purposes of this article. Following \cite{FurberJ14a}, we call $\Rdnl$ the \emph{Radon monad}.

\subsection{Markov Processes and Coalgebra}\label{sec:background:coalgebra}

$ \ $ \\

Informally, a (discrete-time) Markov process consists of a set of states $X$ and a transition function $\alpha$ that associates to each state $x\!\in\!X$ a probability distribution $\alpha(x)$ on the state space $X$. This mathematical object is interpreted,  given an initial state $x_0$, as generating an infinite trajectory (or ``computation'') $(x_n)_{n\in\mathbb{N}}$ in the state space $X$, where $x_{n+1}$ is chosen randomly using the probability distribution $\alpha(x_n)$. A slight variant of this model, allowing the generation of infinite as well as finite trajectories,  uses transition functions $\alpha$ associating to each state $x$ a subprobability distribution $\alpha(x)$. The intended interpretation is that the computation will stop at state $x$ with probability $1-m_x$, where $m_x\!\in\![0,1]$ is the total mass of $\alpha(x)$, and will continue with probability $m_x$ following the (normalized) probability distribution $\alpha(x)$.

\begin{example}\label{logic:example1a}
Consider the following Markov process having state space $X=\{x_1,x_2\}$ and transition function $\alpha$ defined by: $\alpha(x_1) = \frac{1}{3}x_1 + \frac{1}{2} x_2$ and $\alpha(x_2) =\frac{1}{3} x_1$ depicted below:
\vspace{-4mm}
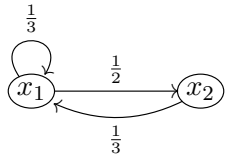
\begin{center}
$$
\SelectTips{cm}{}
	\xymatrix @=20pt {
		\nodeC{x_1} \ar@{->}[rr]^{\frac{1}{2}}\ar@{->}@(ul,ur)^{\frac{1}{3}}   & &  \nodeC{x_2} \ar@{->}@/^10pt/[ll]^{\frac{1}{3}}   	}
$$
\end{center}
From the state $x_1$ the computation progresses to $x_1$ itself with probability $\frac{1}{3}$, to $x_2$ with probability $\frac{1}{2}$ and it halts with probability $\frac{1}{6}$.  From the state $x_2$ the computation progresses to $x_1$ with probability $\frac{1}{3}$ and it  halts with probability $\frac{2}{3}$.
\end{example}

This informal description readily translates to a formal definition for Markov processes having finite or countably infinite state space $X$, also known as \emph{Markov chains}. Sometimes, however, it is interesting to model Markov process having uncountable state spaces (e.g., $X=[0,1]$). When $X$ is uncountable, the notion of discrete probability distribution is naturally replaced by that of probability measure and, therefore, $X$ is often assumed to be a topological or measurable space and $\alpha$ is defined as a map from $X$ to the collection of (sub--)probability measures on $X$ satisfying certain convenient regularity assumption.

In this work we define Markov processes as follows.
\begin{definition}\label{markov_def_1}
A \emph{Markov process} is a pair $(X,\alpha)$ such that $X$ is a compact Hausdorff topological space and $\alpha:X\rightarrow\Rdnl(X)$ is a continuous map.
\end{definition}

\begin{example}[Finite Markov chains]\label{example:mp:1}
Finite Markov chains, such as the one defined in the example \ref{logic:example1a} above, can be formalized as Markov processes in the sense of Definition \ref{markov_def_1}.  Indeed, the finite state space $X$, endowed with the discrete topology, is a compact Hausdorff space. And the transition function $\alpha: X\rightarrow \Rdnl(X)$ is continuous (since $X$ is discrete). Observe that the space $\Rdnl(X)$ is isomorphic to the set $\mathcal{D}^{\leq 1}(X)=\{ d:X\rightarrow[0,1] \mid \sum_x d(x)\leq 1\}$ of all subprobability distributions on $X$.

\end{example}

\begin{example}[Uncountable Markov process]\label{example:mp:2}

We define a Markov process having state space $X=[0,1]$ where, from the state $x\in[0,1]$, the computation progresses to $x$ with probability $x$ and it halts with probability $1-x$.  This is formalized by defining the transition function $\alpha:X\rightarrow\Rdnl(X)$ as follows:

$$
\alpha(x) = x \cdot \delta_x 
$$ 
where $\delta_x$ is the Dirac probability measure centred on $x\in [0,1]$. Note that $\alpha$, being the pointwise product of the continuous function ($x\mapsto x)$ and the continuous function ($x\mapsto \delta_x$), is indeed continuous.
\end{example}

The previous example is included particularly because it turns out to be useful in proving Theorem \ref{uniform-incompleteness:robert}, via Example \ref{example:logic:3} and Lemma \ref{PiecePolyLemma}. However, as a general example it is amenable to criticism on the grounds that it does not use what we would consider to be continuous probability distributions. Therefore we provide one further example.

\begin{example}
We define a Markov process $\alpha: S^1 \rightarrow \Rdnl(S^1)$ on the unit circle $S^1$ with the property that from each point there is a nonzero probability of moving to any interval within $S^1$. The circle $S^1$ is measurably isomorphic to $(-\pi,\pi]$ via the usual parametrization by angle. Therefore we can consider the Lebesgue measure $\lambda$ restricted to $S^1$ and define an $S^1$-indexed family of probability density functions on $S^1$:
\begin{align*}
f_\zeta &: S^1 \rightarrow \R \\
f_\zeta(\theta) &= \frac{1 + \cos(\zeta + \theta)}{2\pi},
\end{align*}
from which we can define an $S^1$-indexed family of probability measures on $S^1$, each with a peak at $\zeta$ but nonzero probability of moving to any interval within $S^1$:
\[
\alpha(\zeta) = f_\zeta \cdot \lambda,
\]
which is to say, for every bounded real-valued measurable function $g$ on $S^1$ we have
\[
\int_{S^1} g \diff \alpha(\zeta) = \int_{S^1} g f_\zeta \diff \lambda.
\]
It follows from the continuity of $\cos$ and arithmetic operations that if $(\zeta_i)_{i \in \mathbb{N}}$ converges to $\zeta$ in $S^1$, then $f_{\zeta_i} \to f_\zeta$ pointwise. By the dominated convergence theorem, this implies convergent sequences in $S^1$ map to convergent sequences in $\Rdnl(S^1)$ under $\alpha$. Since $S^1$ is metrizable, we can conclude that $\alpha$ is continuous. 
\end{example}

\begin{remark}
While these examples are natural, Definition \ref{markov_def_1} might appear unnecessarily restrictive because several practically interesting classes of probabilistic systems do not have a state space endowed with a compact  topology (e.g., $\mathbb{N}$ and $\mathbb{R}$ are not compact) and often the transition functions are not continuous (e.g., they are just Borel measurable). The choice of using the class of compact Hausdorff spaces and continuous transition functions in Definition \ref{markov_def_1} is mostly motivated by mathematical convenience since, as described later, this is the class of topological spaces appearing in the duality theory of Riesz spaces. We will explain in detail in Section \ref{other:models:section} how this is not at all a restriction when it comes to Riesz modal logic.
\end{remark}

The theory of coalgebra (for a comprehensive introduction see \cite{Jacobs2016}) provides a convenient framework for formalizing the notion of morphism between Markov processes. The following is an equivalent reformulation of Definition \ref{markov_def_1} in coalgebraic terms and relies on the fact, discussed earlier, that $\Rdnl$ is an endofunctor on the category $\CHaus$.
\begin{definition}\label{markov_process_coalgebra}
A \emph{Markov process} is a coalgebra of the endofunctor $\Rdnl$ in the category  $\CHaus$, \ie, it is a morphism $\alpha\!:\! X\! \rightarrow\! \Rdnl(X)$ in $\CHaus$. A \emph{(coalgebra) morphism} between the coalgebra $\alpha\!:\! X\! \rightarrow\! \Rdnl(X)$ and the coalgebra $\beta\!:\! Y\! \rightarrow\! \Rdnl(Y)$ is a continuous function $f\!:\!X\!\rightarrow\! Y$ such that the following diagram commutes:
\begin{equation}
\label{MarkovMorphDiag}
\vcenter{\xymatrix{
X \ar[r]^-\alpha \ar[d]_f & \Rdnl(X) \ar[d]^{\Rdnl(f)} \\
Y \ar[r]_-\beta & \Rdnl(Y).
}}
\end{equation}
Such a morphism will be denoted by $\alpha \stackrel{f}{\rightarrow}\beta$.
\end{definition}

\begin{definition}[Category of Markov Processes]
We define the category $\Markov$ of Markov processes to be $\CoAlg(\Rdnl)$ where the objects are coalgebras $\alpha\!:\! X\! \rightarrow\! \Rdnl(X)$  in  $\CHaus$
and morphisms $\alpha \stackrel{f}{\rightarrow}\beta$ are coalgebra morphisms. 
\end{definition}

It is a well known fact that $\CoAlg(F)$ is always a category, for any functor $F$. 
In computer science, and in particular in the field of categorical semantics of programming languages, one specific coalgebra in $\CoAlg(F)$ plays an important role. This is (when it exists) the final object $\alpha: X\rightarrow F(X)$ in $\CoAlg(F)$, and is called the \emph{final coalgebra}. The universal property that characterizes $\alpha$ is that, for every other $F$-coalgebra $\beta:Y\rightarrow F(Y)$, there exists one and only one coalgebra morphism $\beta \stackrel{\eta}{\rightarrow} \alpha$ in $\CoAlg(F)$. This property allows to interpret the domain $X$ of $\alpha$ as the space of all ``behaviours'' as follows: given any coalgebra $\beta:Y\rightarrow\Rdnl(Y)$, the behaviour of the state $y$ is the point $\eta(y)\in X$. And two states $y_1,y_2\in Y$ are ``behaviourally equivalent'' if $\eta(y_1)=\eta(y_2)$. 

For this reason, in Section \ref{final:section} we study some properties of the final Markov process, \ie, the final object in $\Markov$.

\subsection{Riesz Spaces}\label{sec:background:riesz}

$ \ $ \\

This section contains the basic definitions and results related to Riesz spaces. We refer to \cite{Luxemburg} for a comprehensive reference to the subject.

A Riesz space is an algebraic structure $(A,0,+,(r)_{r\in\mathbb{R}},\sqcup,\sqcap)$ such that 
$(A,0,+,(r)_{r\in\mathbb{R}})$ is a vector space over the reals, $(A,\sqcup,\sqcap)$ is a lattice and the induced order $(a\leq b \Leftrightarrow a\sqcap b = a)$ is compatible with addition and with the scalar multiplication, in the sense that: (i)  for all $a,b,c\in A$, if $a\leq b$ then $a+c\leq b+c$, and (ii) if $a\geq b$ and $r\in \mathbb{R}_{\geq 0}$ is a non--negative real, then $r a\geq rb$. Formally we have:
\begin{definition}[Riesz Space]
The \emph{language} $\mathcal{L}_R$ of Riesz spaces is given by the (uncountable) signature $\{ 0,+, (r)_{r\in\mathbb{R}}, \sqcup, \sqcap\}$ where $0$ is a constant, $+$, $\sqcup$ and $\sqcap$ are binary functions and $r$ is a unary function, for all $r\in\mathbb{R}$. A \emph{Riesz space} is a $\mathcal{L}_R$-algebra satisfying the equations of Figure \ref{axioms:of:riesz:spaces}. We use the standard abbreviations of $-x$ for $(-1)x$ and $x\leq y$ for $x\sqcap y = x$.
\end{definition}

Note how the compatibility axioms have been equivalently formalized in Figure \ref{axioms:of:riesz:spaces} as inequalities and not as implications by using $(x\sqcap y)$ and $y$ as two general terms automatically satisfying the hypothesis $(x\sqcap y)\leq y$. Since the inequalities can be rewritten as equations using the lattice operations ($x\leq y \Leftrightarrow x\sqcap y = x$), the family of Riesz spaces is a variety in the sense of universal algebra.

\begin{figure}[h!]
\begin{mdframed}
\begin{center}
 \begin{enumerate}
\item Axioms of real vector spaces:
\begin{itemize}
\item Additive group: $x + (y + z) = (x + y) + z$, $x + y = y + x$, $x + 0 = x$, $x - x= 0$,
\item Axioms of scalar multiplication: $r_1(r_2 x) = (r_1\cdot r_2) x$, $1x = x$, $r(x+y) = (rx) + (ry)$, $(r_1 + r_2)x = (r_1 x) + (r_2 x)$,
\end{itemize}
\item Lattice axioms:    (associativity) $x \sqcup (y \sqcup z) = (x \sqcup y) \sqcup z$,  $x \sqcap (y \sqcap z) = (x \sqcap y) \sqcap z$, (commutativity) $z \sqcup y = y \sqcup z$, $z \sqcap y = y \sqcap z$,
(absorption) $z \sqcup (z \sqcap y) = z$, $z \sqcap (z \sqcup y) = z$, (idempotence) $x\sqcup x =x$,  $x\sqcap x =x$. 
\item Compatibility axioms:  
\begin{itemize}
\item $(x \sqcap y) + z \leq  y + z $,
\item $r (x \sqcap y) \leq  ry$, for all scalars $r\geq 0$. 
\end{itemize}
\end{enumerate}
\end{center}
\end{mdframed}
\caption{Equational axioms of Riesz spaces.}
\label{axioms:of:riesz:spaces}
\end{figure}


\begin{example}\label{background:example1}
The most familiar example is the real line $\mathbb{R}$ with its usual linear ordering, \ie, with $\sqcup$ and $\sqcap$ being the usual $\max$ and $\min$ operations. An important fact about this Riesz space is the following (see, e.g., \cite{LvA2007}). Given two terms $t_1,t_2$ in the language of Riesz spaces, the equality $t_1=t_2$ holds in all Riesz spaces if and only if $t_1=t_2$ is true in $\mathbb{R}$. In the terminology of universal algebra one says that $\mathbb{R}$ generates the variety of all Riesz spaces. In this sense $\mathbb{R}$ plays in the theory of Riesz spaces a role similar to the two-element Boolean algebra $\{0,1\}$ in the theory of Boolean algebras.
\end{example}

\begin{example}\label{example2_subalgebra}
 For an example of Riesz space whose order is not linear take the vector space $\mathbb{R}^n$ with order defined pointwise: $(x_1,\dots, x_n)\!\leq\! (y_1,\dots, y_n) \Leftrightarrow x_i\!\leq\! y_i$, for all $1\!\leq\! i\!\leq\! n$. More generally, for every set $X$, the set $\mathbb{R}^X=\{ f: X\rightarrow\mathbb{R}\}$ with operations defined pointwise is a Riesz space. Since Riesz spaces are algebras, other examples can be found by taking sub-algebras. For instance, the collection of bounded functions $\ell^\infty(X) = \{ f \in \mathbb{R}^X \mid \textnormal{$f$ is bounded}\}$ is a Riesz subspace of $\mathbb{R}^X$. As another example, if $X$ is a topological space, then the set of continuous functions $C(X)=\{ f \in \mathbb{R}^X \mid \textnormal{$f$ is continuous}\}$ is another Riesz subspace of  $\mathbb{R}^X$.
\end{example}

The following definitions are useful. Let $A$ be a Riesz space. An element $a$ is \emph{positive} if $a\geq 0$. The set of all positive elements is called the \emph{positive cone} and is denoted by $A^+$. Given an element $a\!\in\! A$, we define $a^+=a\sqcup 0$, $a^-= -a\sqcup 0$ and $|a|=a^+ + a^-$. Note that $a^+, a^-, |a|\in A^+$,  $a^+ = (-a)^-$, $a^- = (-a)^+$ and $a= a^+ - a^-$.

\begin{definition}[Archimedean Riesz space]\label{archimedean:def}
An element $a\in A$ of a Riesz space is \emph{infinitely small} if there exists some $b\in A$ such that $n |a|\leq |b|$, for all $n\in\mathbb{N}$. Clearly, $0$ is infinitely small. The Riesz space $A$ is \emph{Archimedean} if $0$ is the only infinitely small element in $A$. Equivalently, $A$ is Archimedean if it satisfies the following (countably) infinitary rule:

\begin{figure}[h!]
\begin{mdframed}
\begin{center}
$\infer[(\mathbb{A})]
                        {a = 0 }            {  |a|\leq | b| \ \ \ \  2|a|\leq |b| \ \ \ \  3|a|\leq |b| \ \ \ \ \dots \ \ \ \  n|a|\leq |b|  \ \ \ \ \dots  }$
\end{center}
\end{mdframed}
\caption{Archimedean Rule}
\label{archimedean_rule:fig}
\end{figure}
\end{definition}

All the Riesz spaces in Examples \ref{background:example1} and \ref{example2_subalgebra} are Archimedean. Not all Riesz spaces are Archimedean, however, as the following example shows.
\begin{example}\label{background:example2} The vector space $\mathbb{R}^2$ with the lexicographic order, defined as $(x_1,y_1)\leq (x_2,y_2) \Leftrightarrow$ either $x_1< x_2$ or $x_1=x_2$ and $y_1\leq y_2$, is not Archimedean. For instance, $(0,1)$ is infinitely small with respect to $(1,0)$.
\end{example}

As usual in universal algebra, a homomorphism between Riesz spaces is a function $f\!:\!A\!\rightarrow\! B$ preserving all operations. Therefore a \emph{Riesz homomorphism} is a linear map preserving finite meets and joins. 
\begin{definition}[Ideals and Maximal Ideals]\label{ideals_def}
A subset $J\subseteq A$ of a Riesz space $A$ is an \emph{ideal} if it is the kernel of a homomorphism $f:A\rightarrow B$, in the sense that $J=f^{-1}(\{0\})=\{ a \mid f(a)=0\}$, for some Riesz space $B$.  The sets $\{0\}$ and $A$ itself are trivially ideals. All other ideals are called \emph{proper}.
Ideals in $A$ can be partially ordered by inclusion. An ideal $J\subseteq A$ is called $\emph{maximal}$ if it is a proper ideal and there is no larger proper ideal $J\subsetneq J^\prime$.
\end{definition}
The following alternative characterization of ideals (see, e.g., Section 3.9 of \cite{vulikh}) is often much more simple to deal with.
\begin{proposition}
Let $A$ be a Riesz space. A subset $J\subseteq A$ is an ideal if and only $J$ is a Riesz subspace of $A$ (i.e., closed under all operations) and furthermore, for all $a\in J$ and $b\in A$, if $|a|\in J$ and $|b|\leq |a|$ then $b\in J$.  
\end{proposition}

\begin{example}\label{infinitesimals:are:ideals}
Given any Riesz space $A$, the collection of infinitely small elements (see Definition \ref{archimedean:def}) is an ideal of $A$. 
\end{example}

We now introduce the important concept of a strong unit.
\begin{definition}[Strong Unit]
An element $u\in A$ is called a \emph{strong unit} if it is positive (\ie, $u\in A^+$) and for every $a\in A$ there exists $n\in\mathbb{N}$ such that $|a|\leq n (u)$.
\end{definition}

\begin{example}\label{background:example3}
The real line $\mathbb{R}$ has $1$ as strong unit. The space $\mathbb{R}^\mathbb{N}$ does not have a strong unit. Its subspace $\ell^\infty(\mathbb{N})$ consisting of bounded functions has $\one_\mathbb{N}$ (the constant $n\mapsto 1$ function) as strong unit. Similarly, let $X$ be a compact topological space and $C(X)$ the Riesz space of continuous functions into $\mathbb{R}$. Since $X$ is compact, any function $f\in C(X)$ is bounded and therefore $\one_X$ is a strong unit of $C(X)$. 
\end{example}

We now introduce a notion of convergence in Riesz spaces which plays an important role in the duality theory of Riesz spaces.

\begin{definition}[$u$-convergence and $u$-uniform Cauchy sequences]
Let $A$ be a Riesz space and $u$ be a positive element $u\!\geq\! 0$. We say that a sequence $(a_n)_{n\in\mathbb{N}}$ \emph{converges $u$-uniformly} to $b$, written $(a_n)\rightarrow_{u}b$, if for every positive real $\epsilon>0$ there exists a natural number $N_\epsilon$ such that $|b-a_n|\leq \epsilon u$, for all $n>N_\epsilon$. We say that $(a_n)_{n\in\mathbb{N}}$ is a \emph{$u$-uniform Cauchy sequence} if for every $\epsilon >0$ there exists a number $N_\epsilon$ such that $|a_i-a_j|\leq \epsilon u$, for all $i,j>N_\epsilon$. 
\end{definition}
Clearly, if $(a_n)\!\rightarrow_{u}\!b$ then $(a_n)$ is a $u$-uniform Cauchy sequence.

\begin{definition}[uniform completeness]\label{def_uniformly_complete}
A Riesz space $A$ is \emph{$u$-uniformly complete} if for every $u$-uniform Cauchy sequence $(a_n)$ there exists $b\!\in\! A$ such that $(a_n)\rightarrow_{u}b$. It is \emph{uniformly complete} if it is $u$-uniformly complete, for all $u \!\in\! A^+$. 
\end{definition}

We now state important properties related to uniform completeness of Archimedean Riesz spaces with strong unit.

\begin{theorem}[45.5 in \cite{Luxemburg}]\label{theorem_unital_convergence}
If  $A$ is Archimedean and has strong unit $u$, then $A$ is uniformly complete if and only if it is $u$-uniform complete.
\end{theorem}


\begin{example}
The Riesz space $\mathbb{R}$ has $1$ as strong unit. It is a $1$-uniformly complete space as the notion of $1$-uniform Cauchy sequence coincides with the usual notion of Cauchy sequence of reals. Therefore $\mathbb{R}$ is uniformly complete. Let $X$ be a compact Hausdorff space, $C(X)$ the Riesz space of continuous functions $f\!:\!X\!\rightarrow\!\mathbb{R}$ and $\one_{X}\!\in\! C(X)$ the constant function $x\mapsto 1$. Then $C(X)$ is $\one_X$-uniformly complete (\cite[Example 27.7, Theorem 43.1]{Luxemburg}). Once again, $C(X)$ is uniformly complete because $\one_X$ is a strong unit.
\end{example}

\begin{theorem}[Theorem 43.1 in \cite{Luxemburg}]\label{normed_space}
Let $A$ be Archimedean with strong unit $u\in A$. Let $\| \_ \| : A\rightarrow\mathbb{R}_{\geq 0}$ be defined as:
\begin{equation}
\| a \| = \inf\{ r\in\mathbb{R} \mid  |a|\leq ru \}
\end{equation}
Then $\|\_\|$ is a norm on $A$, i.e., $\|0\|\!=\!0$, $\| a+b\|\leq \|a\|+\|b\|$ and $\| r a\|= |r|\cdot \|a\|$, for all $a,b\!\in\!A$ and $r\!\in\!\mathbb{R}$. 
\end{theorem}
As a consequence, each Archimedean Riesz space with strong unit is a normed vector space and therefore can be endowed with the metric $d_A\!:\!A^2\rightarrow\mathbb{R}_{\geq 0}$ defined as $d_A(a,b)\!=\! \| a- b  \|$. Accordingly, we say that a Riesz homomorphism $f\!:\!A\!\rightarrow\! B$ between Archimedean spaces with strong units is \emph{continuous} (resp. is an \emph{isometry}) if it is continuous (resp. distance preserving) with respect to the metrics of $A$ and $B$.

Importantly, on Archimedean spaces with strong unit, the notion of uniform convergence and convergence in the norm (i.e., in the metric $d$) coincide.
\begin{theorem}[(Theorem 43.1 in \cite{Luxemburg})]
Let $A$ be an Archimedean Riesz space with strong unit $u$. A sequence $(a_n)$ converges $u$-uniformly to $b$ if and only if $(a_n)$ converges in norm to $b$. The space $A$ is uniformly complete if and only if it is complete as a metric space.
\end{theorem}


\subsection{Riesz Spaces with a distinguished positive element}
\label{sec:riesz:back2}

$ \ $ \\

It is now convenient to extend the language of Riesz spaces with a new constant symbol $u$ for a positive element.

\begin{definition}\label{riesz_positive_def}
A Riesz space with \emph{distinguished positive element} $u$ is a pair $(A,u)$ where $A$ is a Riesz space and $u\geq 0$.  A morphism between $(A,u)$ and $(B,v)$ is a Riesz homomorphism $f:A\rightarrow B$ such that $f(u)=v$. If $u$ is a strong unit in $A$ we say that $(A,u)$ is \emph{unital}.
\end{definition}

When confusion might arise, we will stress the fact that a homomorphism $f\!:\!(A,u)\!\rightarrow\!(B,v)$ preserves the distinguished positive elements (\ie, $f(u)\!=\!v$) by saying that $f$ is a \emph{unital (Riesz) homomorphism}. We write $\CRiesz^{u}$ for the category having Riesz spaces $(A,u)$ with a distinguished positive element as objects and unital homomorphisms as morphisms. We write $\URiesz$ for the subcategory of $\CRiesz^{u}$ whose objects are unital Riesz spaces.

\begin{example}\label{background:example_unital}
The basic example is the real line $(\mathbb{R},1)$. Since $1$ is a strong unit, this is in fact a unital Riesz space. Furthermore it follows easily from the result mentioned in Example \ref{background:example1} that $(\mathbb{R},1)$ generates the variety $\CRiesz^{u}$.
\end{example}

The following theorem (see, e.g., \cite[Thm 27.3-4]{Luxemburg}) expresses a key property of unital Riesz spaces. 
\begin{theorem}
\label{MaxIdealTheorem}
Let $(A,u)$ be a unital Riesz space. Then, for every unital homomorphism $f\!:\!(A,u)\!\rightarrow\!(\mathbb{R},1)$, the ideal $f^{-1}(0)$ is maximal. Conversely, every maximal ideal $J$ in $(A,u)$ is of the form $f_J^{-1}(0)$ for a unique unital Riesz homomorphism $f_J\!:\!(A,u)\!\rightarrow\!(\mathbb{R},1)$.
\end{theorem}
Hence there is a one-to-one correspondence between maximal ideals in unital Riesz spaces $(A,u)$ and homomorphisms into $(\mathbb{R},1)$ preserving the unit. Observe, once again (cf. Examples \ref{background:example1} and \ref{background:example_unital}), how the Riesz space $(\mathbb{R},1)$ plays, in the theory of unital Riesz spaces, a role similar to two element Boolean algebra $\{0,1\}$, in the theory of Boolean algebras. 

We say that a unital Riesz space $(A,u)$ is Archimedean if $A$ is Archimedean. We write $\AURiesz$ for the category of Archimedean unital Riesz spaces with unital Riesz homomorphisms. We write $\CAURiesz$ for the category of Archimedean and uniformly   complete unital Riesz spaces with unital Riesz homomorphisms.
$$
\CAURiesz \hookrightarrow \AURiesz \hookrightarrow \URiesz \hookrightarrow \CRiesz^{u}
$$

\begin{example}\label{CAU_example}
Let $X$ be a compact Hausdorff space. Let $\one_{X}$ be the constant $(x\mapsto 1)\in C(X)$ function. Then $(C(X),\one_{X})$ is an Archimedean unital  and uniformly complete  Riesz space \cite[Example 27.7, Theorem 43.1]{Luxemburg}.
\end{example}

The following results describe the property of being Archimedean for unital Riesz spaces. 
\begin{theorem}\label{unit_archimedean_thm3}
Let $(A,u)$ be a unital Riesz space. An element $a\in A$ is infinitely small if and only if $n|a|\leq u$, for all $n\in\mathbb{N}$. This means that the Archimedean rule (cf. Definition \ref{archimedean:def}) can be equivalently reformulated as follows:
\begin{center}
$\infer[\mathbb{A}]
                        {a = 0 }            {  |a|\leq u \ \ \ \  2|a|\leq u \ \ \ \  3|a|\leq u \ \ \ \ \dots \ \ \ \  n|a|\leq u \ \ \ \ \dots  }$
\end{center}
Furthermore, $A$ is Archimedean if and only if  for every $a\!\neq\! 0$  there exists a unital Riesz homomorphism $f\!:\!(A,u)\!\rightarrow\!(\mathbb{R},1)$ such that $f(a)\!\neq\! 0$.
\end{theorem}

\begin{corollary}\label{hom:map:inf2inf}
Let $(A,u)$ and $(B,v)$ be unital Riesz spaces and $f:A\rightarrow B$ a unital Riesz homomorphism. If $a\in A$ is infinitely small then $f(a)\in B$ is also infinitely small.
\end{corollary}
\begin{proof}
By assumption we have that for all $n\in\mathbb{N}$ the inequality $n|a|\leq u$ holds. Since $f$ is a homomorphism we have that $f$ is monotone (and thus $f(n|a|)\leq f(u)$) and that $f(n|a|) = n|f(a)|$. Furthermore, since $f$ is unital, we have that $f(u)=v$. Therefore we have:
$$
f(n|a|)  =  n| f(a) | \leq v
$$
for all $n\in\mathbb{N}$ which means that $f(a)$ is infinitely small.
\end{proof}

\subsection{Yosida's Theorem and Duality Theory of Riesz Spaces}\label{duality_section_background}

$ \ $ \\

In this section we assume familiarity with the basic notions from category theory regarding  equivalences of categories and adjunctions. A standard reference is \cite{maclane}.

The celebrated Stone duality theorem states that any Boolean algebra $B$ is isomorphic to the Boolean algebra of clopen sets (or equivalently continuous functions $f\!:\!X\!\rightarrow\!\{0,1\}$ where $\{0,1\}$ is given the discrete topology) of a unique (up to homeomorphism) \emph{Stone space}, \ie, a compact Hausdorff and zero--dimensional topological space $X$.  Here $X$ is the collection $\Spec(B)$ of maximal (Boolean) ideals in $B$ endowed with the \emph{hull--kernel} topology. In fact this correspondence can be made into a \emph{categorical equivalence} between $\cat{Stone}$ and $\cat{Bool}\op$.

A similar representation theorem, due to Yosida \cite{yosida}, states that every uniformly complete, unitary and Archimedean Riesz space $(A,u)$ is isomorphic to $(C(X),\one_X)$, the Riesz space of all continuous functions $f\!:\!X\!\rightarrow\!\mathbb{R}$, of a unique (up to homeomorphism) compact Hausdorff space $X$. This correspondence can be made into a categorical equivalence
\begin{equation}
\label{YosidaEquivDisplayEqn}
\CHaus \simeq \CAURiesz\op
\end{equation}
see, e.g., \cite{westerbaan2016} for a detailed proof. 
In fact Yosida proved a more general result which can be  conveniently formulated as an adjunction between $\CHaus$ and $\AURiesz\op$ which restricts to the equivalence (\ref{YosidaEquivDisplayEqn}) on the subcategory $\CAURiesz\op$.
In the rest of this section we describe it as a unit--counit adjunction $(\eta,\epsilon): C \dashv \Spec$ consisting of two functors:

\hspace{1cm}$C\! : \!\CHaus \!\rightarrow\! \CAURiesz\op\! \hookrightarrow\!  \AURiesz\op$

\hspace{1cm}$\Spec\! :\! \AURiesz\op \!\rightarrow\! \CHaus$

\noindent and two natural transformations:

\hspace{1cm}$\eta: \id_{\CHaus} \Rightarrow \Spec\circ C$

\hspace{1cm}$\epsilon: C\circ \Spec \Rightarrow \id_{\AURiesz\op}$

\noindent 
called unit and counit, respectively.

We first define the functor  $C \!:\! \CHaus \rightarrow \CAURiesz\op$. 

On objects, for a compact Hausdorff space $X$, we define $C(X)$ as the set of continuous real-valued functions on $X$, equipped with the Riesz space operations defined pointwise from those on $\R$ (see Example \ref{example2_subalgebra}) and strong unit $\one_X$ ($x\mapsto 1$). As discussed earlier (see Example \ref{CAU_example}) this is indeed a uniformly complete Archimedean and unital Riesz space. On continuous maps $f \!:\! X\! \rightarrow \! Y$, we define
$C(f)(b) = b \circ f$, for all $b \!\in\! C(Y)$. This is easily proven to be a unital Riesz space morphism by the fact that the Riesz space operations are defined pointwise. 

We now turn our attention to the description of the functor $\Spec : \AURiesz\op \rightarrow \CHaus$. 

As in the Stone duality theorem, on objects $(A,u)$ in $\AURiesz$, the functor $\Spec(A)$ is defined as the \emph{spectrum of $A$}, \ie, the collection of all maximal ideals of $A$ (see Definition \ref{ideals_def}) equipped with the hull--kernel topology which can be defined as follows. 
A subset $X\subseteq \Spec(A)$ is closed in the hull--kernel topology if and only if there exists a (not necessarily maximal) ideal $I\subseteq A$ such that $X=\hull(I)$ where $
\hull(I) = \{ J \in \Spec(A) \mid I \subseteq J \}$. 
See, e.g., \cite[Theorem 36.4 (ii)]{Luxemburg} for a proof that $\Spec(A)$ is indeed a compact Hausdorff space.
On maps, for a unital morphism $f \!:\! (A,u_A) \!\rightarrow\! (B, u_B)$ we define, for every $J \in \Spec(B)$,  $\Spec(f)(J) = f^{-1}(J)$.

We now turn our attention to the description of the unit $\eta\!:\! \id_{\CHaus} \!\Rightarrow \!\Spec\circ C$. This is a collection of maps $\{ \eta_{X} :  X\rightarrow \Spec( C(X))\}$ indexed by compact Hausdorff spaces. For a fixed compact Hausdorff space $X$ and $x\!\in\! X$ we can define the map $\delta_x \!:\! C(X) \!\rightarrow\! \R$ as $\delta_x(f) \!=\! f(x)$ which is easily seen to be a unital Riesz homomorphism. Therefore, by Theorem \ref{MaxIdealTheorem} the set $N_x \!=\! \delta_x^{-1}(0)$ is a maximal ideal in $C(X)$, i.e., $N_x\!\in\! \Spec(C(X))$. We then define $\eta_X$ as $\eta_X(x)\!=\!N_x$.

Lastly, we now proceed with the definition of the counit $\epsilon: C\circ \Spec \Rightarrow \id_{\AURiesz\op}$. This is a collection of morphisms $\{ \epsilon_A:  C(\Spec(A)) \rightarrow A\}$ in $\AURiesz\op$, or equivalently a collection of morphisms $\{ \epsilon_A:  A \rightarrow C(\Spec(A))\}$ in $\AURiesz$, indexed by unital and Archimedean Riesz spaces $(A,u_A)$.  For a fixed such $(A,u)$ and $a\!\in\! A$ we can define a function $\hat{a}\!:\!\Spec(A)\!\rightarrow\!\mathbb{R}$ as $\hat{a}(J) \!=\! f_J(a)$, where $f_J$ is the homomorphism from Theorem \ref{MaxIdealTheorem}. That is (see \cite[Thm 27.3-4]{Luxemburg}) the value $\hat{a}(J)$ is defined as the unique real number $r$ such that  $r u_A - a \!\in\! J$. The map $\hat{a}$ is continuous, i.e., $\hat{a}\in C(\Spec(A))$. We then define $\epsilon_A$ as $\epsilon_A(a)=\hat{a}$. 

The statement of Yosida's theorem can then be formulated by the following two theorems (see \cite[Theorems 1--3]{yosida},  also \cite[Theorems 45.3 and 45.4]{Luxemburg} and \cite{westerbaan2016}).
\begin{theorem}
\label{YosidaAdjunctionTheorem}
Both $C$ and $\Spec$ are functors. Both $\eta$ and $\epsilon$ are natural transformations. The quadruple $(\eta,\epsilon): C \dashv \Spec$ is a unit-counit adjunction. The counit $\epsilon_A$ is an isometric isomorphism between $A$ and its image in $C(\Spec(A))$. 
\end{theorem}
\begin{proof}
The fact that $C$ is indeed a functor follows from  elementary properties of composition of functions and identity maps. We now show $\Spec$ is a functor. Let $f : (A,u_A) \rightarrow (B, u_B)$ be a unital Riesz homomorphism and $J \subseteq B$ a maximal ideal. By Theorem \ref{MaxIdealTheorem} there is a unital Riesz morphism $\phi_J : B \rightarrow \R$ such that $\phi_J^{-1}(0) = J$. The composite $\phi_J \circ f$ is a unital Riesz homomorphism $A \rightarrow \R$, so
\[
\Spec(f)(J) = f^{-1}(J) = (\phi_J \circ f)^{-1}(0)
\]
is a maximal ideal in $A$. This shows $\Spec(f)$ is a function $\Spec(B) \rightarrow \Spec(A)$. We show it is continuous by showing that the preimage of a closed set is closed. Any closed set in $\Spec(A)$ is $\hull(I)$ for some ideal $I \subseteq A$. By \cite[Theorem 59.2 (iii)]{Luxemburg}, $f(I)$ is also an ideal. By elementary manipulations of the definitions, $\Spec(f)^{-1}(\hull(I)) = \hull(f(I))$, which, as we started with an arbitrary closed set, proves the continuity of $\Spec(f)$. By basic properties of the preimage mapping, $\Spec$ preserves identity maps and reverses composition, and is therefore a contravariant functor, as required. 

We now consider the unit. In \cite[Theorem 4]{yosida} Yosida shows that the mapping $N_\blank : X \rightarrow \Spec(C(X))$ is a homeomorphism onto its image and has dense image. The compactness of $X$ then implies that the image of $N_\blank$ is closed, and therefore is all of $\Spec(C(X))$, \ie, $N_\blank$ is a homeomorphism $X \rightarrow \Spec(C(X))$. 

The proof of naturality, \ie, that $N_{f(x)} = \Spec(C(f))(N_x)$ for all continuous maps $f : X \rightarrow Y$ and for all $x \in X$, is done by expanding the definitions on each side, so is omitted. As a result, $\Spec \circ C \cong \Id_\CHaus$.

For the counit, Yosida shows that $\hat{a} \in C(\Spec(A))$, and $\hat{\blank}$ is a unital Riesz space homomorphism with norm-dense image \cite[Theorems 1--2]{yosida} (see also \cite[Theorem 45.3]{Luxemburg}). 

The naturality of $\hat{\blank}$, \ie, that for all $f : A \rightarrow B$ a unital Riesz homomorphism, $a \in A$, $J \in \Spec(B)$ we have $\widehat{f(a)}(J) = C(\Spec(f))(\hat{a})(J)$ reduces to showing that $\hat{a}(f^{-1}(J)) \cdot u_A - f(a) \in I$, which is easily done using the linearity and unitality of $f$ and the definition of $\hat{\blank}$. 

To show that $\Spec$ is a right adjoint to $C$, we only need to prove that the following diagrams commute:
\[
\xymatrix{
C(X) & C(\Spec(C(X))) \ar[l]_-{C(\eta_X)} & \Spec(A) \ar[r]^-{\eta_{\Spec(A)}} \ar[rd]_\id & \Spec(C(\Spec(A))) \ar[d]^{\Spec(\epsilon_A)} \\
 & C(X) \ar[ul]^\id \ar[u]_{\epsilon_{C(X)}}  & & \Spec(A),
}
\]
where $X$ is a compact Hausdorff space and $A$ a unital Archimedean Riesz space.

For the first diagram, we want to show that if $a \in C(X)$ and $x \in X$, we have $C(\eta_X)(\epsilon_{C(X)}(a))(x) = a(x)$. Expanding the definitions, this is equivalent to showing
\begin{equation}
\label{UsefulHatEquation}
\hat{a}(N_x) = a(x).
\end{equation}
By the definition of $\hat{\blank}$, we have $\hat{a}(N_x) \cdot \one_X - a \in N_x$. Applying the definition of $N_x$ and elementary algebra then gives us the result. 

For the second diagram, we want to show that for each $J \in \Spec(A)$ and $a \in A$ that $a \in \Spec(\epsilon_A)(\eta_{\Spec(A)}(J)) \Leftrightarrow a \in J$. This can be proved simply by expanding the definitions. 

We have therefore shown that $C \dashv \Spec$, \ie, $\Spec$ is a right adjoint to $C$.
\end{proof}

\begin{theorem}
\label{YosidaDualityTheorem}
When restricted to $\CAURiesz\op$, the adjunction becomes an equivalence of categories. An object $(A,u)$ of $\AURiesz$ is uniformly complete (i.e., it belongs to $\CAURiesz$) if and only if  $\epsilon_A$ is  a Riesz isomorphism.
\end{theorem}
\begin{proof}
Yosida shows that $\hat{\blank}$ is a unital Riesz space isomorphism iff $A$ is uniformly complete in \cite[Theorem 3]{yosida} (see also \cite[Theorem 45.4]{Luxemburg}). As the (norm) unit ball of $A$ is exactly the inverse image of the unit ball of $C(\Spec(A))$, we have that the embedding $\hat{\blank}$ is also an isometry, and therefore $C(\Spec(A))$ is isomorphic to the Banach space completion of $A$. 

We already saw in Theorem \ref{YosidaAdjunctionTheorem} that Yosida proved that $\eta_X$ is always an isomorphism for $X$ a compact Hausdorff space. Therefore $(C,\Spec,\eta,\epsilon)$ is an adjoint equivalence when restricted to $\CAURiesz$ \cite[\S IV.4]{maclane}.
\end{proof}

The functor $C \circ \Spec\!:\!\ \AURiesz\op\rightarrow\! \CAURiesz\op$ maps (not necessarily uniformly complete) Archimedean unital Riesz spaces to uniformly complete ones. In fact, Yosida showed that
$A$ embeds densely 
 in  $C(\Spec(A))$. Therefore $C(\Spec(A))$ is isomorphic to the completion of $A$ in its norm (as defined in the statement of Theorem \ref{normed_space}).

\begin{definition}\label{uniform_completion}
The uniform Archimedean and unital Riesz space $C(\Spec(A))$ is called the \emph{uniform completion of $A$} and is simply denoted by $\hat{A}$. We always identify $A$ with the (isomorphic) dense sub-Riesz space $\epsilon_A(A)$ of $\hat{A}$.
\end{definition}

\begin{proposition}
\label{DensenessCor}
For every $A\!\in\! \AURiesz$, the two spaces $\Spec(A)$ and $\Spec(\hat{A})$ are homeomorphic. Furthermore, for every $B\!\in\! \AURiesz$ and unital homomorphism $f \!:\! A \rightarrow B$ there exists a unique unital Riesz homomorphism $\hat{f}:\hat{A}\rightarrow \hat{B}$ extending $f$.\end{proposition}


\subsection{Dedekind Complete Riesz Spaces}\label{sec:dedekind}

$ \ $ \\

We conclude this section by discussing Dedekind--complete Riesz spaces. We refer to \cite[\S 4]{vulikh} for a detailed introduction. Dedekind--complete Riesz spaces, due to their order--completeness properties, play a role when fixed--point extensions of Riesz modal logic, based on the Knaster--Tarski theorem, are considered (see, e.g., \cite{MIO2012b,MioSimpsonFI2017,MIO2014a,Mio18}).

\begin{definition}
Let $(L,\sqcup,\sqcap)$ be a lattice. The lattice $L$ is \emph{complete} if for every subset $A\subseteq L$ there exist in $L$ both a least upper bound ($\bigsqcup\! A$) and a greatest lower bound ($\bigsqcap\! A$). The lattice $L$ is \emph{Dedekind--complete} if, for every \emph{bounded subset} $A\subseteq L$,  there exist in $L$ both $\bigsqcup\! A$ and $\bigsqcap\! A$.
\end{definition}

It follows from this definition that every complete lattice is Dedekind--complete, but the converse is not true. The real numbers $\mathbb{R}$ is an example of  Dedekind--complete lattice which is not complete since suprema and infima of unbounded sets do not exist in $\mathbb{R}$.


\begin{definition}[Dedekind--complete Riesz space]\label{background:dedekind}
The Riesz space $A$ is called \emph{Dedekind--complete} if its underlying lattice is Dedekind--complete. \end{definition}

The following is an important property of Dedekind--complete Riesz spaces (see, e.g., Theorem 25.1 of \cite{Luxemburg}).

\begin{theorem}
Every Dedekind--complete Riesz space $A$ is Archimedean and uniformly complete.
\end{theorem}

In particular, if the Dedekind--complete Riesz space $A$ has a strong unit $u_A$, the Yosida duality described in Section \ref{duality_section_background}, can be applied. Therefore every Dedekind--complete Riesz space $A$ with strong unit $u_A$ is isomorphic to the space of real--valued functions $C(X)$ of a unique (up to homeomorphism) compact Hausdorff space $X$.  To the order--theoretic property of Dedekind--completeness corresponds, via Yosida duality, a topological property of the dual space: $X$ is \emph{extremally disconnected}.

\begin{definition}\label{extremally:spaces}
A topological space is \emph{extremally disconnected} if the closure of every open set is clopen. An extremally disconnected space that is also compact and Hausdorff is called \emph{Stonean}. 
\end{definition} 
It is well--known that a Stonean space $X$ is (up to homeomorphism) the Stone dual of a unique (up to isomorphism) complete Boolean algebra. Hence we get (see, e.g., Chapter IV of \cite{vulikh}) the following:
\begin{proposition}
Let $A$ be a Dedekind--complete Riesz space with strong unit $u_A$. Then $A$ is isomorphic to $C(X)$ for a unique (up to homeomorphism) Stonean space  $X$.
\end{proposition} 

The following theorem, due to Yudin (see \cite[Thm IV.11.1]{vulikh}), is the Riesz space equivalent of the Dedekind-MacNeille completion theorem in the theory of lattices. It states that it is possible to embed an arbitrary unital Archimedean Riesz space $A$ into a \emph{essentially minimal} Dedekind--complete unital Riesz space, called the \emph{Dedekind--completion of $A$}, preserving all suprema and infima existing in $A$.

\begin{theorem}[Dedekind completion]\label{dedekind-completion-1}
For every Archimedean unital Riesz space $R$ there exists a Dedekind complete Archimedean and unital space $\overline{R}$, called the Dedekind completion of $R$, such that:
\begin{enumerate}
\item $R$ embeds in $\overline{R}$, so we can just write $R\subseteq \overline{R}$,
\item $R$ is order--dense in $\overline{R}$, i.e., for every $f < h\in \overline{R}$ there exists $g\in R$ such that $f<g<h$,
\item existing suprema and infima in $R$ are preserved in $\overline{R}$. This means that for every $A\subseteq R$ and $f=\bigsqcup A$ (sup existing and taken in $R$) then $f=\bigsqcup A$ in $\overline{R}$ too.
\item $\overline{R}$ is the smallest Dedekind complete space satisfying the properties above.
\end{enumerate}
\end{theorem}


\section{Riesz Modal Logic, Syntax and Transition Semantics}
\label{logic_section}
In this section we formally introduce Riesz modal logic for Markov processes.

\begin{definition}[Syntax]\label{syntax_formulas}
The set of formulas $\texttt{Form}$ is generated by the following grammar:
$$
\phi, \psi ::= 0 \mid 1 \mid r\phi \mid \phi + \psi \mid \phi \sqcup \psi \mid \phi \sqcap \psi \mid \Diamond \phi \ \ \ \ \ \ \ \ \ \textnormal{where } r\in\mathbb{R}.
$$
\end{definition}

The semantics of a formula $\phi$, interpreted over a Markov process $\alpha:X\rightarrow \Rdnl(X)$ (see Definition \ref{markov_process_coalgebra}), is a continuous function $\sem{\phi}_{\alpha}:X\rightarrow\mathbb{R}$ defined as follows.
\begin{definition}[Semantics]\label{sematics:riesz:logic}
Let $\alpha:X\rightarrow \Rdnl(X)$ be a Markov process. The semantics (or interpretation) of a formula $\phi$ relative to the Markov process $\alpha$ is the continuous function $\sem{\phi}_{\alpha}\in C(X)$ defined by induction on $\phi$ as follows:
\begin{center}
$\sem{0}_{\alpha} (x) = 0 \ \ \ \ \ \sem{1}_{\alpha} (x) = 1$

$\ $\\

$ \sem{r\phi}_{\alpha}(x) = r \cdot\big(\sem{\phi}_{\alpha}(x)\big) \ \ \ \ \ \sem{\phi+\psi}_{\alpha}(x) = \sem{\phi}_{\alpha}(x) + \sem{\psi}_{\alpha}(x)$

$\ $\\

$\sem{\phi\sqcup\psi}_{\alpha}(x) = \max\big\{ \sem{\phi}_{\alpha}(x) ,  \sem{\psi}_{\alpha}(x) \big\}$ 

$\ $\\

$\sem{\phi\sqcap\psi}_{\alpha}(x) = \min\big\{\sem{\phi}_{\alpha}(x) , \sem{\psi}_{\alpha}(x)\big\}$

$\ $\\

$\sem{\Diamond\phi}_{\alpha}(x) =\displaystyle \int_{X}  \sem{\phi}_{\alpha} \diff \alpha(x) = \E_{\alpha(x)}(\sem{\phi}_{\alpha})$
\end{center}
\end{definition}

Hence $\sem{0}_{\alpha}$ and $\sem{1}_{\alpha}$ are the constant functions $\zero_X$ ($x\!\mapsto\! 0$) and $\one_X$ ($x\!\mapsto\! 1$), respectively. The connectives $\{r(\_), +, \sqcup, \sqcap\}$ correspond to the real vector space and lattice operations of $\mathbb{R}$ lifted to $C(X)$ pointwise (see examples and \ref{example2_subalgebra} and \ref{CAU_example}). The semantics of the formula $\Diamond\phi$ is the function that assigns to $x$ the expected value of $\sem{\phi}_{\alpha}$ with respect to the subprobability measure $\alpha(x)$.

\begin{example}\label{first:example:logic}
Consider the Markov process $\alpha:X\rightarrow \Rdnl(X)$ of Example \ref{logic:example1a}, having state space $X=\{x_1,x_2\}$ endowed with the discrete topology. The semantics $\sem{\phi}_\alpha$ of a formula $\phi$ is thus a real--valued function $ \sem{\phi}_\alpha: X\rightarrow\mathbb{R}$. Furthermore, since the state space is discrete, the interpretation of $\Diamond$ can be simply expressed as a weighted sum:

$$\sem{\Diamond\phi}_{\alpha}(x) =\displaystyle \sum_{y \in X} \big(  \sem{\phi}_{\alpha}(y) \cdot d_x(y) \big) $$
\noindent
where $d_x = \alpha(x)$, i.e., $d_x\in \mathcal{D}^{\leq 1}(X)$ is the subprobability distribution over $X$ assigned to $x$ by the transition function $\alpha$.

Now consider the formula $\Diamond1$. The formula $\Diamond 1$ can be understood as mapping each state $x\in X$ to the total mass of the probability distribution $\alpha(x)$. Therefore in this example we have $\sem{\Diamond 1}_\alpha (x_1) = \frac{5}{6}$ and $\sem{\Diamond 1}_\alpha (x_1) = \frac{1}{3}$.

Consider next the formula $-(\Diamond 1)+1$. Simple calculations show that $\sem{-(\Diamond 1)+1}_{\alpha}(x_1) = \frac{1}{6}$ and $\sem{-(\Diamond 1)+1}_{\alpha}(x_2) = \frac{2}{3}$. The formula  $-(\Diamond 1)+1$ assigns to each state the probability of terminating the computation at that state.

The semantics of the formula $\Diamond(\Diamond 1)$ can be calculated as follows:

$$\sem{\Diamond\Diamond 1}_\alpha(x_1) = \frac{1}{3}\big(  \sem{\Diamond 1}(x_1) \big)   + \frac{1}{2}\big(  \sem{\Diamond 1}(x_2)  \big) =  \frac{1}{3}\cdot\frac{5}{6} +  \frac{1}{2}\cdot\frac{1}{3} =  \frac{8}{18}$$
\noindent
and
$$\sem{\Diamond\Diamond 1}_\alpha(x_2) = \frac{1}{3}\big(  \sem{\Diamond 1}(x_1) \big)   + 0 \big(  \sem{\Diamond 1}(x_2)  \big) =  \frac{1}{3}\cdot\frac{5}{6} + 0  =  \frac{5}{18}$$
\noindent
Indeed the meaning of $\Diamond\Diamond 1$ is the function that assigns to each state the probability of making two computational steps (without halting) starting from that state. 
\end{example}

\begin{example}
Consider the Markov process $\alpha:X\rightarrow \Rdnl(X)$ having state space $\{x,y,z\}$ depicted as follows:

$$
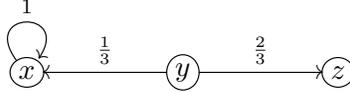
\SelectTips{cm}{}
	\xymatrix @=20pt {
		\nodeC{x} \ar@{->}@(ul,ur)^{1}   & & \nodeC{y}  \ar@{->}[ll]_{\frac{1}{3}} \ar@{->}[rr]^{\frac{2}{3}} & &  \nodeC{z}   	}
$$
Consider the two formulas $\phi_1 = \Diamond( \psi_1 \sqcup \psi_2)$ and $\phi_2 = \Diamond(\psi_1) \sqcup \Diamond(\psi_2)$, where $\psi_1=\Diamond 1$ and  $\psi_2=(1-\Diamond1)$ and observe that $\sem{\psi_1}(x)=1= \sem{\psi_2}(z)$ and  $\sem{\psi_1}(z)=0= \sem{\psi_2}(x)$. 

The semantics of the formulas $\phi_1$ and $\phi_2$ at the state $y$ is calculated as follows:

\begin{center}
\begin{tabular}{l l l}
$\sem{\phi_1}_{\alpha} (y)$ & = & $ \frac{1}{3}\big( \sem{\psi_1 \sqcup \psi_2}_{\alpha}(x)   \big) + \frac{2}{3}\big(   \sem{\psi_1 \sqcup \psi_2}_{\alpha}(z) \big)$ \\
 & = & $ \frac{1}{3}\big( 1 \sqcup 0 ) + \frac{2}{3}\big( 0 \sqcup 1\big) $\\
  & = & $\frac{1}{3} + \frac{2}{3} $ \\
   & = & $1$
   \end{tabular} 
\end{center}
\noindent
and
\begin{center}
\begin{tabular}{l l l}
$\sem{\phi_2}_{\alpha} (y)$ & = & $\big(\frac{1}{3} \sem{\psi_1}_{\alpha}(x) + \frac{2}{3} \sem{\psi_1}_{\alpha}(z)\big) \sqcup  \big(\frac{1}{3} \sem{\psi_2}_{\alpha}(x) + \frac{2}{3} \sem{\psi_2}_{\alpha}(z)\big) $ \\
 & = & $ (\frac{1}{3} + 0) \sqcup (0 + \frac{2}{3}) $\\
  & = & $ \frac{2}{3}$. 
   \end{tabular} 
\end{center}
Hence this example shows that the distributivity law: 
$$
\Diamond(x \sqcup y ) = \Diamond(x) \sqcup \Diamond(y)
$$
generally fails in Riesz modal logic.
\end{example}

\begin{example}
\label{example:logic:3}
Consider the Markov process $\alpha$ of Example \ref{example:mp:2}, having state space $X=[0,1]$. The semantics of the formula $\Diamond 1$ is the continuous function $\sem{\Diamond 1}_\alpha: [0,1] \rightarrow \mathbb{R}$. Again, the formula $\Diamond 1$ maps each state $x$ to the total mass of the subprobability measure $\alpha(x)$. Hence, by expanding the definitions, it corresponds to the identity function:
$$
\sem{\Diamond 1}_\alpha (x) =  \displaystyle \int_{[0,1]} 1 \diff \alpha(x) =  \displaystyle \int_{[0,1] } 1 \cdot x \diff \delta_{x} = x
$$
Similarly,  the semantics of $\Diamond (\Diamond 1)$ is the quadratic function: 
\[
\sem{\Diamond\Diamond 1}_\alpha (x) =  \displaystyle \int_{[0,1]} \sem{\Diamond 1}_\alpha \diff \alpha(x) =  \displaystyle \int_{[0,1]} x \cdot x \diff \delta_{x} = x^2 \ .
\tag*{\qed}
\]
\end{example}

The fact that $\sem{\Diamond\phi}_{\alpha}$ is indeed continuous, for any formula $\phi$ and Markov process $\alpha$, is a direct consequence of the Riesz--Markov--Kakutani theorem, as we now prove   (Lemma \ref{continuity_formulas} below).

Recall from Section \ref{TopMeasRieszRepSubSect} that,  by the Riesz--Markov--Kakutani theorem, we have the correspondence
$$
\Rdnl(X) \simeq (X\stackrel{c}{\rightarrow}\mathbb{R}) \stackrel{l}{\rightarrow} \mathbb{R}  \ \ \ \ \ \ \  \ \ \ \ \mu \longleftrightarrow \E_{\mu}
$$
where we used the letters $c$ and $l$ as a reminder of when the space of continuous functions and the space of positive, linear and $\one_X$-decreasing functions are considered. 
Therefore, each Markov process $\alpha:X\rightarrow \Rdnl(X)$ can be identified as the function:
$$
\alpha: X\stackrel{c}{\rightarrow} \big( (X\stackrel{c}{\rightarrow}\mathbb{R}) \stackrel{l}{\rightarrow} \mathbb{R} \big)
$$
where
$$
\alpha(x)(f)= \E_{\alpha(x)}(f) = \int_{X}f \diff \alpha(x).
$$

By swapping the arguments of $\alpha$ as a curried function, we obtain a positive linear map $C(X) \rightarrow C(X)$ (where $C(X) = X\stackrel{c}{\rightarrow}\mathbb{R}$) which, for clarity, we denote by $\Diamond_\alpha$:
\begin{equation}
\label{DiamondShortDefn}
\Diamond_{\alpha} \!:\!  (X\stackrel{c}{\!\rightarrow\!}\mathbb{R}) \rightarrow  (X\stackrel{c}{\!\rightarrow\!}\mathbb{R}) \textnormal{, where } \Diamond_{\alpha}(f)(x)=\alpha(x)(f)  
\end{equation}
To see that $\Diamond_\alpha(f)$ is indeed a continuous function, for any $f\! \in\! C(X)$,  let $(x_i)_{i \in I}$ be a net in $X$ converging to $x\! \in\! X$. We need to prove that
$\lim_{i \in I} \Diamond_\alpha(f)(x_i)  =  \Diamond_\alpha(f) \left(\lim_{i \in I} x_i\right)$.
This follows from the definition (\ref{DiamondShortDefn}) and from 
\begin{align*}
 \lim_{i \in I}\alpha(x_i)(f) = \left(\lim_{i \in I}\alpha(x_i)\right)(f) &= \alpha\left(\lim_{i \in I} x_i\right)(f) 
\end{align*}
where the first equality follows from the definition of the weak-* topology and the second from the continuity of $\alpha$.

The following proposition then follows.
\begin{proposition}\label{properties_modality}
Let $\alpha\!:\!X\!\rightarrow\! \Rdnl(X)$ be a Markov process and let $\Diamond_{\alpha}$ be defined as above. Then, for every $f,g\in C(X)$, the operator $\Diamond_{\alpha}$ has the following properties:
\begin{itemize}
\item (Linear) $\Diamond_{\alpha}(r f ) = r \Diamond_{\alpha}(f)$ and $\Diamond_{\alpha}(f+g) =\Diamond_{\alpha}(f) + \Diamond_{\alpha}(g)$,
\item (Positive) if $f\geq \zero_X$ then $\Diamond_{\alpha}(f)\geq \zero_X$,
\item ($\one_X$-decreasing) $\Diamond_{\alpha}(\one_X)\leq \one_X$.
\end{itemize}  
\end{proposition}


This discussion allows us to equivalently rephrase the definition of the semantics of Riesz modal logic formulas.

\begin{definition}[Semantics, rephrased]\label{def_semantics_2}
Let $\alpha\!:\!X\!\rightarrow\! \Rdnl(X)$ be a Markov process. The semantics $\sem{\phi}_{\alpha}\!\in\! C(X)$ of $\phi$ can be defined by induction on $\phi$ as follows:
\begin{center}
$\sem{0}_{\alpha}=\zero_X \ \ \ \ \ \sem{1}_{\alpha} = \one_X$

\vspace{1mm}

$ \sem{r\phi}_{\alpha}  = r\sem{\phi}_{\alpha} \ \ \  \sem{\phi+\psi}_{\alpha} = \sem{\phi}_{\alpha} + \sem{\psi}_{\alpha} $ 

\vspace{1mm}

$\sem{\phi\sqcup\psi}_{\alpha} = \sem{\phi}_{\alpha} \sqcup  \sem{\psi}_{\alpha} \ \ \ \  
\sem{\phi\sqcap\psi}_{\alpha} = \sem{\phi}_{\alpha} \sqcap  \sem{\psi}_{\alpha}$

\vspace{1mm}

$\sem{\Diamond\phi}_{\alpha} = \Diamond_{\alpha}(\sem{\phi}_{\alpha})$
\end{center}
\end{definition}

The following lemma now becomes obvious since $\Diamond_{\alpha}$ maps continuous functions to continuous functions. 

\begin{lemma}\label{continuity_formulas}
For every $\phi$ the function $\sem{\phi}_{\alpha}$ is continuous.
\end{lemma}

The following simple to prove proposition states that the the semantics of formulas is invariant under coalgebra morphisms.

\begin{proposition}\label{preservation_lemma}
Let $\alpha\!:\!X\!\rightarrow\! \Rdnl(X)$ and $\beta\!:\!Y\!\rightarrow\! \Rdnl(Y)$ be two Markov processes and let $\alpha\stackrel{f}{\rightarrow}\beta$ be a coalgebra morphism. For every formula $\phi$ the equality $\sem{\phi}_{\alpha} = \sem{\phi}_{\beta}\circ f$ holds, i.e.,  $\sem{\phi}_{\alpha}(x) = \sem{\phi}_{\beta}(f(x))$, for all $x\!\in\! X$.
\end{proposition}
\begin{proof}
We simply need to unfold the definitions. Recall that a coalgebra morphism $\alpha\stackrel{f}{\rightarrow}\beta$ is a continuous function $f:X\rightarrow Y$ such that $\beta(f(x)) = \Rdnl(f)\big( \alpha(x))$ holds. By definition of the action of the Radon functor  $\Rdnl$ on morphisms (see Section \ref{TopMeasRieszRepSubSect}) we have that the probability measure $\beta(f(x))$, or equivalently  its corresponding expectation functional $\E_{\beta(f(x))}:C(Y)\rightarrow\mathbb{R}$, is definable as follows:
\begin{equation}\label{obs_1}
\E_{\beta(f(x))}(b)= \E_{\alpha(x)}(b\circ f)
\end{equation}
for all $b\in C(Y)$.  We prove the statement $
\sem{\phi}_{\alpha} = \sem{\phi}_{\beta}\circ f
$ by induction on the structure of $\phi$. The only non trivial case is that of $\phi=\Diamond\psi$. By definition we have:
$$
\sem{\Diamond\psi}_{\alpha}(x)= \E_{\alpha(x)}(\sem{\psi}_{\alpha})
\textnormal{ and }
\sem{\Diamond\psi}_{\beta}(f(x))= \E_{\beta(f(x))}(\sem{\psi}_{\beta})
$$
Therefore, by Equation \ref{obs_1} above, we obtain the equality 
 $
\sem{\Diamond\psi}_{\beta}(f(x))= \E_{\alpha(x)}(\sem{\psi}_{\beta}\circ f)
$.
The inductive hypothesis $\sem{\psi}_{\alpha}\!=\!\sem{\psi}_{\beta}\circ f$ on  $\psi$ then concludes the proof.
\end{proof}

\subsection{Semantic equivalence of formulas}
We now turn our attention to the set of valid equalities between modal Riesz formulas.

\begin{definition}[Equivalence of formulas] \label{sem:equivalence}Given a Markov process $\alpha:X\rightarrow \Rdnl(X)$, we say that two formulas $\phi$ and $\psi$ are $\alpha$-\emph{equivalent}, written $\phi\sim_{\alpha}\psi$, if it holds that $\sem{\phi}_{\alpha}=\sem{\psi}_{\alpha}$. Similarly, we say that two formulas are equivalent, written $\phi\sim\psi$, if for all $\alpha\in\Markov$ it holds that $\phi\sim_{\alpha} \psi$.
\end{definition}

It is clear, from the unital Riesz space structure of $(C(X),\one_X)$, that all Riesz spaces axioms hold true with respect to the equivalence relation $\sim$. For example $\phi + \psi \sim \psi + \phi$ and $(r+s)\phi\sim r\phi + s\psi$. It also follows from the previous discussion on the semantics of the formula $\Diamond\phi$ that 
\begin{itemize}
\item (Linearity) $r\Diamond\phi \sim \Diamond(r\phi)$ and $\Diamond (\phi+\psi) \sim \Diamond\phi + \Diamond \psi$
\item (Positivity) $\Diamond (\phi \sqcup 0) \sqcup 0 \sim \Diamond (\phi \sqcup 0)$
\item ($\one$-decreasing) $\Diamond(1) \sqcup 1 \sim 1$
\end{itemize}

One of the main goals of this work is to show that, in fact, this set of axioms (axioms of Riesz spaces with a positive element together with the axioms listed above for $\Diamond$) is \emph{complete} in the sense that any valid equality $\phi\sim\psi$ can be derived syntactically from these axioms using the inference rules of equational logic and the Archimedean rule. This is stated precisely as Theorem \ref{completeness_theorem_app} in Section \ref{section_applications}.

\subsection{Relation with other probabilistic logics in the literature}\label{sec:other:logics}

Other real-valued logics for expressing properties of Markov chains or similar systems (e.g., Markov decision processes, weighted systems, etc.) have an underlying\footnote{Often these logics extend their basic modal fragment with other operators (e.g., defined by fixed--point equations) which increase the overall expressive power.} basic real--valued modal logic which  differs from Riesz modal logic in the choice of the basic connectives. It turns out that most of such basic modal logics can be interpreted within Riesz modal logic.

For example, the real--valued modal logic of Panangaden (see \cite[\S 8.2]{PrakashBook}), which is particularly important because it characterizes the Kantorovich pseudo-metric on Markov processes\footnote{Here we are slightly abusing the terminology because the notion of Markov process of \cite[\S 8.2]{PrakashBook} differs from ours in that it allows analytic spaces (i.e., continuous images of Polish spaces) as state--spaces, rather than compact Hausdorff spaces, and measurable transition maps, rather than continuous ones. in Section \ref{other:models:section} we discuss how Markov processes in the sense of  \cite[\S 8.2]{PrakashBook}) can be embedded into Markov processes in our sense. But for the purpose of this paragraph, it is enough to compare the logic of Panangaden with Riesz modal logic on probabilistic models that fit both definitions: e.g., Markov processes in our sense having a Polish state--space.} has real-valued semantics of type $\sem{\phi}_{\alpha}\!:\!X\!\rightarrow\![0,1]$ with formulas defined by the syntax:
$$\phi, \psi ::= 1 \mid 1-\phi \mid \phi \sqcap \psi \mid \Diamond \phi\mid \phi\ominus r \ \ \ \ \ \ \ \ \textnormal{where }\!\in\![0,1]$$
and semantics of the arithmetic connectives given as $\sem{\phi\ominus r}_{\alpha}(x)= \max\{0, \sem{\phi}_\alpha(x) - r\}$. Therefore this logic can be directly interpreted in Riesz modal logic by defining $$\phi\ominus r = 0 \sqcup (\phi - r1).$$ 
 
 Similarly, the real--valued modal logic underlying the \L ukasiewicz modal $\mu$-calculus (see \cite{MioSimpsonFI2017} and \cite{MIO2014a}), which is important because this logic (once extended with fixed--point operators) is sufficiently expressive to interpret probabilistic CTL,  has also real-valued semantics of type $\sem{\phi}_{\alpha}\!:\!X\!\rightarrow\![0,1]$ with formulas defined by the syntax:
 $$\phi, \psi ::= 0 \mid 1 \mid r\phi \mid \phi\oplus \psi \mid \phi \odot \psi \mid  \phi \sqcup \psi \mid \phi \sqcap \psi \mid \Diamond \phi
\ \ \ \ \ \ \ \ \textnormal{where }\!\in\![0,1]$$
 and semantics of the arithmetic connectives given as: 
 $$\sem{\phi\oplus \psi}_{\alpha}(x)=\min\{ 1, \sem{\phi}_{\alpha}(x)+ \sem{\phi}_{\alpha}(x)\} \ \ \ \ \ \ \sem{\phi\odot \psi}_{\alpha}(x)=\max\{ 0, \sem{\phi}_{\alpha}(x)+ \sem{\phi}_{\alpha}(x)-1\}.$$
 \noindent 
Therefore, also this logic can be interpreted in Riesz modal logic by defining 
 $$\phi\oplus\psi = 1\sqcap(\phi +\psi)  \ \ \ \ \ \  \phi\odot\psi = 0\sqcup(\phi +\psi-1).$$

This implies that the extension of Riesz modal logic with fixed--point operators is also sufficiently expressive to  interpret probabilistic CTL.



\section{Modal Riesz Spaces}
\label{section_modal}

In this section we introduce the notion of modal Riesz space. This will be the variety of algebras corresponding to Riesz modal logic for Markov processes.

\begin{definition}\label{modal_riesz_def}
A \emph{modal Riesz space} is a structure $(A,u,\Diamond)$ where $(A,u)$ is a Riesz space with designated positive element $u$ (Definition \ref{riesz_positive_def}) and $\Diamond\!:\!A\!\rightarrow\! A$ is a unary operation satisfying:
\begin{enumerate}
\item (Linearity) $\Diamond(a+b) = \Diamond(a) + \Diamond (b)$ and $\Diamond(r a) = r(\Diamond a)$, for all $r\!\in\!\mathbb{R}$
\item (Positivity) $\Diamond(a\sqcup 0)\geq 0$, 
\item ($u$-decreasing) $\Diamond(u)\leq u$.
\end{enumerate}
The full list of axioms in presented in Figure \ref{axioms:of:modal:riesz:spaces}. 
\end{definition}

Thus the class of modal Riesz spaces is a variety in the sense of universal algebra (because the inequalities here can be rewritten as equalities using the lattice operations). Homomorphisms of modal Riesz spaces are unital Riesz homomorphisms which further preserve the $\Diamond$ function (i.e., $f(\Diamond(a))\!=\!\Diamond(f(a)$). We say that $(A,u,\Diamond)$ is \emph{Archimedean} (resp. \emph{unital} and \emph{$u$-complete}) if $(A,u)$ is Archimedean (resp. unital and $u$-complete). We denote by $\CRiesz^{u}_{\Diamond}$ the category having modal Riesz spaces as objects and homomorphisms of modal Riesz spaces as morphisms. We also define $\URiesz_{\Diamond}$, $\AURiesz_{\Diamond}$ and $\CAURiesz_{\Diamond}$ to be the categories of unital, Archimedean and unital, $u$-complete Archimedean and unital modal Riesz spaces, respectively.
$$
\CAURiesz_{\Diamond} \!\hookrightarrow\! \AURiesz_{\Diamond} \!\hookrightarrow\! \URiesz_{\Diamond}\! \hookrightarrow\! \CRiesz^{u}_{\Diamond}
$$

\begin{figure}[h!]
\begin{mdframed}
\begin{center}
 \begin{enumerate}
\item  \textbf{Axioms of Riesz spaces:}
\begin{itemize}
\item Real Vector space: 
\begin{itemize}
\item Additive group: $x + (y + z) = (x + y) + z$, $x + y = y + x$, $x + 0 = x$, $x - x= 0$,
\item Axioms of scalar multiplication: $r_1(r_2 x) = (r_1\cdot r_2) x$, $1x = x$, $r(x+y) = (rx) + (ry)$, $(r_1 + r_2)x = (r_1 x) + (r_2 x)$,
\end{itemize}
\item Lattice axioms:    (associativity) $x \sqcup (y \sqcup z) = (x \sqcup y) \sqcup z$,  $x \sqcap (y \sqcap z) = (x \sqcap y) \sqcap z$, (commutativity) $z \sqcup y = y \sqcup z$, $z \sqcap y = y \sqcap z$,
(absorption) $z \sqcup (z \sqcap y) = z$, $z \sqcap (z \sqcup y) = z$, (idempotence) $x\sqcup x =x$,  $x\sqcap x =x$. 
\item Compatibility axioms: 
\begin{enumerate}
\item $(x \sqcap y) + z \leq  (y + z) $,
\item $r (x \sqcap y) \leq  ry$, for all scalars $r\geq 0$. 
\end{enumerate}
\end{itemize}
\item \textbf{Axiom of the positive element:} $0 \leq u$,
\item \textbf{Modal axioms:}
\begin{itemize}
\item Linearity: $\Diamond(r_1x + r_2y ) = r_1\Diamond(x) + r_2\Diamond(y)$,
\item Positivity: $0 \leq \Diamond (x\sqcup 0) $,
\item $u$-decreasing: $\Diamond 1 \leq 1$.
\end{itemize}
\end{enumerate}
\end{center}
\end{mdframed}
\caption{Equational axioms of modal Riesz spaces.}
\label{axioms:of:modal:riesz:spaces}
\end{figure}

\begin{remark}\label{remark_monotonicity}
Note that in the presence of linearity, positivity of $\Diamond$ is equivalent to monotonicity of $\Diamond$ (i.e., $a\leq b$ implies $\Diamond(a)\leq \Diamond(b)$). Clearly monotonicity implies positivity. In the other direction, assume $\Diamond$ is positive and let $a\leq b$. Note that $a\leq b\Leftrightarrow b-a\geq 0$ \cite[Thm 11.4]{Luxemburg}. Then by positivity $\Diamond(b-a)  \geq 0$. By linearity, $\Diamond(b)-\Diamond(a)\geq 0$ and this is equivalent to $\Diamond(b)\geq \Diamond(a)$.
\end{remark}

\begin{example}\label{trivial_example}
As a trivial example, note that every Riesz space $(A,u)$ can be given the structure of a modal Riesz space by taking, e.g., $\Diamond$ to be the constant $0$ function $\Diamond(a)=0$ or the identity function $\Diamond(a)=a$.
\end{example}

More interestingly, each Markov process gives rise to a modal Riesz space.

\begin{example}
\label{MarkovToRieszExample}
Let $\alpha:X\rightarrow\Rdnl(X)$ be a Markov process. As discussed in Section \ref{logic_section} we can view $\alpha$ as the operator $\Diamond_{\alpha}:C(X)\rightarrow C(X)$ acting on the unital Riesz space $(C(X),\one_X)$. By Proposition \ref{properties_modality} the operator $\Diamond_{\alpha}$ satisfies the required properties to make $(C(X), \one_X, \Diamond_{\alpha})$ a modal Riesz space. Furthermore, since $(C(X),\one_X)\in \CAURiesz$ we have that $(C(X), \one_X, \Diamond_{\alpha})\in \CAURiesz_{\Diamond}$.
\end{example}

Hence, to each Markov process $\alpha\!:\!X\!\rightarrow\! \Rdnl(X)$  corresponds the modal Riesz space $A_{\alpha}\!=\!(C(X), \one_X, \Diamond_{\alpha})\!\in\! \CAURiesz_{\Diamond}$.

By combining the Riesz--Markov--Kakutani representation theorem and Yosida's theorem we have in fact that this correspondence is bijective on isomorphism classes.
\begin{theorem}
\label{DiamondIsoTheorem}
For each $A\!=\!(A,u,\Diamond)\in \CAURiesz_{\Diamond}$, given a choice of isomorphism $A \cong C(X)$, there exists one and only one Markov process $\alpha\!\in\!\Markov$ such that $A \cong A_{\alpha}$.
\end{theorem}
\begin{proof}
By Yosida's theorem (Theorem \ref{YosidaDualityTheorem}), $(A,u)$ is isomorphic to $(C(X),\one_X)$ for a unique (up to homeomorphism) compact Hausdorff space $X\!=\!\Spec(A)$. Fixing such an isomorphism and conjugating the original $\Diamond$ by the isomorphism, we get a positive linear $\one_X$-decreasing map $\Diamond\!:\! C(X)\stackrel{l}{\rightarrow} C(X)$:

$$
\Diamond: (X \stackrel{c}{\rightarrow} \mathbb{R})\stackrel{l}{\rightarrow}(X \stackrel{c}{\rightarrow} \mathbb{R})
$$
and by swapping the arguments as a curried function, we equivalently get a function which, for clarity, we denote by $\alpha_{\Diamond}$:\\
\begin{equation}
\label{DiamondAlphaDefEqn}
\alpha_{\Diamond}: X \stackrel{c}{\rightarrow} \big( (X\stackrel{c}{\rightarrow} \mathbb{R}) \stackrel{l}{\rightarrow}\mathbb{R}\big) \ \ \ \ \ \ \ \alpha_{\Diamond}(x)(f) = \Diamond(f)(x),
\end{equation}
By using the Riesz--Markov--Kakutani theorem, the space $\big( (X\stackrel{c}{\rightarrow} \mathbb{R}) \stackrel{l}{\rightarrow}\mathbb{R}\big)$ coincides with $\Rdnl(X)$. We can show that $\alpha_\Diamond$ is indeed continuous using the definition of continuity in terms of nets, as follows. Let $(x_i)_{i \in I}$ be a net converging to $x\! \in\! X$. 
Since $\Diamond(f)$ is a continuous function, for each $f\!\in\! C(X)$, we have 
$ \Diamond(f)(\lim_i x_i) = \lim_i (\Diamond(f)(x_i))$ and therefore, from the definition $\alpha_\Diamond$ we have
$$ \alpha_\Diamond(\lim_i x_i)(f) = \lim_i (\alpha_\Diamond(x_i)(f))$$
As this holds for all $f \in C(X)$, this shows that $\alpha_\Diamond(\lim_i x_i) = \lim_i \alpha_\Diamond(x_i)$, where the latter limit is with respect to the weak-* topology, and proves that $\alpha_\Diamond$ is continuous. 

Therefore we can see that $\alpha_{\Diamond}\!:\! X\! \rightarrow\! \Rdnl(X)$ is the unique Markov process corresponding to $(A,u,\Diamond)$.
\end{proof}

\begin{example}
For a fixed compact Hausdorff space $X$, let $\alpha\!:\! X \rightarrow \Rdnl(X)$  be the Markov process defined as $\alpha(x)\!=\!\delta_{x}$, for all $x\in X$, where $\delta_x \!\in\!  \Rdnl(X)$ is the \emph{Dirac measure} defined as $\delta_x(A) = 1$ if $x\!\in\! A$ and $\delta_x(A)\!=\!0$ otherwise, for all Borel sets $A\subseteq X$. More colloquially, $\alpha$ is the Markov process where each state $x\!\in\! X$ loops back to itself with probability $1$. Let $A_{\alpha}=(C(X), \one_X, \Diamond_{\alpha})$ be the modal Riesz space corresponding to $\alpha$. It is easy to check that $\Diamond_{\alpha}$ is just the identity map, i.e., $\Diamond_{\alpha}(f)=f$, for all $f\in C(X)$.
\end{example}

Hence there is a bijective correspondence between the (isomorphism classes of) objects of $\Markov$ and the objects of $ \CAURiesz_{\Diamond}$. It will be shown in the next section that this correspondence lifts to a \emph{duality} between the two categories.


We now establish two useful propositions regarding modal Riesz spaces. The first establishes a simple but useful inequality that will be invoked several times.
The second states that the Riesz--ideal of infinitely small elements (see Definition \ref{archimedean:def}) of a modal Riesz space is closed under the ($\Diamond$) operation.

\begin{proposition}\label{useful_equality}
The following equality holds in all modal Riesz spaces: $ | \Diamond (x) | \leq \Diamond ( |x|)$.
\end{proposition}
\begin{proof}
We can express $x$ as the difference of two positive elements: $x = x^+  - x^-$.  Also, recall that $|x|=x^+ + x^-$. By monotonicity and linearity of $\Diamond$ we get:
$$
| \Diamond( x)| = | \Diamond ( x^+ - x^-) | \leq   | \Diamond ( x^+ + x^-) | =  | \Diamond ( x^+) + \Diamond( x^-) |$$
and using the fact that $\Diamond$ is positive, we obtain

$$
| \Diamond ( x^+) + \Diamond( x^-) | =  \Diamond (x^+) + \Diamond(x^-)  =  \Diamond ( x^+ + x^-) = \Diamond ( |x |) 
$$
as desired.
\end{proof}

\begin{proposition}\label{modal-riesz-ideals}
Let $(A,u,\Diamond)$ be a modal Riesz space and $a\in A$ an infinitely small element of $A$. Then $\Diamond(a)$ is also an infinitely small element of $A$.
\end{proposition}
\begin{proof}
The assumption says that, for some $b\in B$ and for all $n\in\mathbb{N}$, the inequality $n |a| \leq |b|$ holds. By monotonicity of $\Diamond$ we obtain that $\Diamond(n |a|) \leq \Diamond(|b|)$. By linearity, $\Diamond(n|a|)=n\Diamond(|a|)$. Hence, 
$$
n \Diamond(|a|) \leq \Diamond(|b|)
$$
for all $n\in\mathbb{N}$. Furthermore, by positivity of $\Diamond$, we know that $\Diamond(|a|)$ and $\Diamond(|b|)$ are  positive elements, i.e., $\Diamond(|a|) = | \Diamond( |a|) |$ and $\Diamond(|b|) = | \Diamond( |b|) |$. Hence

$$
n |\Diamond(|a|)| \leq |\Diamond(|b|)|
$$
for all $n\in\mathbb{N}$, which means that $\Diamond(|a|)$ is an infinitely small element. 

We can conclude the proof by observing that $0\leq | \Diamond (a) | \leq \Diamond ( |a|)$ (see Proposition \ref{useful_equality}) since this implies
$$n | \Diamond(a) | \leq  n \Diamond(|a|)     \leq \Diamond(|b|)$$
i.e., that $\Diamond(a)$ is infinitely small.
\end{proof}

\subsection{Dedekind complete modal Riesz spaces}\label{dedekind:sec}

We have stated in Section \ref{sec:dedekind} as Theorem \ref{dedekind-completion-1} the fundamental fact that each Archimedean unital Riesz space $R$ can be embedded in a Dedekind complete unital Riesz space $\overline{R}$. 

In this section we extend this result by showing that modal Archimedean unital Riesz spaces can be embedded in Dedekind complete modal Riesz spaces. 
This is a direct consequence of a theorem of Kantorovich about the extension of positive linear operators on Riesz spaces.

\begin{theorem}[(Dedekind extension of modal Riesz spaces)]\label{completion:thm2}
Let $(R,\Diamond)$ be a Archimedean and unital modal Riesz space. Then there exists a Dedekind complete Archimedean and unital modal Riesz space $(\overline{R},\overline{\Diamond})$ such that:
\begin{enumerate}
\item $\overline{R}$ is the Dedekind completion of $R$ (from Theorem \ref{dedekind-completion-1}) so we view $R\subseteq \overline{R}$,
\item $\overline{\Diamond}$ extends $\Diamond$, i.e., $\Diamond(f)=\overline{\Diamond}(f)$ for all $f\in R$.
\end{enumerate}
\end{theorem}
\begin{proof}
By Definition \ref{modal_riesz_def}, the operation $\Diamond:R\rightarrow R$ is positive, linear and $1$-decreasing. Kantorovich's theorem (see, e.g., Theorem X.3.1 and subsequent discussion in \S X.4.1 in \cite{vulikh}) states that any function $F:R\rightarrow R$ which is positive ($F(0)\geq 0$) and linear ($F(f+g)=F(f)+F(g)$ and $F(r f)=r F(f)$) can be extended to a positive and linear operator $\overline{F}:\overline{R}\rightarrow\overline{R}$ on the Dedekind completion of $R$. Thus we just need to verify that the resulting $\overline{\Diamond}$ is also $1$-decreasing ($\overline{\Diamond}(1)\leq 1$) and this is clear since $1\!\in\! R$ and therefore $\overline{\Diamond}(1)=\Diamond(1)$ and $\Diamond(1)\leq 1$ because $\Diamond$ is $1$-decreasing.
\end{proof}
\begin{remark}
The choice of $\overline{\Diamond}$ is, in general, not unique. 
\end{remark}
In other words $(R,\Diamond)$ embeds (preserving the modal operation) in $(\overline{R},\overline{\Diamond})$ and existing suprema and infima are preserved.

We denote with $\DAURiesz_\Diamond$ the category of Dedekind complete modal Riesz spaces. The result of Theorem \ref{completion:thm2} implies the following corollary.

\begin{corollary}The equational theories of $\DAURiesz_\Diamond$ spaces and $\AURiesz_\Diamond$ spaces coincide.
\end{corollary}


\subsection{Relation with other works in the literature}
\label{relation:MValgebras}

Following the celebrated theorem of Mundici, which states that the category of abelian lattice--ordered groups and  that of MV--algebras  are equivalent  (see, e.g., \cite{MundiciBook} for a detailed presentation), much work has focused on the study of MV--algebras and its variants.

In \cite{RMV2011} the authors have introduced \emph{Riesz MV--algebras}, which are MV--algebras endowed with the operation of scalar multiplication by reals  in  the unit interval $[0,1]$. They have proved that the categories of  Riesz spaces with strong unit ($\URiesz$) and that of Riesz MV--algebras are equivalent. 
We decided to develop the theory of our probabilistic modal logic on top of the language of Riesz spaces, rather than that of Riesz--MV algebras, because the operations of addition and scalar multiplication by reals are natural for expressing the axioms of the $\Diamond$ operator, whereas the operations of MV--algebras are arguably harder to understand and would result in less readable axioms. However, rephrasing the work presented in this paper using the language of Riesz--MV algebra should be, in principle, possible.

Flaminio and Montagna have extended the notion of MV--algebra to that of \emph{state MV--algebra} \cite{FM2009}. These algebras are MV--algebras extended with a modal operator $(\sigma$) satisfying certain axioms. Their main result is that the $\sigma$ modality can always be identified with a \emph{state on the MV--algebra} and this, in turn, can always be identified with an integration operation on the spectral representation of the underlying MV--algebra \cite{kroupa2006}. The similarities between state MV--algebras from \cite{FM2009} and modal Riesz spaces are, at the present moment, rather unclear. The two notions are unlikely to be equivalent (even via an equivalence of categories) because the $\sigma$ modality of state MV--algebras satisfies the axiom $\sigma(\sigma(x)) = \sigma(x)$ (see Lemma 3.3.G of \cite{FM2009}) while the equation $\Diamond(\Diamond x) =\Diamond x$ does not hold in modal Riesz spaces (see Example \ref{first:example:logic} in Section \ref{logic_section}).

The precise connection between state MV--algebras and modal Riesz spaces is an interesting topic for further research.


\section{Duality between Markov Processes and modal Riesz Spaces}
\label{section_duality}

In this section we extend the adjunction $(\eta,\epsilon): C\dashv \Spec$ between $\CHaus$ and $\AURiesz\op$ of Section \ref{duality_section_background} to one between $\Markov$ and $\AURieszD\op$ which becomes a duality when restricted to the subcategory $\CAURieszD\op$. The unit-counit adjunction is described by the quadruple
$(\eta^\Diamond,\epsilon^\Diamond): C^\Diamond \dashv \Spec^\Diamond$ consisting of the two functors:

\hspace{1cm}$C^\Diamond\! : \! \Markov \!\rightarrow\! \CAURieszD\op\! \hookrightarrow\!  \AURieszD\op$

\hspace{1cm}$\Spec^\Diamond\! :\! \AURieszD\op \!\rightarrow\! \Markov$

\noindent and the two natural transformations:

\hspace{1cm}$\eta^\Diamond: \id_{\Markov} \Rightarrow \Spec^\Diamond\circ C^\Diamond$

\hspace{1cm}$\epsilon^\Diamond: C^\Diamond\circ \Spec^\Diamond \Rightarrow \id_{\AURieszD\op}$\\

We start by defining the functor  $C^\Diamond\! : \! \Markov \!\rightarrow\! \CAURieszD$. On objects $\alpha:X\rightarrow\Rdnl(X)$ in $\Markov$, it is defined as $C^\Diamond(\alpha) \!=\! A_{\alpha}\!=\!(C(X),\one_X,\dia{\alpha})$, as in \eqref{DiamondShortDefn} and Proposition \ref{properties_modality}. On (coalgebra) maps $\alpha\! \stackrel{f}{\rightarrow}\!\beta$ between $\alpha\!:\!X\!\rightarrow\!\Rdnl(X)$ and $\beta\!:\!Y\rightarrow\!\Rdnl(Y)$ having underlying function $f\! :\! X\! \rightarrow\! Y$, we define $C^\Diamond(f)$ to be $C(f)$, where $C\! : \! \CHaus \!\rightarrow\! \CAURiesz\op$ is the functor described in Section  \ref{duality_section_background}.

We now turn our attention to the definition of the functor $\Spec^\Diamond\! :\! \AURieszD\op \!\rightarrow\! \Markov$. On objects $A=(A,u, \Diamond)$ belonging to $\CAURieszD$ the Markov process 
\[
\alpha_\Diamond : \Spec(A) \rightarrow \Rdnl(\Spec(A))
\]
can be defined as in \eqref{DiamondAlphaDefEqn} from Theorem \ref{DiamondIsoTheorem}. If instead $A$ just belongs to $\AURieszD$ we only have (Theorem \ref{YosidaDualityTheorem}) that $A$ is isomorphic, via the counit $\epsilon_A(a)\!=\!\hat{a}$, to a dense subspace of $C(\Spec(A))$. In this case, for each $J\!\in\! \Spec(A)$, we give a partial definition of the subprobability measure (seen as a linear functional) $\alpha_\Diamond(J)$ on all functions $\hat{a}\in C(\Spec(A))$ as in Theorem \ref{DiamondIsoTheorem}:
\begin{equation}
\label{ThetaDiamondDefEqn}
\alpha_\Diamond(J)(\hat{a}) = \widehat{\Diamond(a)}(J)
\end{equation}
We can then uniquely extend $\alpha(J)$ to the whole space $C(\Spec(A))$ by using the fact that $\epsilon_A$ is an isometry with dense image. On a morphism $f \!:\! (A,u_A,\Diamond_A) \rightarrow (B,u_B,\Diamond_B)$ we define $\Spec^\Diamond(f)$ as $\Spec(f)$, where $\Spec\! :\! \AURiesz\op \!\rightarrow\! \CHaus$ is the functor described in Section  \ref{duality_section_background}.

The unit $\eta^\Diamond: \id_{\Markov} \Rightarrow \Spec^\Diamond\circ C^\Diamond$ is defined exactly as the unit $\eta$ from Section \ref{duality_section_background}. That is, for all Markov processes $\alpha\!:\!X\!\rightarrow\!\Rdnl(X)$, we define $\eta^\Diamond_{\alpha}\!=\! \alpha \stackrel{\eta_X}{\rightarrow} \Spec^\Diamond(C^\Diamond(\alpha))$  having underlying function $\eta_X\!:\!X\rightarrow \Spec(C(X))$.

Similarly, the counit $\epsilon^\Diamond: C^\Diamond\circ \Spec^\Diamond \Rightarrow \id_{\AURieszD\op}$ is defined exactly as the counit $\epsilon$ from Section \ref{duality_section_background}. That is, for $A=(A,u,\Diamond)$ in $\AURieszD$ we define $\epsilon^{\Diamond}_{A}=\epsilon_{A}$.

The fact that all previous definitions are consistent, e.g., that $C^\Diamond$ indeed maps coalgebra morphisms to modal Riesz space morphisms or that $\eta^{\Diamond}$ is indeed a collection of coalgebra morphisms, are summarized by the following theorem.



\begin{theorem}\label{main_theorem_paper}
As defined above, $C^\Diamond$ and $\Spec^\Diamond$ are functors and $\eta^\Diamond$ and $\epsilon^\Diamond$ are natural transformations. Furthermore $\Spec^\Diamond$ is a right adjoint to $C^\Diamond$, and restricts to an equivalence $\Markov \simeq \CAURieszD\op$.
\end{theorem}
\begin{proof}
By Example \ref{MarkovToRieszExample} we have that if $(X,\alpha)$ is a Markov process, $(C(X),\dia{\alpha})$ is an object of $\CAURiesz_{\Diamond}$. We now show that if $f : X \rightarrow Y$ underlies a Markov process homomorphism $(X, \alpha) \rightarrow (Y, \beta)$, then $C^\Diamond(f)$ is a morphism in $\CAURiesz_{\Diamond}$ from $C(Y) \rightarrow C(X)$, \emph{i.e.} if the diagram \eqref{MarkovMorphDiag} commutes, then $C(f) \circ \dia{\beta} = \dia{\alpha} \circ C(f)$, as follows. We prove this by applying the left hand side to arbitrary elements $b \in C(Y)$ and $x \in X$:
\begin{align*}
C(f)(\dia{\beta}(b))(x) &= \dia{\beta}(b)(f(x))  \\
 &= \beta(f(x))(b) & \text{by \eqref{DiamondShortDefn}} \\
 &= \Rdnl(f)(\alpha(x))(b) & \text{by \eqref{MarkovMorphDiag}}\\
 &= \alpha(x)(b \circ f) & \text{by \eqref{RdnlMapDefn}} \\
 &= \alpha(x)(C(f)(b)) &  \\
 &= \dia{\alpha}(C(f)(b))(x) & \text{by \eqref{DiamondShortDefn}}.
\end{align*}
We then have that $C^\Diamond$ preserves the identity maps and composition because $C$ does so as a functor $\CHaus \rightarrow \CAURiesz\op$. 

We show that, for $(A,u,\Diamond) \in \AURieszD$, $(\Spec(A),\alpha_\Diamond)$ is a Markov process as follows. By the pointwiseness of the definitions, $\alpha_\Diamond(J)$ is positive and unital for all $J \in \Spec(A)$, at least on the subspace $\hat{\blank}(A) \subseteq C(\Spec(A))$. Its extension to $C(\Spec(A))$ is positive and unital because the positive cone in $C(\Spec(A))$ is (norm) closed. We show that $\alpha_\Diamond$ is a continuous map as follows. Let $(J_i)_{i \in I}$ be a net converging to an ideal $J$ in the hull-kernel topology of $\Spec(A)$. For each $a \in A$, we have that
\begin{align*}
\alpha_\Diamond\left(\lim_i J_i\right)(\hat{a}) &= \widehat{\Diamond(a)}\left(\lim_i J_i\right) = \lim_i \widehat{\Diamond(a)}(J_i) \\
 &= \lim_i \alpha_\Diamond(J_i)(\hat{a}).
\end{align*}
Therefore $\alpha_\Diamond$ is continuous in the weak-* topology defined by $\hat{\blank}(A) \subseteq C(\Spec(A))$. As $\hat{\blank}(A)$ is dense in $C(\Spec(A))$ and $\Rdnl(\Spec(A)) \subseteq C(\Spec(A))^*$ is norm-bounded and therefore equicontinuous, this topology agrees with the usual weak-* topology defined by $C(\Spec(A))$ on $\Rdnl(\Spec(A))$ \cite[III.4.5]{schaefer}. Therefore $\Spec(A,u,\Diamond)$ is always a Markov process. 

Let $f : (A,u_A,\Diamond_A) \rightarrow (B,u_B,\Diamond_B)$ be a morphism in $\AURieszD$. We want to show that $\Spec^\Diamond(f)$ is a morphism of Markov processes, \ie, $\alpha_{\Diamond_A} \circ \Spec(f) = \Rdnl(\Spec(f)) \circ \alpha_{\Diamond_B}$. We do this by proving that for all $J \in \Spec(B)$ and $a \in A$ that $\alpha_{\Diamond_A}(\Spec(f)(J))(\hat{a}) = \Rdnl(\Spec(f))(\alpha_{\Diamond_B}(J))(\hat{a})$, using the denseness of $\hat{\blank}(A) \subseteq C(\Spec(A))$. We have, writing ``nat'' to indicate the use of the naturality of $\hat{\blank}$ from Theorem \ref{YosidaDualityTheorem},
\begin{align*}
\Rdnl(\Spec(f))(\alpha_{\Diamond_B}(J))(\hat{a}) &= \alpha_{\Diamond_B}(J)(C(\Spec(f))(\hat{a})) \\
 &= \alpha_{\Diamond_B}(J)(\widehat{f(a)}) & \text{nat} \\
 &= \widehat{\Diamond_B(f(a))}(J) \\
 &= \widehat{f(\Diamond_A(a))}(J) \\
 &= C(\Spec(f))(\widehat{\Diamond_A(a)})(J) & \text{nat} \\
 &= \widehat{\Diamond_A(a)}(\Spec(f)(J)) \\
 &= \alpha_{\Diamond_A}(\Spec(f)(J))(\hat{a}).
\end{align*}
As in the case of $C^\Diamond$, the rest of the proof that $\Spec^\Diamond$ is a functor follows as in Theorem \ref{YosidaDualityTheorem} from the fact that $\Spec$ is a functor. 

We can finish the proof that this is a dual adjunction that restricts to a duality $\CAURieszD\op \simeq \Markov$ by proving that $N_\blank$ and $\hat{\blank}$, the unit and counit of the adjunction in Theorem \ref{YosidaDualityTheorem}, are a morphism of Markov processes and a modal Riesz homomorphism, respectively. The reason for this is that diagrams in $\Markov$ (respectively, in $\AURieszD$) commute iff their underlying diagrams in $\CHaus$ (respectively, in $\AURiesz$) commute, and morphisms in $\Markov$ (respectively, in $\AURieszD$) are isomorphisms iff their underlying morphisms in $\CHaus$ (respectively, in $\AURiesz$) are isomorphisms. 

We first show that $\hat{\blank}$ is a modal Riesz homomorphism, \ie, that if $(A,u,\Diamond)$ is an object of $\AURieszD$, $\widehat{\Diamond(a)} = \Diamond_{\alpha_{\Diamond}}(\hat{a})$ for all $a \in A$. Let $J \in \Spec(A)$:
\[
\Diamond_{\alpha_{\Diamond}}(\hat{a})(J) = \alpha_{\Diamond}(J)(\hat{a}) = \widehat{\Diamond(a)}(J).
\]

Now we want to show $N_\blank$ is a Markov morphism, \ie, that if $(X,\alpha)$ is a Markov process, $\Rdnl(N_\blank) \circ \alpha = \alpha_{\Diamond_\alpha} \circ N_\blank$. We use the fact that each $b \in C(\Spec(C(X)))$ is of the form $b = \hat{a}$ for some $a \in C(X)$ (Theorem \ref{YosidaDualityTheorem}) to reduce this to showing that $\Rdnl(N_\blank)(\alpha(x))(\hat{a}) = \alpha_{\Diamond_{\alpha}}(N_x)(\hat{a})$. Observe that
\begin{align*}
\alpha_{\Diamond_\alpha}(N_x)(\hat{a}) &= \widehat{\Diamond_\alpha(a)}(N_x) \\
 &= \Diamond_\alpha(a)(x) &\eqref{UsefulHatEquation} \\
 &= \alpha(x)(a) \\
 &= \alpha(x)(\hat{a} \circ N_\blank) & \eqref{UsefulHatEquation} \\
 &= \Rdnl(N_\blank)(\alpha(x))(\hat{a}).
\end{align*}
This concludes the proof.
\end{proof}

\section{Initial Algebra}
\label{final_coalgebra_section}

In this section we study the properties of the initial modal Riesz space in the category $\CAURiesz_{\Diamond}$ of Archimedean unital modal Riesz spaces. Its dual object is the final Markov process (the final coalgebra) in the category $\Markov$ of Markov processes.



\subsection{Initial object of $\CRiesz^{u}_{\Diamond}$}

We first start by considering  the initial object, which we denote by $\mathbb{I}$, in the category $\CRiesz^{u}_{\Diamond}$ of all modal Riesz spaces, thus including non--Archimedean and non--unital spaces. 

Since $\CRiesz^{u}_{\Diamond}$ is a variety in the sense of universal algebra, the initial object exists and it can be constructed as a ground term algebra (free algebra of no generators) in a standard way, as follows.

Let $\Formulas$ be the set of terms without variables in the signature  $\{0,+,(r)_{r\in\mathbb{R}},\sqcap,\sqcup,1,\Diamond\}$ of modal Riesz spaces. Equivalently, $\Formulas$ is the set of formulas of Riesz modal logic.

 We define the equivalence relation $\equiv\ \subseteq \Formulas \times\Formulas$ as: $\phi\equiv \psi$ if and only if $\phi$ and $\psi$ are provably equal (in equational logic) from the axioms of modal Riesz spaces when interpreting the atomic formula $1\!\in\! \Formulas$ as the the constant $u$ in the language of modal Riesz spaces (Definition \ref{modal_riesz_def}):
$$
\phi\equiv \psi\ \ \   \Longleftrightarrow\ \ \   (\textnormal{Axioms of }\CRiesz^{u}_{\Diamond})\vdash \phi=\psi
$$
 
The collection of equivalence classes of $\equiv$ 
$$\Formulas/_{\equiv}=\{ [\phi]_{\equiv} \mid \textnormal{ $\phi$ is a Riesz modal logic formula}\}$$
is endowed with the structure of a modal Riesz space $\mathbb{I}$ as follows: 
$$\mathbb{I}=\langle \Formulas/_{\equiv}, 0^\mathbb{I}, +^\mathbb{I} , (r^\mathbb{I})_{r\in\mathbb{R}}, \sqcap^\mathbb{I}, \sqcup^\mathbb{I},  u^\mathbb{I}, \Diamond^\mathbb{I} \rangle$$
 where:
\begin{center} 
 $0^\mathbb{I} = [0]_{\equiv}$ $\ \ \ $ $u^\mathbb{I}=[1]_{\equiv}$  $\ \ \ $ 
 $\big([\phi]_{\equiv}+^\mathbb{I} [\psi]_{\equiv}\big) = [\phi+\psi]_{\equiv}$\\
 
 \vspace{2mm}
 
 $r^\mathbb{I}([\phi]_{\equiv}) = [r\phi]_{\equiv}$ $\ \ \ $ 
 $\big([\phi]_{\equiv}\sqcap^\mathbb{I} [\psi]_{\equiv}\big) = [\phi\sqcap \psi]_{\equiv}$\\
 
  \vspace{2mm}
 $\big([\phi]_{\equiv}\sqcup^\mathbb{I} [\psi]_{\equiv}\big) = [\phi\sqcup \psi]_{\equiv}$ $\ \ \ $ 
 $\Diamond^\mathbb{I}([\phi]_\equiv) = [\Diamond \phi]_{\equiv}$
\end{center}

\begin{proposition}\label{initial_object_modalrieszspace}
The modal Riesz space $\mathbb{I}$  is the initial object in the category $\CRiesz^{u}_{\Diamond}$ of modal Riesz  spaces. For each modal Riesz space $(A,u,\Diamond)$ there is a unique modal Riesz homomorphism $!_A : \mathbb{I} \rightarrow A$ defined inductively as:

\begin{center}
$!_A ( [0]_\equiv ) = 0_A  \ \ \ \ \ \ !_A ( [1]_\equiv ) = u$
\end{center}

\begin{center}
$!_A ( [\phi + \psi]_\equiv ) \ =\  !_A([\phi]_{\equiv})\ + \ !_A([\psi]_{\equiv}) \  \ \ \ \ \ \ \ \ \ !_A ( [r \phi]_\equiv  ) \ =\  r \big( !_A([\phi]_{\equiv})\big)  $
\end{center}

\begin{center}
$!_A ( [\phi \sqcup \psi]_\equiv ) \ =\  !_A([\phi]_{\equiv})\ \sqcup \ !_A([\psi]_{\equiv}) \  \ \ \ \ \ \ \ \ \ !_A ( [\phi \sqcap \psi]_\equiv ) \ =\  !_A([\phi]_{\equiv})\ \sqcap \ !_A([\psi]_{\equiv})  $
\end{center}

\begin{center}
$!_A ( [\Diamond \phi]_\equiv ) \ =\  \Diamond\big( !_A([\phi]_{\equiv})\big)$.
\end{center}

\end{proposition}

\subsection{Initial object of $\URiesz_{\Diamond}$}

We  now observe that the positive element $[1]_{\equiv}$ is a strong unit of $\mathbb{I}$. 

\begin{theorem}\label{strong_unit_theorem}
The element $[1]_{\equiv}$ is a strong unit of $\mathbb{I}$.
\end{theorem}
\begin{proof}
We need to prove that for every formula $\phi$, there exists some $n\in\mathbb{N}$ such that the inequality $|\phi|\leq n 1$ is derivable from the axioms. This follows easily by induction on the structure of $\phi$ as follows. The base cases $\phi\!=\!0$ and $\phi\!=\!1$ are trivial. For the case $\phi=\phi_1 + \phi_2$ let us fix, using the inductive hypothesis, $n_1,n_2\!\in\!\mathbb{N}$ such that $|\phi_1|\leq n 1$ and $|\phi_2|\leq n_2 1$ respectively. Then the inequality $\phi_1+\phi_2\leq (n_1+n_2)1$ is easily derivable. The cases for $\phi=r\phi_1$, $\phi=\phi_1\sqcap \phi_2$ and $\phi=\phi_1\sqcup \phi_2$ are similar. For the case $\phi=\Diamond \phi_1$, we can use the inductive hypothesis to get a number $n_1$ such that $|\phi_1|\leq n_1 1$. Hence, by monotonicity of $\Diamond$, we get 
$ \Diamond ( |\phi_1| ) \leq \Diamond (n_1 1)$. Using the linearity of $\Diamond$ and the axiom $\Diamond 1 \leq 1$ we obtain 
$$ \Diamond ( |\phi_1| ) \leq \Diamond (n_1 1) = n_1 (\Diamond 1) \leq n_1 1$$
\noindent
We can now conclude using the fact that $| \Diamond \phi_1 | \leq \Diamond( |\phi_1 |)$ (see Proposition \ref{useful_equality}).
%
%
%
\end{proof}

Hence $\mathbb{I}$ belongs to the subcategory $\URiesz_{\Diamond}$ and it is its initial object.

\begin{proposition}
The modal Riesz space $\mathbb{I}$  is the initial object in the category $\URiesz_{\Diamond}$ of modal unital Riesz  spaces.
\end{proposition}

\subsection{Initial object of $\AURiesz_{\Diamond}$}\label{sec:AURiesz}

In \cite{MFM2017} it was claimed (Theorem  VI.3  in \cite{MFM2017}) that $\mathbb{I}$ enjoys the Archimedean property and it is, therefore, the initial object in the category $\AURiesz_{\Diamond}$ of Archimedean unital modal Riesz spaces. The proof, however, contains a fatal mistake. At the present moment, we do not know if $\mathbb{I}$ is Archimedean or not. \\

\emph{Open Problem:} Does the modal Riesz space $\mathbb{I}$ satisfy the Archimedean property?\\

\noindent
However we are able to construct explicitly the initial object $\mathbb{I}_a$ of $\AURiesz_{\Diamond}$ in such a way that If $\mathbb{I}$ is Archimedean, then $\mathbb{I}=\mathbb{I}_a$ (as claimed in \cite{MFM2017}); otherwise, $\mathbb{I}\neq\mathbb{I}_a$.

The Archimedean modal Riesz space $\mathbb{I}_a$ is obtained by quotienting $\mathbb{I}$ by the congruence relation $\approx$ corresponding to the ideal  $\textnormal{Inf}_{\mathbb{I}}$ of infinitely small elements in $\mathbb{I}$: 

$$ [\phi]_\equiv \approx [\psi]_\equiv \Longleftrightarrow   \big( | [\phi]_\equiv - [\psi]_\equiv |\big) \in  \textnormal{Inf}_{\mathbb{I}}$$

Recall from Theorem \ref{unit_archimedean_thm3} that:

\begin{definition}
An element $a=[\phi]_\equiv$ in $\mathbb{I}$ is infinitely small if, for every $n\in\mathbb{N}$ it holds that $na\leq [1]_\equiv$, i.e., by definition of the equivalence relation $\equiv$, if $\vdash n \phi \leq 1$ is derivable by the axioms of modal Riesz spaces. We denote with $\textnormal{Inf}_{\mathbb{I}}$ the collection of  infinitely small elements of  $\mathbb{I}$.
\end{definition} 

The set $\textnormal{Inf}_{\mathbb{I}}$ is a Riesz--ideal of  $\mathbb{I}$ (see Proposition \ref{infinitesimals:are:ideals}) and thus closed under all Riesz operations. Furthermore, by Proposition \ref{modal-riesz-ideals}, $\textnormal{Inf}_{\mathbb{I}}$ is also closed under the $\Diamond$ operation. This implies that the quotient algebra $\mathbb{I}/_\approx$ is a well--defined modal Riesz space which we denote with  $\mathbb{I}_a$:
$$
\mathbb{I}_a = \mathbb{I}/_\approx.
$$
\noindent Note that if $\mathbb{I}$ is Archimedean (see open problem above), then the only infinitely small element is the zero--element ($\textnormal{Inf}_{\mathbb{I}} = \{  [0]_{\equiv} \}$) and therefore $\mathbb{I} = \mathbb{I}_a$.

\begin{theorem}
The modal Riesz space $\mathbb{I}_a$ is the initial object in $\AURiesz_{\Diamond}$. 
\end{theorem}
\begin{proof}
For each modal Riesz space $(A,u,\Diamond)$, we have the unique modal Riesz map $!_A : \mathbb{I} \rightarrow A$ defined inductively in Proposition \ref{initial_object_modalrieszspace}. 

We have the map $[\blank] : \mathbb{I} \rightarrow \mathbb{I}_a$ taking each element to its equivalence class modulo difference by an infinitesimal: 
$
[a] \mapsto [a]/_{\approx}
$.

Given an Archimedean modal Riesz space $(A,u,\Diamond)$, and an infinitesimal element $a \in \mathbb{I}$, we have $!_A(a) = 0$ because infinitesimals map to infinitesimals under unital Riesz homomorphisms (Proposition \ref{hom:map:inf2inf}) and the only infinitesimal in $A$ is $0$, because $A$ is Archimedean by assumption. Therefore the map $!_A : \mathbb{I} \rightarrow A$ factorizes as $!'_A \circ [\blank]$ where $!'_A : \mathbb{I}_a \rightarrow A$, \emph{i.e.} $!_A'$ is well-defined on $\approx$-equivalence classes. 

To show that $!'_A$ is the unique modal Riesz homomorphism from $\mathbb{I}_a$ to $A$, assume the existence of another $f : \mathbb{I}_a \rightarrow A$. Then $f \circ [\blank] : \mathbb{I} \rightarrow A = !_A$. As $[\blank]$ is surjective, this implies $f = !'_A$, proving that $\mathbb{I}_a$ is initial.
\end{proof}

We now give another characterization of $\mathbb{I}_a$ which, being completely syntactic, is sometimes more convenient to work with.

We define the equivalence relation $\equiv_a\ \subseteq \Formulas \times\Formulas$ as: $\phi\equiv_a \psi$ if and only if $\phi$ and $\psi$ are provably equal from the axioms of modal Riesz spaces and the Archimedean infinitary rule $\mathbb{A}$ (see Definition \ref{archimedean:def}):

$$
\phi\equiv_a \psi\ \ \   \Longleftrightarrow\ \ \   (\textnormal{Axioms of }\CRiesz^{u}_{\Diamond}) + \mathbb{A} \vdash \phi=\psi
$$
\noindent
The collection of equivalence classes of $\equiv_a$ 
$$\Formulas/_{\equiv_a}=\{ [\phi]_{\equiv_a} \mid \textnormal{ $\phi$ is a Riesz modal logic formula}\}$$
is endowed with the structure of a modal Riesz space in the same way used for defining $\mathbb{I}$.

\begin{proposition}\label{proposition-logic-archimedean}
The two modal Riesz spaces $\mathbb{I}_a$ and  $\Formulas/_{\equiv_a}$ are isomorphic.
\end{proposition}
\begin{proof}
The isomorphism maps the element $[\phi]_{\equiv_a}$ of $\Formulas/_{\equiv_a}$ to the element $\big[ [\phi]_\equiv \big ]_{\approx}$ of  $\mathbb{I}_a$.

We need to show that this map preserves equivalence classes, i.e., that $\phi \equiv_a \psi$ holds if and only if $[\phi]_\equiv \approx [\psi]_\equiv$.

So, first assume that $[\phi]_\equiv \approx [\psi]_\equiv$. This means, by definition, that $  | [\phi]_\equiv  - [\psi]_\equiv | \in  \textnormal{Inf}_{\mathbb{I}}$. This in turn means that:
$$
\textnormal{For all $n\in\mathbb{N}$ it holds that }  (\textnormal{Axioms of }\CRiesz^{u}_{\Diamond}) \vdash  n| \phi - \psi | \leq 1.
$$
Now applying the Archimedean axiom on these premises, we obtain a proof of $| \phi - \psi |=0$, i.e., $(\textnormal{Axioms of }\CRiesz^{u}_{\Diamond}) + \mathbb{A}  \vdash | \phi - \psi | = 0$ and, from this, it is easy (see Theorem III.3.2(e) of \cite{vulikh}) to obtain a proof of $\phi = \psi$. Thus, by definition, we have that $\phi \equiv_a \psi$ as desired.

Now, for the other direction,  assume that $\phi \equiv_a \psi$, i.e., that $(\textnormal{Axioms of }\CRiesz^{u}_{\Diamond}) + \mathbb{A} \vdash \phi = \psi$. We reason, by induction on the structure of the proof, on the applications of the Archimedean rule. For the base case, if the proof does not use the Archimedean rule then we in fact have $\phi\equiv \psi$ and therefore $[\phi]_\equiv \approx [\psi]_\equiv$. Now assume instead that the proof is concluded by application of the Archimedean rule as follows:
\begin{center}
$\infer[\mathbb{A}]
                        { |\phi_1| = 0 }            {  |\phi_1 |\leq |\gamma| \ \ \ \  2|\phi_1|\leq  |\gamma|\ \ \ \  3|\phi_1|\leq  |\gamma| \ \ \ \ \dots \ \ \ \  n|\phi_1|\leq  |\gamma| \ \ \ \ \dots  }$
\end{center}                        
for some formulas $\gamma$ (in this case $\phi=|\phi_1|$ and $\psi=0$). Then, by induction, we have that for each $n\in\mathbb{N}$ the inequality  $n|\phi_1| \leq  |\gamma|$ holds in $\mathbb{I}_a$, i.e..
$
\big[ n[\phi_1 ]_\equiv \big]_{\approx} \leq \big[  [\gamma ]_\equiv \big]_{\approx} 
$. Since $\mathbb{I}_a$ is Archimedean, it then follows that $\big[ [\phi_1 ]_\equiv \big]_{\approx}  = \big[ [0]_\equiv\big]_\approx$, as desired.
\end{proof}

\subsection{Initial object of $\CAURiesz_{\Diamond}$}

We are now finally ready to construct the initial object of the category  $\CAURiesz_{\Diamond}$ of Archimedean unital and complete modal Riesz spaces. This is the uniform completion (see Definition \ref{uniform_completion}) $C^{\Diamond}(\Spec^{\Diamond}((\mathbb{I}_a)))$ of $\mathbb{I}_a$, which we denote by $\hat{\mathbb{I}}_a$ for convenience. 

\begin{proposition}\label{initiality}
The uniformly complete modal Riesz space $\hat{\mathbb{I}}_a$ is the the initial object in the category $\CAURiesz_{\Diamond}$.
\end{proposition}
\begin{proof}
The functor $\Spec^\Diamond$, being a right adjoint, preserves limits, so it preserves terminal objects. 
Since $\mathbb{I}_a$ is initial in $\AURiesz_{\Diamond}$, this means that $\Spec^\Diamond(\mathbb{I}_a)$ is the final coalgebra of $\Markov$. Since, restricted  $\CAURieszD\op$ the functor $C^\Diamond$ is an equivalence of categories, it does preserve terminal objects. Therefore $ C^\Diamond(\Spec^\Diamond(\mathbb{I}_a))$ is the initial object of $\CAURieszD$ (terminal object of $\CAURieszD\op$). The counit map $\epsilon^\Diamond_{\mathbb{I}_a}$ is such that $C^\Diamond(\Spec^\Diamond(\mathbb{I}_a))$ is isomorphic to the completion of $\mathbb{I}_a$.
\end{proof}

From Proposition \ref{DensenessCor} we get that the two  modal Riesz spaces $\mathbb{I}_a$ and $\hat{\mathbb{I}}_a$ are related by the following fact. 
\begin{proposition}\label{link_proposition}
The modal Riesz space $\mathbb{I}_a$ embeds as a dense subalgebra of $\hat{\mathbb{I}}_a$ and the spectrum $\Spec(\mathbb{I}_a)$ is homeomorphic to the spectrum $\Spec(\hat{\mathbb{I}}_a)$. In particular, there is a one-to-one correspondence between maximal ideals in $\mathbb{I}_a$ and $\hat{\mathbb{I}}_a$. 
\end{proposition}

One may wonder if this universal completion step is really necessary. Indeed, if $\mathbb{I}_a$ is uniformly complete then $\mathbb{I}_a=\hat{\mathbb{I}}_a$. In the rest of this section we show that the completion is indeed necessary because $\mathbb{I}_a$  is not uniformly complete (Theorem \ref{uniform-incompleteness:robert}) and, therefore, does not belong to $\CAURiesz_{\Diamond}$.

\subsubsection{Uniform Incompleteness Proof}\label{subsection:uniform:incompleteness}
We use Example \ref{example:logic:3} here. Specifically, we apply Proposition \ref{properties_modality} to it, obtaining a modal Riesz space $(C([0,1]), \one_{[0,1]}, \Diamond_\alpha)$. We can see that, for each $a \in C([0,1])$,
\[
\Diamond_\alpha(a)(x) = \int_X a \diff \alpha(x) = \int_X a \cdot x \diff \delta_x = x \cdot a(x),
\]
which is to say, if we write $x$ for the identity function rather than a variable, $\Diamond_\alpha(a) = x a$. It is convenient to write it like this because we will be reasoning about the action of $\Diamond_\alpha$ in terms of polynomials, in order to characterize the image of $!_{C([0,1])}$. 

We first need a lemma about the action of $\sqcup$ and $\sqcap$ on polynomials.
\begin{lemma}
\label{PiecewisePolyMeetLemma}
Let $[r,s]$ be a closed interval in $\mathbb{R}$ (so $s \geq r$), and let $a,b : [r,s] \rightarrow \mathbb{R}$ be polynomial functions. There exists a strictly increasing finite (possibly empty) sequence $(x_i)_{i=1}^n$ of numbers in $[r,s]$ such that, defining $x_0 = r$ and $x_{n+1} = s$, $a \sqcup b$ is equal to a polynomial on the interval $[x_i,x_{x+1}]$ for all $i \in \{0, \ldots, n \}$. 
\end{lemma}
\begin{proof}
If $a = b$, then $a \sqcup b = a$, so we take $n =0$ and the sequence to be empty, and we are finished. So for the rest of the proof we assume that $a \neq b$. Therefore $a - b \neq 0$, so it has finitely many real roots in $[r,s]$, which we form into a strictly increasing sequence (without multiplicity) $(x_i)_{i=1}^n$. By the intermediate value theorem, if $(a-b)(x) > 0$ for one $x \in [x_i,x_{i+1}]$, then $(a-b)(x) \geq 0$ for all $x \in [x_i,x_{x+1}]$, otherwise there would be a root in between. Therefore $a(x) \geq b(x)$ and therefore $(a \sqcup b)(x) = a(x)$ for all $x \in [x_i,x_{i+1}]$. If $(a-b)(x) < 0$, then we have $(a \sqcup b)(x) = b(x)$ for all $x \in [x_i,x_{i+1}]$. It cannot happen that $(a - b)(x) = 0$ for $x \in (x_i,x_{i+1})$, so we are finished. 
\end{proof}

Using this, we can characterize the image of $!_{C([0,1])} : \mathbb{I} \rightarrow C([0,1])$, which, as $[\blank] : \mathbb{I} \rightarrow \mathbb{I}_a$ is surjective, characterizes the image of $!'_{C([0,1])} : \mathbb{I}_a \rightarrow C([0,1])$ too. For short, we will write $!$ and $!'$ for these maps.

\begin{lemma}
\label{PiecePolyLemma}
The image of the map $!_{C([0,1])} : \mathbb{I} \rightarrow C([0,1])$ consists of functions $a \in C([0,1])$ such that there exists a finite (possibly empty) strictly increasing sequence of numbers $(x_i)_{i=1}^n$ in $[0,1]$, such that if we take $x_0 = 0$ and $x_{n+1} = 1$, then for all $i \in \{0,\ldots,n\}$, $a$ agrees with a polynomial on $[x_i,x_{i+1}]$. We call such functions \emph{piecewise polynomial}.
\end{lemma}
\begin{proof}
The proof is by induction on the structure of the elements of $\mathbb{I}$. 
\begin{itemize}
\item $0$ and $1$: We have $!(0) = 0$ and $!(1) = 1$, which are both polynomials on $[0,1]$. 
\item Scalar multiplication: Let $r \in \mathbb{R}$ be a scalar, $a \in \mathbb{I}$ an element such that there exists a finite strictly increasing sequence $(x_i)_{i=1}^n$ such that $!(a)|_{[x_i,x_{i+1}]}$ is a polynomial for all $i \in \{0,\ldots n\}$. Then
\[
!(ra)|_{[x_i,x_{i+1}]} = r!(a)|_{[x_i,x_{i+1}]}
\]
which is therefore a polynomial. 
\item $+$: Let $a, b \in \mathbb{I}$ such that there exist strictly increasing finite sequences $(x_i)_{i=1}^n$ and $(y_j)_{j=1}^m$ such that $!(a)$ is a polynomial on each $[x_i,x_{i+1}]$ and $!(b)$ is a polynomial on each $[y_j,y_{j+1}]$. We can enumerate the set $\{x_i\}_{i=1}^n \cup \{y_j\}_{j=1}^m$ in increasing order to obtain a strictly increasing finite sequence $(z_k)_{k=1}^p$ such that both $!(a)$ and $!(b)$ are polynomials on each $[z_k,z_{k+1}]$. Then $!(a+b) = !(a) + !(b)$ will also be equal to a polynomial on each $[z_k,z_{k+1}]$. 
\item $\sqcup$: Let $a,b\in \mathbb{I}$ be elements such that $!(a)$ and $!(b)$ are piecewise polynomial. As in the previous case, form a sequence $(x_i)_{i=1}^n$ such that both $!(a)$ and $!(b)$ are equal to polynomials on each $[x_i,x_{i+1}]$. By Lemma \ref{PiecewisePolyMeetLemma}, there exists a strictly increasing sequence $(y_{i,j})_{j=1}^{n_i}$ such that $!(a) \sqcup !(b)$ is a polynomial on $[y_{i,j},y_{i,j+1}]$ for each $j \in \{ 0,\ldots,n_i\}$. We can therefore enumerate $\{x_i\}_{i=1}^n \cup \bigcup_{i=1}^n \{y_{i,j}\}_{j=1}^{n_i}$ as $(z_k)_{k=1}^p$, and then $!(a)\sqcup !(b)$ is a polynomial on each $[z_k,z_{k+1}]$, as required. 
\item $\Diamond$: Let $a \in \mathbb{I}$ be an element such that $!(a)$ is piecewise polynomial, with sequence $(x_i)_{i =1}^n$. As $!(\Diamond(a))=  xa$, we have that on each $[x_i,x_{i+1}]$ is $xa$, which is a polynomial.
  \qedhere
\end{itemize}
\end{proof}

It is therefore clear that there are elements of $C([0,1])$ outside the image of $!$, and therefore $!'$, such as $x \mapsto e^x : [0,1] \rightarrow \mathbb{R}$. 

We will need the following fact about injective unital Riesz homomorphisms as well. 

\begin{lemma}
\label{InjectiveIsometryLemma}
Let $f : (A,u) \rightarrow (B,v)$ be an injective unital Riesz homomorphism between Archimedean unital Riesz spaces. Then $f$ is an isometry, \emph{i.e.} $\|f(a)\|_v = \|a\|_u$. 
\end{lemma}
\begin{proof}
First observe that injective Riesz homomorphisms are order embeddings, \emph{i.e.} 
\[
f(a) \leq f(b) \Leftrightarrow f(a) \sqcap f(b) = f(a) \Leftrightarrow f(a \sqcap b) = f(a) \Leftrightarrow a \sqcap b = a \Leftrightarrow a \leq b.
\]
Therefore, using the definition of the norm,
\begin{align*}
\|f(a)\|_v &= \inf \{ r \in \mathbb{R}_{\geq 0} \mid |f(a)| \leq rv \} \\
 &= \inf \{ r \in \mathbb{R}_{\geq 0} \mid f(|a|) \leq f(ru) \} \\
 &= \inf \{ r \in \mathbb{R}_{\geq 0} \mid |a| \leq ru \} \\
 &= \|a\|_u.
   \tag*{\qedhere}
\end{align*}
\end{proof}

The following lemma uses a result about completeness of unital Archimedean Riesz spaces (Corollary \ref{QuotientCompleteCor}) that we could only find a convoluted proof of, which is in appendix \ref{QuotientCompleteSubsection}. 

\begin{lemma}
\label{SurjectiveUniformlyCompleteLemma}
Let $f : (A,u) \rightarrow (B,v)$ be a unital Riesz homomorphism between Archimedean Riesz spaces with strong units. If $(A,u)$ is uniformly complete and $f$ has dense image, then $f$ is surjective. 
\end{lemma}
\begin{proof}
As $f$ is unital, it is bounded, and therefore continuous. So the ideal $J = f^{-1}(0)$ (an ideal by \cite[Theorem 18.3 (iii)]{Luxemburg}) is norm-closed. By Corollary \ref{QuotientCompleteCor}, $(A/J,[u])$ is a uniformly complete Archimedean Riesz space with strong unit $[u]$. The map $\tilde{f} : (A/J,[u]) \rightarrow (B,v)$ is injective, and so is an isometry by Lemma \ref{InjectiveIsometryLemma}. Because the image of $f$ is dense, so is the image of $\tilde{f}$, which means that for each $b \in B$ there is a sequence $(a_i)_{i \in \mathbb{N}}$ in $A/J$ such that $\tilde{f}(a_i) \to b$ as $i \to \infty$. Therefore $(\tilde{f}(a_i))_{i \in \mathbb{N}}$ is $\|\blank\|_v$-Cauchy, so $(a_i)_{i \in \mathbb{N}}$ is $\|\blank\|_{[u]}$-Cauchy, and therefore converges to an element $a \in A/J$. By continuity $\tilde{f}(a) = b$, so $\tilde{f}$, and therefore $f$, is surjective. 
\end{proof}

We can now prove the final result.
\begin{theorem}\label{uniform-incompleteness:robert}
$\mathbb{I}_a$ is not uniformly complete.
\end{theorem}
\begin{proof}
By Lemma \ref{PiecePolyLemma} and the lattice form of the Stone-Weierstrass theorem \cite[Lemma 45.2]{Luxemburg}, $!' : \mathbb{I}_a \rightarrow C([0,1])$ has dense image. If $\mathbb{I}_a$ were uniformly complete, $!'$ would be surjective (Lemma \ref{SurjectiveUniformlyCompleteLemma}). However, by Lemma \ref{PiecePolyLemma}, $!'$ is not surjective, \emph{e.g.} $x \mapsto e^x$ is not in the image of $!'$ because it is not equal to a polynomial on any nontrivial interval. 
\end{proof}

\section{Final Coalgebra}\label{final:section}

We have identified in the previous section the initial object $\hat{\mathbb{I}}_a$ in the category $\CAURiesz_{\Diamond}$. From the duality between $\CAURiesz_{\Diamond}$ and $\Markov$ we can infer that the dual object of $\hat{\mathbb{I}}_a$ is the final coalgebra in $\Markov$. We denote this Markov process by $\alpha_{\mathbf{F}}:\mathbf{F}\rightarrow \Rdnl(\mathbf{F})$.

Recall, from Theorem \ref{DiamondIsoTheorem}, that its state space $\mathbf{F}$ 
is the compact Hausdorff space consisting of the collection of maximal ideals in $\hat{\mathbb{I}}_a$ endowed with the hull-kernel topology:
$$\mathbf{F}=\Spec(\hat{\mathbb{I}}_a).$$

Therefore we can study the topological structure of the final coalgebra $\alpha_\mathbb{F}$ of $\Markov$ by studying the (hull-kernel topology of) the collection of maximal ideals in $\hat{\mathbb{I}}_a$. By Proposition \ref{link_proposition}, we can equivalently study the maximal ideals in $\mathbb{I}_a$:
\begin{equation}
\mathbf{F}=\Spec(\hat{\mathbb{I}}_a) = \Spec(\mathbb{I}_a)
\end{equation}
To illustrate this method, we now show that the compact Hausdorff space $\mathbf{F}$ is Polish, i.e., separable and (completely) metrizable.

Recall from Definition \ref{normed_space} that $\mathbb{I}_a$, being unital and Archimedean, is a normed space with norm $\|\_\|$ and compatible metric $d$. We first establish the following useful property of $\mathbb{I}_a$ by elementary syntactic arguments.

\begin{lemma}\label{density_rationals}
For each formula $\phi$ there exists a formula $\psi$ having only rational coefficients such that
$d( [\phi]_{\equiv_a}, [\psi]_{\equiv_a})\leq\epsilon$.
\end{lemma}
\begin{proof}
By Proposition \ref{proposition-logic-archimedean}, we need to prove $d( [\phi]_{\equiv_a}, [\psi]_{\equiv_a})\leq\epsilon$, i.e., that the inequality $ |  [\phi]_{\equiv_a} -  [\psi]_{\equiv_a} | \leq [\epsilon 1]_{\equiv_a}$ holds in $\mathbb{I}_a$. This means we need to derive the inequality $
|\phi - \psi| \leq \epsilon 1 
$ form the axioms of modal Riesz spaces and the Archimedean rule. Our proof does not make, in fact, any use of the Archimedean rule.

The proof goes by induction on the modal-depth $m(\phi)$ of $\phi$ defined inductively by $m(0)\!=\!m(1)\!=\!0$, $m(\Diamond\phi)=1+m(\phi)$ and $m(\phi_1+\phi_2)= \max\{ m(\phi_1),m(\phi_2)\}$ and similarly for all other connectives. The base case $m(\phi)=0$ is trivial, as $\phi$ can be identified with a real number $r_{\phi}\in\mathbb{R}$ and we can choose $\psi$ to be $s1$ for some rational $s$ such that $|r-s|<\epsilon$.

Suppose now that $m(\phi)=k+1$. We have to consider all separate cases. The most interesting is the case $\phi\!=\!\Diamond\phi_1$. We can pick, by inductive hypothesis on $\phi_1$, a formula $\psi_1$ with rational coefficients such that $d(\phi_1,\psi_1) < \epsilon$, i.e., $[|\phi_1 -\psi_1|]_{\equiv_a} \leq [\epsilon 1]_{\equiv_a}$. Then, using the inequality $ | \Diamond (x) | \leq \Diamond ( |x|)$ from Proposition \ref{useful_equality}, it is simple to show that:
$$|\Diamond\phi_1 - \Diamond\psi_1|  = |\Diamond( \phi_1 - \psi_1)| \leq \Diamond(| \phi_1 - \psi_1 | )  \leq \Diamond(\epsilon 1) = \epsilon\Diamond(1)\leq \epsilon 1$$

All other cases involving the other connectives follow easily form the inductive hypothesis. For example, if $\phi\!=\!\phi_1+ \phi_2$ then, by inductive hypothesis, we can pick $\psi_1$ and $\psi_2$ such that $|\phi_i - \psi_i | <\frac{1}{2}\epsilon$. It then clearly follows that $\psi=\psi_1+\psi_2$ has the required property.
\end{proof}

Note that the set of formulas $\psi$ having rational coefficients is countable. Hence $\mathbb{I}_a$ is \emph{separable} as a metric space.

\begin{corollary}
The unital Archimedean algebra $\mathbb{I}_a$ is separable as a metric space.
\end{corollary}



As a corollary of the previous proposition, we get the following interesting property regarding the topology of $\mathbf{F}$. 

\begin{theorem}\label{corollary:polish}
Let $(A,u)$ be an Archimedean unital Riesz space which is separable as a metric space in its norm. Then $\Spec(A)$ is Polish. Therefore $\Spec(\mathbb{I}_a)$ is Polish.
\end{theorem}
\begin{proof}
The space $(A,u)$ embeds densely in $C(\Spec(A))$, so $C(\Spec(A))$ is separable in norm. This implies that the compact space $\Spec(A)$ is metrizable \cite[V.6.6 Theorem]{conway}. Every compact metric space is Polish \cite[Chapter I, \S 4.2 Proposition]{Kechris}.
\end{proof}


\section{Applications of Duality to Riesz modal Logic}
\label{section_applications}

In this section we use the duality theorem and the characterization of the initial modal Riesz space $\hat{\mathbb{I}}_a$ and its dual, the final Markov process $\alpha_\mathbf{F}$, to prove basic results about the Riesz modal logic.

We start by proving that the proof system presented in Figure \ref{figure:full:axiomatisation} for deriving equalities between Riesz modal logic formulas is sound and complete 
with respect to the semantic equivalence relation  $(\sim$) of Definition \ref{sem:equivalence}.

\begin{figure}[h!]
\begin{mdframed}
\begin{center}
 \begin{enumerate}
\item  \textbf{Axioms of Riesz spaces:}
\begin{itemize}
\item Real Vector space: 
\begin{itemize}
\item Additive group: $x + (y + z) = (x + y) + z$, $x + y = y + x$, $x + 0 = x$, $x - x= 0$,
\item Axioms of scalar multiplication: $r_1(r_2 x) = (r_1\cdot r_2) x$, $1x = x$, $r(x+y) = (rx) + (ry)$, $(r_1 + r_2)x = (r_1 x) + (r_2 x)$,
\end{itemize}
\item Lattice axioms:    (associativity) $x \sqcup (y \sqcup z) = (x \sqcup y) \sqcup z$,  $x \sqcap (y \sqcap z) = (x \sqcap y) \sqcap z$, (commutativity) $z \sqcup y = y \sqcup z$, $z \sqcap y = y \sqcap z$,
(absorption) $z \sqcup (z \sqcap y) = z$, $z \sqcap (z \sqcup y) = z$, (idempotence) $x\sqcup x =x$,  $x\sqcap x =x$. 
\item Compatibility axioms: 
\begin{itemize}
\item $(x \sqcap y) + z \leq  (y + z) $,
\item $r (x \sqcap y) \leq  ry$, for all scalars $r\geq 0$. 
\end{itemize}
\end{itemize}
\item \textbf{Axiom of the positive element:} $0 \leq u$,
\item \textbf{Modal axioms:}
\begin{itemize}
\item Linearity: $\Diamond(r_1x + r_2y ) = r_1\Diamond(x) + r_2\Diamond(y)$,
\item Positivity: $0 \leq \Diamond (x\sqcup 0) $,
\item $u$-decreasing: $\Diamond 1 \leq 1$.
\end{itemize}
\item \textbf{Archimedean rule:}
\begin{center}
$\infer[\mathbb{A}]
                        {x = 0 }            {  |x|\leq | y| \ \ \ \  2|x|\leq |y| \ \ \ \  3|x|\leq |y| \ \ \ \ \dots \ \ \ \  n|x|\leq |y|  \ \ \ \ \dots  }$
\end{center}
\end{enumerate}
\end{center}
\end{mdframed}
\caption{Sound and Complete Proof System for the Riesz modal logic.}
\label{figure:full:axiomatisation}
\end{figure}

\begin{theorem}\label{completeness_theorem_app}
Let $\phi,\psi\in\texttt{Form}$ be two modal Riesz logic formulas. Then
$$
\phi \sim \psi \ \ \ \ \Longleftrightarrow \ \ \ \  (\textnormal{Axioms of }\CRiesz^{u}_{\Diamond}) + \mathbb{A}\vdash \phi=\psi.
$$
\end{theorem}
\begin{proof}
Direction $(\Leftarrow)$ (soundness). We know that to each Markov process $\alpha\!:\!X\!\rightarrow\!\Rdnl(X)$ corresponds the Archimedean unital modal Riesz space $A_{\alpha}\!=\!(C(X),\one_X,\Diamond_{\alpha})$ and, by Definition \ref{def_semantics_2}, that $\sem{\phi}_{\alpha}\!=\!\sem{\psi}_{\alpha}$ holds if the equality $\phi \!=\!\psi$ holds in $A_{\alpha}$. The assumption $(\textnormal{Axioms of }\CRiesz^{u}_{\Diamond}) + \mathbb{A} \vdash \phi=\psi$ means that $\phi\!=\!\psi$ is true in all Archimedean modal Riesz spaces and in particular in all $A_{\alpha}$.

Direction $(\Rightarrow)$ (completeness). Assume $\phi\!\sim\!\psi$ holds, i.e., the equality $\sem{\phi}_\alpha \!=\! \sem{\psi}_{\alpha}$ holds for all Markov processes $\alpha\!:\!X\!\rightarrow\!\Rdnl(X)$. In particular the equality holds on the final Markov process $\alpha_{\mathbf{F}}\!:\!\mathbf{F}\!\rightarrow\!\Rdnl(\mathbf{F})$. By duality, we have $C(\mathbf{F}) \simeq \hat{\mathbb{I}}_a$.

Therefore $\sem{\phi}_{\alpha_{\mathbf{F}}} = \sem{\psi}_{\alpha_{\mathbf{F}}}$ means $[\phi]_{\equiv} \approx [\psi]_{\equiv}$ or, equivalently by Proposition \ref{proposition-logic-archimedean}, that $\phi\equiv_a \psi$. By definition of the equivalence relation $\equiv_a$  this means:
$$(\textnormal{Axioms of }\CRiesz^{u}_{\Diamond}) + \mathbb{A} \vdash \phi=\psi$$
and the proof is completed.
\end{proof}

The following theorem is another simple consequence of the machinery based on duality.
\begin{theorem}\label{behavioral_characterization}
Let $x,y\in\mathbf{F}$ two points in the final coalgebra. If $x\neq y$ then there exists a formula $\phi$ such that that $\sem{\phi}_{\alpha_{\mathbf{F}}}(x)\neq\sem{\phi}_{\alpha_{\mathbf{F}}}(y)$.\end{theorem}
\begin{proof}
The space $\mathbf{F}=\Spec(\mathbb{I}_a)$ is compact Hausdorff. Therefore points can be separated by continuous functions meaning that $x\!\neq\! y$ if and only if there exists a continuous function $f\!\in\! C(\mathbf{F})$ such that $f(x)\!\neq\! f(y)$. 
By duality we have that $C(\mathbf{F}) \simeq \hat{\mathbb{I}}_a$. Furthermore we know by Proposition \ref{link_proposition} that $\mathbb{I}_a$ is a dense subalgebra of $\hat{\mathbb{I}}_a$. Hence, by choosing a sufficiently close approximation of $f$, we obtain a function $g\!\in\!\mathbb{I}_a\subseteq C(\mathbf{F})$ such that $g(x)\!\neq\! g(y)$. Now $g\!=\![\phi]_{\equiv_a}$ for some formula $\phi\in\texttt{Form}$ and this is the desired separating formula.
\end{proof}

By combining Theorem \ref{behavioral_characterization} above with Proposition \ref{preservation_lemma} we then get the following corollary which states that modal Riesz logic formulas characterize behavioural equivalence. Recall that two states of a Markov process $\alpha$ are called \emph{behaviourally equivalent} if $\eta(x)=\eta(y)$ where $\alpha\stackrel{\eta}{\rightarrow}\alpha_{\mathbf{F}}$ is the unique coalgebra morphism from $\alpha$ to the final coalgebra.

\begin{corollary}\label{corollary:bisimulation}
Let $\alpha:X\rightarrow \Rdnl(X)$ be a Markov process and $x,y\in X$. Then $x$ and $y$ are behaviourally equivalent if and only if $\sem{\phi}_{\alpha}(x)= \sem{\phi}_{\alpha}(y)$ for all formulas $\phi$.
\end{corollary}


\section{Other Classes of Models}
\label{other:models:section}

The completeness result (Theorem \ref{completeness_theorem_app}) may be considered, at a first glance, as slightly artificial. This is because the class of models we are considering (i.e., Markov processes in the sense of Definition \ref{markov_def_1}) have a compact Hausdorff space as state--space and the transition function is required to be continuous. But often, in practice, one is interested in interpreting probabilistic logics on probabilistic transition systems whose state--space is not a compact space or on systems having a discontinuous transition function.

\begin{example}
Consider a Markov chain having a countably infinite state space. Then its state space (when viewed as a topological space with the discrete topology) is not a compact space and thus the Markov chain cannot be naturally modelled as a Markov process in the sense of Definition \ref{markov_def_1}. 
\end{example}

\begin{example}
A (discrete--time) random walk on the real line could be modelled as a Markov kernel (i.e., measurable map) $\tau:\mathbb{R}\rightarrow \mathcal{M}^{\leq 1}(\mathbb{R})$ (where $ \mathcal{M}^{\leq 1}$ is the Giry monad from \cite{giri1982}). This very natural model does not fit the definition of Markov process of Definition \ref{markov_def_1} for two reasons: $\mathbb{R}$ is not compact and $\tau$ is, generally, not continuous but merely measurable. 
\end{example}

More generally, many interesting examples of Markov processes are naturally modelled as measurable maps $\tau:X\rightarrow \mathcal{M}^{\leq 1}(X)$ where $X$ is a standard Borel space and $\tau$ is measurable. Several other examples can be found in the literature: for example Markov processes defined on analytic spaces \cite{PrakashBook} or even measurable spaces \cite{KMP2013}.

These models do not fit Definition \ref{markov_def_1}. Yet, Riesz modal logic can be naturally interpreted over them simply by defining: 
$$
\sem{\Diamond\phi}_\tau(x)= \int_{X}\sem{\phi}_\tau \diff \tau(x) \ \ \ 
$$
Let us write $\mathcal{B}(X,\R)$ for the set of bounded measurable real-valued functions on a measurable space $(X,\Sigma)$. This is a subset of $\ell^\infty(X)$, and is in fact a Riesz subspace when equipped with the pointwise operations defined from those in $\R$. Similarly to \eqref{DiamondShortDefn}, we can define a positive linear $\one_X$-decreasing map $\Diamond_\tau : \mathcal{B}(X,\R) \rightarrow \mathcal{B}(X,\R)$ by
\[
\Diamond_\tau(f)(x) = \int_X f \diff \tau(x).
\]
This definition goes back to the predicate transformer semantics defined in \cite[\S 2]{Kozen1983}. It follows from the fact that $\tau(x)$ is a probability measure and linearity of integration that $\Diamond_\tau$ is a positive linear $\one_X$-decreasing map $\mathcal{B}(X,\R) \rightarrow \ell^\infty(X)$. By the definition of the $\sigma$-algebra on $\mathcal{M}^{\leq 1}(X)$, for all $S \in \Sigma$ the function $\Diamond_\tau(\chi_S) \in \mathcal{B}(X,\R)$, and it then follows by a standard argument using linearity and the dominated convergence theorem that $\Diamond_\tau(f) \in \mathcal{B}(X,\R)$ for all $f \in \mathcal{B}(X,\R)$. 

It is then simple to show, by induction on the complexity of $\phi$, that the semantics $\sem{\phi}_\tau$ of every formula $\phi$ under this interpretation is a \emph{bounded} Borel measurable function $\sem{\phi}_\tau :X\rightarrow\mathbb{R}$.

We now explain how our completeness theorem still holds if all the models of the examples above (and arguably most other models in the literature) were considered in  addition to Markov processes ($\Markov$) defined on compact Hausdorff spaces with continuous transitions (as in Definition \ref{markov_def_1}). The key idea is that more general models can be \emph{embedded} into Markov processes in the sense of  Definition \ref{markov_def_1}. This can be proved, as we now show, using the duality results of Section \ref{section_duality}.

Let us denote by $\mathcal{C}$ one of the classes of measure-theoretic models discussed above\footnote{In increasing generality, Markov processes with a state space that is a standard Borel space, an analytic space, or just a general measurable space.} together with the corresponding interpretation of Riesz modal logic in terms of (measurable) bounded functions.

\begin{theorem}[(Extended Model Completeness)]\label{model_completeness_theorem}
Given two formulas $\phi$ and $\psi$ of Riesz modal logic, the following assertions are equivalent:
\begin{enumerate}
\item  $(\textnormal{Axioms of }\CRiesz^{u}_{\Diamond}) + \mathbb{A}\vdash \phi=\psi$
\item $\sem{\phi}_\tau\! =\! \sem{\psi}_\tau$ holds in all models in $\Markov \cup \mathcal{C}$.
\end{enumerate}
\end{theorem}
\begin{proof}
It is sufficient to prove that if an equality fails in some model in $\mathcal{C}$ then it fails in some Markov process in the sense of Definition \ref{markov_def_1}. Formally, we need to prove that if $\sem{\phi}_\tau\! \neq\! \sem{\psi}_\tau$ in some model $(X,\tau)\in \mathcal{C}$ then there exist a Markov process $(Y,\sigma)\in  \Markov$ (i.e., in the sense of Definition \ref{markov_def_1}) such that $\sem{\phi}_\sigma \neq \sem{\psi}_\sigma$.

Recall that $\mathcal{B}(X,\mathbb{R})$ is the space of bounded Borel measurable functions of type $X\rightarrow\mathbb{R}$. By the discussion above, $\mathcal{B}(X,\mathbb{R})$  is an Archimedean unital Riesz space with strong unit $\sem{1}_{\tau} = (x\mapsto 1)$, and the interpretation $\sem{\Diamond}_\tau$ described above makes the structure $A= (\mathcal{B}(X,\mathbb{R}),\sem{\Diamond}_\tau)$ into a modal Riesz space. By assumption, the equality $\phi=\psi$ fails in $A$, and so by duality, the modal Riesz space $A$ space is isomorphic to a subspace of $(C(Y), \sem{\Diamond}_\sigma)$ for some $(Y,\sigma)\in \Markov$, in which $\phi = \psi$ therefore fails as well, and this concludes the proof.
\end{proof}

The proof of the above theorem shows that, as long as we deal with reasonable models of probabilistic transition systems, the denotation of Riesz modal logic formulas  belongs to some Archimedean Riesz space of bounded real--valued functions and thus, using the duality theory,  it can be also be equally interpreted in some Markov process in the sense of Definition \ref{markov_def_1}.


\subsection{Labelled Markov processes}

In this paper we have modelled Markov processes as transition functions mapping states to subprobability measures. This choice was made,  once again, for mathematical convenience: the axiomatization of the $\Diamond$ operator of Riesz modal logic is simple and intelligible.

In operational semantics (see, e.g., \cite{SOS,Sokalova2011,PrakashBook}) it is very common to consider transition systems having labelled transitions. Labelled Markov processes, still based on sub--probability measures, can be defined as follows:

\begin{definition}
Let $L$ be a set of labels. A \emph{labelled Markov process} is a pair $(X, \{\tau_l\}_{l\in L})$ where $X$ is a compact Hausdorff space and $\tau_l:X\rightarrow\mathcal{R}^{\leq 1}(X)$ is a continuous map.
\end{definition}

Riesz modal logic can naturally be adapted to be interpreted over labelled Markov processes by replacing the single modality $\Diamond$ with a $L$-indexed family of modalities $\diam{l}$, and by interpreting these labelled modalities as expected (see Definition \ref{def_semantics_2}):
$$
\sem{\diam{l}\phi} (x) = \displaystyle \mathbb{E}(\sem{\phi},\tau_l(x)) 
$$

This multimodal variant of Riesz modal logic can be axiomatized just by duplicating the axioms of $\Diamond$ for each $\diam{l}$ and $l\in L$. For example, if $L=\{a,b\}$, the axiomatization is obtained by taking the axioms of Riesz spaces and the following equations 
\begin{center}
\begin{itemize}

\item modal axioms for $\diam{a}$:
\begin{enumerate}[label=]
\item (Linearity) $\diam{a}(f\!+\!g) = \diam{a}(f)\! +\! \diam{a} (g)$ and $\diam{a}(r f) = r(\diam{a} f)$, for all $r\!\in\!\mathbb{R}$,
\item (Positivity) $\diam{a}(f\sqcup 0)\geq 0$, 
\item ($1$-decreasing) $\diam{a}(1)\leq 1$.
\end{enumerate}

\item modal axioms for $\diam{b}$:
\begin{enumerate}[label=]
\item (Linearity) $\diam{b}(f\!+\!g) = \diam{b}(f)\! +\! \diam{b} (g)$ and $\diam{b}(r f) = r(\diam{b} f)$, for all $r\!\in\!\mathbb{R}$,
\item (Positivity) $\diam{b}(f\sqcup 0)\geq 0$, 
\item ($1$-decreasing) $\diam{b}(1)\leq 1$.
\end{enumerate}
\end{itemize}
\end{center}

\section{Conclusions}
\label{conclusion_section}

We have introduced \emph{Riesz modal logic}, a real--valued endogenous probabilistic modal logic for expressing properties of probabilistic transition systems. The syntax and the semantics of the logic are directly inspired by the theory of Riesz spaces and this has allowed us to develop a mathematically convenient duality theory. We have also shown that Riesz modal logic can interpret other basic real--valued probabilistic logics appeared in the literature: most importantly, the  \emph{\L ukasiewicz modal logic} from \cite{MioThesis,MIO2014a,MioSimpsonFI2017}. This implies that the extension of Riesz modal logic with fixed-point operators results in a very expressive probabilistic logic capable of interpreting very popular probabilistic logics such as probabilistic CTL.

The study of specific fixed--point extensions, including questions such as axiomatizations and decidability properties, is a very interesting topic for further work (see, e.g., \cite{Mio18} for preliminary results). In this paper, we have proved a key extension theorem (Theorem  \ref{completion:thm2}) which is likely going to be of fundamental importance when the existence of the fixed--point considered is guaranteed by the Knaster--Tarski theorem (as in, e.g., \cite{MioThesis,MIO2014a,MioSimpsonFI2017,Mio18} and Kozen's modal $\mu$--calculus \cite{Kozen83}).

We have left open an important question (see Open Problem in Section \ref{sec:AURiesz}) regarding the Archimedean property of the initial modal Riesz space. A positive answer to this question (as claimed in \cite{MFM2017} but using a wrong argument) would imply that the axiomatization of Riesz modal logic (Figure \ref{figure:full:axiomatisation} in Section \ref{section_applications}) remains complete even when the Archimedean rule is removed from the proof system. This of course has some practical  interest since the Archimedean rule is infinitary and not easily tractable. 

Lastly, another aspect not considered in this work, and left for future work, is the design of convenient analytical proof systems (e.g., sequent--calculus) for Riesz modal logic and its extensions. This is a very interesting direction for future work. See, e.g., \cite{LM2019}, for preliminary results.


\section*{Acknowledgements}
The work of  Mio has been partially supported by the French project ANR-16-CE25-0011 REPAS. The work of Furber and Mardare has been partially supported by the DFF project 4181-00360 funded by the Danish Council for Independent Research.




\appendix

\section{Proof that Quotients of Complete Archimedean Riesz Spaces are Complete}
\label{QuotientCompleteSubsection}
In this appendix we present a proof for the fact that if $(A,u)$ is an Archimedean Riesz space with strong unit $u$ and uniformly complete, $I \subseteq A$ a closed ideal, then $A/I$ is Archimedean and uniformly complete. This fact is used in Section \ref{subsection:uniform:incompleteness}. We could not prove it the direct way, by showing that the $[u]$-norm of $A/I$ is the quotient norm of $A$ (which it is easy to prove is complete if $A$ is), so we used Yosida duality.

\begin{lemma}
\label{QuotientStrongUnitLemma}
Let $A$ be a Riesz space, $u \in A_+$ a strong unit, and $I \subseteq A$ an ideal. Then $[u]$ is a strong unit in $A/I$. 
\end{lemma}
\begin{proof}
As $I$ is an ideal, $[\blank] : A \rightarrow A/I$ is a Riesz homomorphism, and is therefore monotone and linear. If $[a] \in A/I$, there exists $n \in \mathbb{N}$ such that $-nu \leq a \leq nu$ (equivalently $|a| \leq nu$), so $-n[u] \leq [a] \leq n[u]$, proving $[u]$ is a strong unit. 
\end{proof}

Let $X$ be a compact Hausdorff space. Given $Y \subseteq X$ a closed subset, define
\[
I(Y) = \{ a \in C(X) \mid \forall y \in Y. a(y) = 0 \}.
\]
By the definition of the Riesz operations in $C(X)$, it is easy to see that this is an ideal. Since convergence in $C(X)$ corresponds to uniform convergence of functions, $I(Y)$ is always a closed ideal (with respect to the norm defined by the unit of $C(X)$). If $J \subseteq C(X)$ is a norm-closed ideal, we define
\[
Z(J) = \{ x \in X \mid \forall a \in J. a(x) = 0 \} = \bigcap_{a \in J}a^{-1}(0).
\]
Being an intersection of closed sets, $Z(J)$ is a closed subset of $X$. 

\begin{lemma}
\label{ApproxUnitLemma}
Let $X$ be a compact Hausdorff space and $J \subseteq C(X)$ an ideal. There is a directed set $(v_k)_{k \in K}$ of elements of $J$ that are $[0,1]$-valued functions converging pointwise to $1$ on $X \setminus Z(J)$. 
\end{lemma}
\begin{proof}
Let $K$ be the set of finite subsets of $X \setminus Z(J)$. For each point $x \in X \setminus Z(J)$, there exists $a \in J$ such that $a(x) \neq 0$. As $J$ is an ideal, the element $v_{\{x\}} = \frac{|a|}{|a|(x)}\sqcap 1 \in J$, is $[0,1]$-valued, and takes the values $1$ at $x$. We then define $v_{\{x_1,\ldots,x_n\}} = \bigsqcup_{i=1}^nv_{x_i}$. Then $(v_k)_{k \in K}$ is a directed set, and it converges pointwise to $1$ on $X \setminus Z(J)$. 
\end{proof}

\begin{lemma}
\label{UniformConvApproxLemma}
Let $X$ be a compact Hausdorff space, $J \subseteq C(X)$ an ideal, $a \in C(X)$ vanishing on $Z(J)$. Let $\alpha = \sup_{x \in X}a(x) + 1$. Then $\alpha v_k \land a \to a$ uniformly. 
\end{lemma}
\begin{proof}
Let $\epsilon > 0$. Define $C = a^{-1}(\mathbb{R} \setminus (-\epsilon,\epsilon))$, which is closed, and therefore compact. As $C \subseteq X \setminus Z(J)$, $(v_k)_{k \in K}$ converges pointwise to $1$ on $C$, so by Dini's theorem \cite[X.4.1 Theorem 1]{bourbaki} it converges uniformly on on $C$. Therefore there exists a $k \in K$ such that for all $k' \in K$ with $k' \geq k$, and for all $x \in C$, $|1- v_{k'}(x)| < \frac{1}{2 \alpha}$. Therefore $|\alpha - \alpha v_{k'}(x)| < \frac{1}{2}$, which, because $0 \leq v_{k'}(x) \leq 1$, is the same as $\alpha - \alpha v_{k'}(x) < \frac{1}{2}$. So
\begin{align*}
\alpha v_{k'}(x) > \alpha - \frac{1}{2} &= \sup_{x \in X}a(x) + 1 - \frac{1}{2} \\
&= \sup_{x \in X}a(x) + \frac{1}{2}.
\end{align*}
So $(\alpha v_{k'} \sqcap a)(x) = a(x)$ for all $x \in C$, and therefore $|(\alpha v_k \sqcap a)(x) - a(x)| = 0 < \epsilon$. 

For $x \not\in C$, we have $|a(x)| < \epsilon$. As $a(x) - \alpha v_{k'}(x) \leq a(x) < \epsilon$, we have $(a(x) - \alpha v_{k'}(x)) \sqcup 0 < \epsilon$. Therefore
\begin{align*}
|a(x) - (\alpha v_{k'} \sqcap a)(x)| &= a(x) - (\alpha v_{k'} \sqcap a)(x) = a(x) + (-\alpha v_{k'}(x)) \sqcup (-a(x)) \\
 &= (a(x) - \alpha v_{k'}(x)) \sqcup 0 < \epsilon.
\end{align*}
So we have shown that for all $\epsilon > 0$, there exists $k \in K$ such that for all $k' \in K$ with $k' \geq k$ and for all $x \in X$ (whether $x \in C$ or $x \not\in C$) $|(\alpha v_{k'} \sqcap a)(x) - a(x) | < \epsilon$, which is to say, $(\alpha v_k \sqcap a)_{k \in K}$ converges uniformly to $a$. 
\end{proof}

\begin{proposition}
\label{RieszAlgGeomProp}
$Z$ and $I$ form an isomorphism between the set of closed subsets of $X$ and the set of norm-closed ideals of $C(X)$, \emph{i.e.} if $Y \subseteq X$ is closed, then $Z(I(Y)) = Y$, and if $J \subseteq C(X)$ is a norm-closed ideal, then $I(Z(J)) = J$. 
\end{proposition}
\begin{proof}
Let $Y \subseteq X$ be a closed subset. If $x \in Y$ then for all $a \in I(Y)$ we have $a(x) = 0$, so $x \in Z(I(Y))$. If $x \not\in Y$, then by Urysohn's lemma, there exists a continuous function $a : X \rightarrow [0,1]$ such that $a(y) = 0$ for all $y \in Y$ and $a(x) = 1$. Therefore $a \in I(Y)$ but $a(x) \neq 0$, so $x \not\in Z(I(Y))$. 

Now let $J \subseteq C(X)$ be a norm-closed ideal. If $a \in J$, then for all $x \in Z(J)$, we have $a(x) = 0$, by definition, so $a \in I(Z(J))$. Conversely, if $a \in I(Z(J))$, by Lemma \ref{UniformConvApproxLemma}, we have a net $(v_k)_{k \in K}$ of elements of $J$ and a real $\alpha \in \mathbb{R}$ such that $\alpha v_k \sqcap a \to a$ uniformly. As $J$ is an ideal, $\alpha v_k \sqcap a \in J$, and as $J$ is uniformly closed, $a \in J$. 
\end{proof}

\begin{proposition}
\label{QuotientCXCompleteProp}
Let $X$ be a compact Hausdorff space, $J \subseteq C(X)$ a uniformly closed ideal. Let $Y = Z(J)$ and define $\phi : C(X) \rightarrow C(Y)$ by $\phi(a) = a|_Y$. This is a unital Riesz homomorphism. The map $\phi$ vanishes on $J$ and the induced map $\tilde{\phi} : C(X)/J \rightarrow C(Y)$ is a unital isomorphism. 
\end{proposition}
\begin{proof}
The map $\phi$ is a unital Riesz homomorphism because the Riesz operations on $C(X)$ and $C(Y)$ are defined pointwise in terms of the Riesz operations on $\mathbb{R}$. If $a \in J$, then for all $y \in Y = Z(J)$, we have $a(y) = 0$, so $\phi(a) = 0$. 

The map $\tilde{\phi}$ is a unital homomorphism, so we only need to prove that it is a bijection to prove that it is a unital isomorphism \cite[Definition 18.4]{Luxemburg}. We prove that it is surjective by proving that $\phi$ is surjective. If $b \in C(Y)$, then by Tietze's extension theorem \cite[IX.4.2 Theorem 2]{bourbaki} there exists $a \in C(X)$ such that $a|_Y = b$, \emph{i.e.} $\phi(a) = b$. 

To prove that it is injective, suppose that $\tilde{\phi}([a]) = \tilde{\phi}([a'])$ for $a,a' \in C(X)$. This means that $\phi(a-a') = 0$, which is to say, $a-a' \in I(Y) = I(Z(J)) = J$ because $J$ is norm-closed (Proposition \ref{RieszAlgGeomProp}). Therefore $[a] = [a']$. 
\end{proof}

\begin{corollary}
\label{QuotientCompleteCor}
Let $(A,u)$ be an Archimedean Riesz space with strong unit that is uniformly complete, and $J \subseteq A$ a closed ideal in $A$. Then $(A/J, [u])$ is Archimedean, $[u]$ a strong unit, and uniformly complete. 
\end{corollary}
\begin{proof}
The element $[u]$ is a strong unit by Lemma \ref{QuotientStrongUnitLemma}. By Yosida's theorem, there exists a unital Riesz isomorphism $\epsilon_A : (A,u) \cong (C(X),1)$ for $X$ a compact Hausdorff space. As it is an isomorphism, it preserves the norm defined by the unit, so it is an isometry. Therefore $J' = \epsilon_A(J)$ is not just an ideal, but a norm-closed ideal as well. The map $[\blank] \circ \epsilon_A : (A,u) \rightarrow (C(X)/J', [1])$ vanishes precisely on $J$, so $\widetilde{[\blank] \circ \epsilon_A} : (A/J,[u]) \rightarrow (C(X)/J',[1])$ is a unital Riesz isomorphism. By Proposition \ref{QuotientCXCompleteProp}, $C(X)/J'$ is uniformly complete, because it is isomorphic to $C(Z(J'))$, and therefore $A/J$ is uniformly complete. 
\end{proof}


\bibliography{biblio}

\begin{thebibliography}{KLMP13}

\bibitem[BdRV02]{BdRVModal}
Patrick Blackburn, Maarten de~Rijke, and Yde Venema.
\newblock {\em Modal Logic}.
\newblock Cambridge Tracts in Theoretical Computer Science, 2002.

\bibitem[BFKK08]{Brazdil2008}
Tom\'{a}\v{s}; Br\'{a}zdil, Vojtech Forejt, Jan Kret\'{\i}nsk\'{y}, and
  Anton\'{\i}n Kucera.
\newblock {The Satisfiability Problem for Probabilistic CTL}.
\newblock In {\em Proceedings of the 2008 23rd Annual IEEE Symposium on Logic
  in Computer Science}, pages 391--402, 2008.

\bibitem[BK08]{BaierKatoenBook}
Christel Baier and Joost~Pieter Katoen.
\newblock {\em Principles of {M}odel {C}hecking}.
\newblock The {MIT} Press, 2008.

\bibitem[Bou98]{bourbaki}
Nicolas Bourbaki.
\newblock {\em {General Topology}}.
\newblock Ettore Majorana International Science. Springer, 1998.

\bibitem[CES83]{CES83}
Edmund Clarke, E.~Allen Emerson, and Prasad~A. Sistla.
\newblock Automatic verification of finite-state concurrent systems using
  temporal logic specifications.
\newblock In {\em Proc. 10th ACM Symposium on Principles of Programming
  Languages}, 1983.

\bibitem[CF08]{modallogic2008book}
Nino~B. Cocchiarella and Max~A. Freund.
\newblock {\em Modal Logic: an introduction to its syntax and semantics}.
\newblock Oxford University Press, 2008.

\bibitem[Con90]{conway}
John~B. Conway.
\newblock {\em {A Course In Functional Analysis, Second Edition}}, volume~96 of
  {\em Graduate Texts in Mathematics}.
\newblock Springer Verlag, 1990.

\bibitem[dA03]{deAlfaro2003}
Luca de~Alfaro.
\newblock Quantitative verification and control via the mu-calculus.
\newblock In {\em Proc. of CONCUR}, 2003.

\bibitem[DGJP00]{prakash2000}
Jos{\'e}e Desharnais, Vineet Gupta, Radha Jagadeesan, and Prakash Panangaden.
\newblock {Approximating Labelled Markov Processes}.
\newblock In {\em Proc. of LICS}, pages 95--106, 2000.

\bibitem[DHL06]{DHL2006}
Christian Dax, Martin Hofmann, and Martin Lange.
\newblock A proof system for the linear time {$\mu$}-calculus.
\newblock In {\em Proc. of FSTTCS}, 2006.

\bibitem[dJvR77]{JVR1977}
E.~de~Jonge and A.C.M. van Rooij.
\newblock {\em Introduction to Riesz Spaces}, volume~78.
\newblock Mathematical Centre Tracts, Amsterdam, 1977.

\bibitem[Dou17]{doumanephd}
Amina Doumane.
\newblock {\em On the infinitary proof theory of logics with fixed points}.
\newblock PhD thesis, University Paris Diderot, 2017.

\bibitem[FJ14]{FurberJ14a}
Robert Furber and Bart Jacobs.
\newblock {From Kleisli Categories to Commutative C$^*$-algebras: Probabilistic
  Gelfand Duality}.
\newblock {\em Logical Methods in Computer Science}, 11(2), 2014.

\bibitem[FM09]{FM2009}
T.~Flaminio and F.~Montagna.
\newblock {MV-algebras with internal states and probabilistic fuzzy logics}.
\newblock {\em International Journal of Approximate Reasoning}, 1(50), 2009.

\bibitem[Gir80]{giry1981}
Mich{\`e}le Giry.
\newblock {A categorical approach to probability theory}.
\newblock In B.~Banaschewski, editor, {\em Categorical Aspects of Topology and
  Analysis}, pages 68--85. Springer, 1980.

\bibitem[Gir82]{giri1982}
M.~Giri.
\newblock A categorical approach to probability theory.
\newblock In {\em Categorical aspects of topology and analysis}, volume 915 of
  {\em Lecture Notes in Mathematics}. Springer, 1982.

\bibitem[GO07]{GorankoOtto2007}
Valentin Goranko and Martin Otto.
\newblock {\em Model Theory of Modal Logic}, chapter~5.
\newblock Handbook of Modal Logic. Elsevier, 2007.

\bibitem[HJ94]{HJ94}
Hansson Hans and Bengt Jonsson.
\newblock A logic for reasoning about time and reliability.
\newblock {\em Formal aspects of computing}, pages 512--535, 1994.

\bibitem[HK97]{HM96}
M.~Huth and M.~Kwiatkowska.
\newblock Quantitative analysis and model checking.
\newblock In {\em Proc. of LICS}, 1997.

\bibitem[Jac16]{Jacobs2016}
Bart Jacobs.
\newblock {\em Introduction to Coalgebra}.
\newblock Cambridge Univ. Press, 2016.

\bibitem[JT51a]{jonssontarski1951}
B.~J\'{o}nsson and A.~Tarski.
\newblock {Boolean algebras with operators, I}.
\newblock {\em Amer. J. Math}, 1(73):891--939, 1951.

\bibitem[JT51b]{jonssontarski1952}
B.~J\'{o}nsson and A.~Tarski.
\newblock {Boolean algebras with operators, II}.
\newblock {\em Amer. J. Math}, 1(74):127--162, 1951.

\bibitem[Kec94]{Kechris}
A.~S. Kechris.
\newblock {\em Classical Descriptive Set Theory}.
\newblock Springer Verlag, 1994.

\bibitem[Kei09]{Keimel2009}
Klaus Keimel.
\newblock {Abstract Ordered Compact Convex Sets and Algebras of the
  (Sub)Probabilistic Powerdomain Monad over Ordered Compact Spaces}.
\newblock {\em Algebra and Logic}, 48(5):330, 2009.

\bibitem[KKV04]{KKV2004}
Clemens Kupke, Alexander Kurz, and Yde Venema.
\newblock Stone coalgebras.
\newblock {\em Theoretical Computer Science}, 327:109--134, 2004.

\bibitem[KLMP13]{StonePrakash}
Dexter Kozen, Kim~G. Larsen, Radu Mardare, and Prakash Panangaden.
\newblock Stone duality for {M}arkov processes.
\newblock In {\em Proceeding of LICS}, 2013.

\bibitem[KMP13]{KMP2013}
Dexter Kozen, Radu Mardare, and Prakash Panangaden.
\newblock Strong completeness for {M}arkovian logics.
\newblock In {\em Proc. of MFCS}, 2013.

\bibitem[Koz81]{Kozen1981}
Dexter Kozen.
\newblock Semantics of probabilistic programs.
\newblock {\em J. Comput. Syst. Sci.}, 1981.

\bibitem[Koz83]{Kozen83}
Dexter Kozen.
\newblock Results on the propositional mu-calculus.
\newblock In {\em Theoretical Computer Science}, pages 333--354, 1983.

\bibitem[Koz85]{Kozen1983}
Dexter Kozen.
\newblock A probabilistic {PDL}.
\newblock {\em Journal of Computer and System Sciences}, 30(2):162--178, 1985.

\bibitem[Kro06]{kroupa2006}
T.~Kroupa.
\newblock Every state on semisimple {MV}-algebra is integral.
\newblock {\em Fuzzy Sets Systems}, 20(157), 2006.

\bibitem[Kur00]{KurzPHD}
Alexander Kurz.
\newblock {\em Logics for Coalgebras and Applications to Computer Science}.
\newblock PhD thesis, {L}udwig {M}aximilian {U}niversity of {M}unich, 2000.

\bibitem[Lax02]{LAX}
Peter Lax.
\newblock {\em Functional Analysis}.
\newblock Wiley Interscience, 2002.

\bibitem[LM19]{LM2019}
Christophe Lucas and Matteo Mio.
\newblock Towards a structural proof theory of probabilistic mu--calculi.
\newblock In {\em Proc. of FoSSaCS}, 2019.

\bibitem[LS82]{LS1982}
Daniel Lehmann and Saharon Shelah.
\newblock Reasoning with time and chance.
\newblock {\em Information and Control}, 53(3):165--1983, 1982.

\bibitem[LS91]{LS91}
K.~G. Larsen and A.~Skou.
\newblock Bisimulation through probabilistic testing.
\newblock {\em Information and Computation}, 94:1--28, 1991.

\bibitem[LvA07]{LvA2007}
C.~C.~A. Labuschagne and C.~J. van Alten.
\newblock {On the variety of Riesz spaces}.
\newblock {\em Indagationes Mathematicae}, 18(1), 2007.

\bibitem[LZ71]{Luxemburg}
W.~A.~J. Luxemburg and A.~C. Zaanen.
\newblock {\em Riesz Spaces}, volume~1.
\newblock North-{H}olland Mathematical Library, 1971.

\bibitem[MFM17]{MFM2017}
M.~Mio, R.~Furber, and R.~Mardare.
\newblock Riesz modal logic for {Markov} processes.
\newblock In {\em In Proceeding of LICS}, 2017.

\bibitem[Mio12a]{MioThesis}
Matteo Mio.
\newblock {\em Game Semantics for Probabilistic {$\mu$}-Calculi}.
\newblock PhD thesis, School of Informatics, University of Edinburgh, 2012.

\bibitem[Mio12b]{MIO2012b}
Matteo Mio.
\newblock {P}robabilistic {M}odal {$\mu$}-{C}alculus with {I}ndependent
  product.
\newblock {\em Logical Methods in Computer Science}, 8(4), 2012.

\bibitem[Mio14]{MIO2014a}
Matteo Mio.
\newblock Upper-expectation bisimilarity and {{\L}}ukasiewicz {$\mu$}-calculus.
\newblock In {\em Proc. of FoSSaCS}, 2014.

\bibitem[Mio18]{Mio18}
Matteo Mio.
\newblock Riesz modal logic with threshold operators.
\newblock In {\em Proc. of LICS}, 2018.

\bibitem[ML71]{maclane}
Saunders Mac~Lane.
\newblock {\em {Categories for the Working Mathematician}}.
\newblock Graduate Texts in Mathematics. Springer Verlag, 1971.

\bibitem[MM07]{MM07}
A.~McIver and C.~Morgan.
\newblock Results on the quantitative {$\mu$}-calculus q{M}{$\mu$}.
\newblock {\em ACM Transactions on Computational Logic}, 8(1), 2007.

\bibitem[MS17]{MioSimpsonFI2017}
Matteo Mio and Alex Simpson.
\newblock {\L}ukasiewicz mu-calculus.
\newblock In {\em Fundamenta Informaticae}, to appear 2017.

\bibitem[MT46]{bennett_1946}
J.~C.~C. McKinsey and Alfred Tarski.
\newblock On closed elements in closure algebras.
\newblock {\em Journal of Symbolic Logic}, 11(3):83--84, 1946.

\bibitem[Mun11]{MundiciBook}
D.~Mundici.
\newblock {\em Advanced {\L}ukasiewicz Calculus and MV-Algebras}.
\newblock Trends in Logic. Springer-Verlag, 2011.

\bibitem[NL11]{RMV2011}
Antonio~Di Nola and Ioana Leustean.
\newblock Riesz {MV}-algebras and their logic.
\newblock In {\em Proceeding of EUSFLAT}, 2011.

\bibitem[Pan09]{PrakashBook}
Prakash Panangaden.
\newblock {\em Labelled Markov processes}.
\newblock Imperial College Press, 2009.

\bibitem[Plo81]{SOS}
Gordon~D. Plotkin.
\newblock A structural approach to operational semantics.
\newblock {\em Journal of Logic and Algebraic Programming}, 60-61:17--139,
  1981.

\bibitem[Pnu77]{PNUELI77}
A.~Pnueli.
\newblock The temporal logic of programs.
\newblock In {\em Proc. of FOCS}, 1977.

\bibitem[Rey01]{reynolds2001}
M.~Reynolds.
\newblock An axiomatization of full computation tree logic.
\newblock {\em J. Symbolic Logic}, 66(3):1011--1057, 09 2001.

\bibitem[\'S74]{Swirszcz74}
Tadeusz \'Swirszcz.
\newblock {Monadic Functors and Convexity}.
\newblock {\em Bulletin de l'Acad\'{e}mie Polonaise des Sciences, S\'{e}rie des
  Sciences Math. Astr. et Phys.}, 22(1):39--42, 1974.

\bibitem[\'S75]{swirszcz75}
Tadeusz \'Swirszcz.
\newblock {Monadic Functors and Categories of Convex Sets}.
\newblock Institute of Mathematics of the Polish Academy of Sciences, Preprint
  70, 1975.

\bibitem[Sch66]{schaefer}
Helmut~H. Schaefer.
\newblock {\em {Topological Vector Spaces}}, volume~3 of {\em Graduate Texts in
  Mathematics}.
\newblock Springer Verlag, 1966.

\bibitem[Sok11]{Sokalova2011}
Ana Sokolova.
\newblock {P}robabilistic {S}ystems {C}oalgebraically: {A} survey.
\newblock {\em Theoretical Computer Science}, 412(38), 2011.

\bibitem[Sti01]{Stirling96}
C.~Stirling.
\newblock {\em Modal and temporal logics for processes}.
\newblock Springer, 2001.

\bibitem[Stu07]{Studer07}
Thomas Studer.
\newblock On the proof theory of the modal mu-calculus.
\newblock In {\em Studia Logica, Volume 89, Number 3}. Springer Netherlands,
  2007.

\bibitem[SV88]{SambinVaccaro1988}
Giovanni Sambin and Virginia Vaccaro.
\newblock Topology and duality in modal logic.
\newblock {\em Annals of Pure and Applied Logic}, 1(37):249--296, 1988.

\bibitem[vB84]{JVB1984}
Johan van Benthem.
\newblock {\em Handbook of Philosophical Logic}, chapter Correspondence Theory,
  pages 167--247.
\newblock Springer Netherlands, 1984.

\bibitem[Vul67]{vulikh}
B.~Z. Vulikh.
\newblock {\em Introduction to the Theory of Partially Ordered Spaces}.
\newblock Wolters-Noordhoff Scientific Publications LTD. Groningen, 1967.

\bibitem[Wal95]{Walukiewicz-CompletenessofKozen}
Igor Walukiewicz.
\newblock Completeness of {K}ozen's axiomatisation of the propositional
  mu-calculus.
\newblock In {\em Proc. of LICS}, pages 14--24, 1995.

\bibitem[Wan96]{prooftheorymodallogic}
Heinrich Wansing, editor.
\newblock {\em Proof Theory of Modal Logic}, volume~2 of {\em Applied Logic
  Series}.
\newblock Springer Netherlands, 1996.

\bibitem[Wes16]{westerbaan2016}
Bas Westerbaan.
\newblock {Yosida Duality}.
\newblock \url{https://arxiv.org/abs/1612.03327}, 2016.

\bibitem[Yos41]{yosida}
K\^osaku Yosida.
\newblock {On Vector Lattice with a Unit}.
\newblock {\em Proc. Imp. Acad. Tokyo}, 17:121--124, 1940-1941.

\end{thebibliography}
\bibliographystyle{alpha}

%




\end{document}